\def\thetitle{
  Well-tempered ZX and ZH calculi
}
\def\theauthor{
  Niel de Beaudrap
}
\def\theinstitute{%
  \hspace*{-0.5em} Department of Computer Science \hspace*{-0.5em}\\
  University of Oxford, 
  Oxford, UK%
}
\def\theemail{%
  \hspace*{-0.5em} niel.debeaudrap@cs.ox.ac.uk \hspace*{-0.5em}%
}
\title{\thetitle}
\author{\theauthor\institute\theinstitute\email\theemail}
\def\titlerunning{\thetitle}
\def\authorrunning{\theauthor}

\documentclass[submission,copyright,creativecommons]{eptcs}
\providecommand{\event}{QPL 2020}
\usepackage[numbers,compress]{natbib}
\usepackage{amsmath}
\usepackage{amsthm}
\usepackage{amssymb}
\usepackage{bm,mathtools}
\usepackage[utf8x]{inputenc}
\usepackage{enumitem,multicol,multirow}
\usepackage[font={small}]{caption}
\usepackage[mathscr]{euscript}

\usepackage{tikz}
\usetikzlibrary{shapes,arrows,calc,positioning,fit}
\pgfdeclarelayer{background}
\pgfsetlayers{background,main}

\tikzstyle{every picture}=[baseline=-0.25em]

\newcommand\tikzfig[1]{%
\IfFileExists{./figures/#1.pdf}%
  {\raisebox{-0.25em}{$\begin{gathered}\includegraphics{./figures/#1.pdf}\end{gathered}$}}%
  {%
    \IfFileExists{./Figures-Source/#1.tikz}%
      {\input{./Figures-Source/#1.tikz}}%
      {%
	\IfFileExists{./#1.tikz}%
      	  {\input{./#1.tikz}}%
      	  {\tikz[baseline=-0.5em]{%
		\node[draw=red!50!black,font=\color{red!50!black},fill=red!10!white]
		{\textit{#1}};}}%
      }%
  }%
}

\newcommand\tikzfigsrconly[1]{%
    \IfFileExists{./Figures-Source/#1.tikz}%
      {\input{./Figures-Source/#1.tikz}}%
      {%
	\IfFileExists{./#1.tikz}%
      	  {\input{./#1.tikz}}%
      	  {\tikz[baseline=-0.5em]{%
		\node[draw=red!50!black,font=\color{red!50!black},fill=red!10!white]
		{\textit{#1}};}}%
      }%
}

\tikzstyle{Z dot}=
  [fill={rgb,255:red,216;green,248;blue,216}, inner sep=0mm, minimum size=2mm, draw=black, inner sep=2pt, shape=circle]
\tikzstyle{X dot}=
  [Z dot, fill={rgb,255:red,227;green,145;blue,145}]
\tikzstyle{white dot}=
  [Z dot, fill={rgb,255:red,255;green,255;blue,255}]
\tikzstyle{not dot}=
  [Z dot, fill={rgb,255:red,108;green,108;blue,108}]
\tikzstyle{gray dot}=
  [Z dot, fill={rgb,255:red,208;green,208;blue,208}]
\tikzstyle{Z phase dot}=
  [fill={rgb,255:red,216;green,248;blue,216},  minimum size=5mm, font={\footnotesize}, rounded corners=2mm, inner sep=0.2mm, outer sep=-2mm, scale=0.8,draw=black]
\tikzstyle{X phase dot}=
  [Z phase dot, fill={rgb,255:red,227;green,145;blue,145}]
\tikzstyle{H box}=
  [fill=yellow!50!white, draw=black, shape=rectangle, inner sep=0.6mm, minimum height=3mm, minimum width=3mm, font=\footnotesize]
\tikzstyle{small H box}=
  [fill=yellow!50!white, draw=black, shape=rectangle, inner sep=0.6mm, minimum width=2mm, minimum height=2mm, font=\footnotesize]
\def\nubox at (#1) #2;{%
  \node at (#1) [small H box, label={[label distance=-7.5pt, anchor=180]north east:$\phantom{\vert}_{#2}$}] {$\nu$};
}
\tikzstyle{nu box}=
  [fill=yellow!50!white, draw=black, shape=rectangle, inner sep=0.6mm, minimum width=2mm, minimum height=2mm, font=\footnotesize]
\tikzstyle{mini H box}=
  [fill=yellow!50!white, draw=black, shape=rectangle, inner sep=0.6mm, minimum height=1mm, minimum width=1mm, font=\footnotesize]
\tikzstyle{H dot}=
  [Z dot, fill=yellow!50!white, inner sep=1pt]
\tikzstyle{HTriangle}=
  [fill={rgb,255:red,248;green,248;blue,0}, inner sep=0mm, minimum size=3.5mm, draw=black, regular polygon, regular polygon sides=3, shape border rotate=90]
\tikzstyle{Triangle}=
  [fill={rgb,255:red,248;green,248;blue,248}, inner sep=0mm, minimum size=3.5mm, draw=black, regular polygon, regular polygon sides=3, shape border rotate=90]
\tikzstyle{InvTriangle}=
  [Triangle, fill={rgb,255:red,96;green,96;blue,96}]
\tikzstyle{RevTriangle}=
  [Triangle, shape border rotate=-90]
\tikzstyle{InvRevTriangle}=
  [InvTriangle, shape border rotate=-90]
\tikzstyle{UpTriangle}=
  [Triangle, shape border rotate=0]
\tikzstyle{InvUpTriangle}=
  [InvTriangle, shape border rotate=0]
\tikzstyle{DownTriangle}=
  [Triangle, shape border rotate=180]
\tikzstyle{InvDownTriangle}=
  [InvTriangle, shape border rotate=180]

\tikzstyle{Z exp dot}=[Z dot, double, minimum size=2.5mm]
\tikzstyle{X exp dot}=[X dot, double, minimum size=2.5mm]
\tikzstyle{blank exp dot}=[blank dot, double, minimum size=2.5mm]
\tikzstyle{Z phase exp dot}=[Z phase dot, double, minimum size=2.5mm]
\tikzstyle{X phase exp dot}=[X phase dot, double, minimum size=2.5mm]
\tikzstyle{H exp box}=[H box, double, minimum size=3mm]

\tikzstyle{Z exp box}=
  [Z dot, rectangle, minimum width=3mm, minimum height=3mm]
\tikzstyle{X exp box}=
  [X dot, rectangle, minimum width=3mm, minimum height=3mm]
\tikzstyle{blank exp box}=
  [blank dot, rectangle, minimum width=3mm, minimum height=3mm]
\tikzstyle{Z phase exp box}=
  [Z phase dot, rectangle, minimum width=3mm, minimum height=3mm]
\tikzstyle{X phase exp box}=
  [X phase dot, rectangle, minimum width=3mm, minimum height=3mm]

\tikzstyle{ug}=[regular polygon, regular polygon sides=3, fill={zx_red}, draw=black, inner sep=0pt, minimum width=1em, tikzit draw=blue]
\tikzstyle{st}=[star, star points=5, fill=white, draw=black, inner sep=1.2pt, line width=1.2pt, tikzit fill=blue, tikzit draw=red, tikzit category=ZH-pf]
\tikzstyle{triangle}=[regular polygon, regular polygon sides=3, fill=white, draw=black, inner sep=0pt, minimum width=1em, tikzit draw=blue, tikzit category=ZH-pf]
\tikzstyle{not}=[fill={rgb,255: red,170; green,170; blue,170}, draw=black, shape=circle, dot, minimum width=2.5mm, label={center:\tiny$\bm\neg$}]
\tikzstyle{bbindex}=[font={\color{blue}\footnotesize}]

\tikzstyle{diredge}=[->]
\tikzstyle{bbox edge}=[-, draw={rgb,255: red,42; green,145; blue,255}]

\newcommand\ie{\emph{i.e.}}
\newcommand\eg{\emph{e.g.}}

\makeatletter
\def\Rule{%
  \@ifstar\@@Rule\@Rule
}
\def\@Rule(#1){%
  \mbox{\textsf{\upshape (#1)}}%
}
\def\@@Rule(#1){%
  \mbox{\textsf{\upshape #1}}%
}
\makeatother

\makeatletter
\let\orig@maketitle\maketitle
\def\@titlesplit{%
  \\[-0.25ex]
} 
\if@twocolumn
\let\titlesplit\@titlesplit
\else
\let\titlesplit\relax
\fi
\makeatother

\newcommand\sem[1]{%
  {[\mspace{-2.75mu}[ \smash{#1} ]\mspace{-2.75mu}]}%
}
\newcommand\bigsem[1]{%
  {\bigl[\mspace{-4.5mu}\bigl[ \smash{#1} \bigr]\mspace{-4.5mu}\bigr]}%
}
\newcommand\Bigsem[1]{%
  {\Bigl[\mspace{-5.5mu}\Bigl[\, \smash{#1} \,\Bigr]\mspace{-5.5mu}\Bigr]}%
}
\newcommand\biggsem[1]{%
  {\biggl[\mspace{-6.5mu}\biggl[\, \smash{#1} \,\biggr]\mspace{-6.5mu}\biggr]}%
}
\newcommand\Biggsem[1]{%
  {\Biggl[\mspace{-7.5mu}\Biggl[\, \smash{#1} \,\Biggr]\mspace{-7.5mu}\Biggr]}%
}
\newcommand\Sem[2]{%
  {\left[\mspace{-8.25mu}\left[
    \smash{#2} \mathclap{\begin{matrix} \\[#1] \end{matrix}}
  \right]\mspace{-8mu}\right]}%
}

\newcommand\parStyle[1]{\textrm{\mdseries\upshape({#1}\kern0.1ex)}}

\newlength\romanumlabelwd
\newif\ifenumenv
\enumenvfalse

\makeatletter
  \renewenvironment{abstract}{%
      \if@twocolumn
        \subsection*{\abstractname}%
        \small
      \else
        \small
        \begin{center}%
          {\bfseries \abstractname\vspace{-.5em}\vspace{\z@}}%
        \end{center}%
        \quotation
      \fi}
      {\if@twocolumn\par\else\endquotation\fi}

\renewenvironment{thebibliography}[1]
     {\vspace{-1ex}\section*{\refname}\small%
      \@mkboth{\MakeUppercase\refname}{\MakeUppercase\refname}%
      \list{\@biblabel{\@arabic\c@enumiv}}%
           {\settowidth\labelwidth{\@biblabel{#1}}%
            \leftmargin\labelwidth
            \advance\leftmargin\labelsep
            \@openbib@code
            \usecounter{enumiv}%
            \let\p@enumiv\@empty
            \renewcommand\theenumiv{\@arabic\c@enumiv}}%
      \sloppy
      \clubpenalty4000
      \@clubpenalty \clubpenalty
      \widowpenalty4000%
      \sfcode`\.\@m}
     {\def\@noitemerr
       {\@latex@warning{Empty `thebibliography' environment}}%
      \endlist}
\makeatother

\makeatletter
\def\tagform@#1{%
	\ifmmode
		\mbox{\normalsize\maketag@@@{(\ignorespaces#1\unskip\@@italiccorr)}}%
	\else
		\maketag@@@{(\ignorespaces#1\unskip\@@italiccorr)}%
	\fi}

\def\maketag@@@#1{\hbox{\normalsize\m@th\normalfont#1}}	

\def\mylabel#1{\@bsphack
  \protected@write\@auxout{}%
         {\string\newlabel{#1}{{\@currentlabel}{\thepage}}}%
  \@esphack}
\makeatother

\newtheoremstyle{theorem?}
  {\topsep}{\topsep}   																	
  {\itshape}{0pt}{\bfseries}{?}{5pt plus 1pt minus 1pt} 
  {}          																					
\theoremstyle{theorem?}

\makeatletter
\def\@cludgescbf#1#2\end{\textbf{#1\scalefont{0.85}#2}}
\newcommand\cludgescbf[1]{\@cludgescbf#1\end}
\makeatother

\newtheoremstyle{Axiom}
  {\topsep}{\topsep}   																	
  {\itshape}{0pt}{\bfseries}{?}{5pt plus 1pt minus 1pt} 
  {}          																					
\theoremstyle{Axiom}

\theoremstyle{definition}
\newtheorem{definition}{Definition}

\newtheorem*{definition*}{Definition}

\theoremstyle{plain}
\newtheorem{theorem}{Theorem}
\newtheorem{lemma}[theorem]{Lemma}
\newtheorem{proposition}[theorem]{Proposition}
\newtheorem*{proposition*}{Proposition}

\newtheorem*{claim*}{Claim}
\newtheorem{corollary}{Corollary}[theorem]

\newcommand\C{\mathbb{C}}
\newcommand\R{\mathbb{R}}
\newcommand\Z{\mathbb{Z}}

\newcommand\N{\mathbb{N}}

\newcommand\e{\mathrm{e}}

\renewcommand\vec\mathbf


\let\oldepsilon\epsilon
\let\epsilon\varepsilon
\let\varepsilon\oldepsilon

\let\ge\geqslant

\DeclareRobustCommand\idop{%
  \text{\usefont{U}{bbold}{m}{n}1}%
}

\makeatletter
	\newcommand\ket[1]{\left| #1 \right\rangle\@ifnextchar\bra{\mspace{-4mu}}{}}
	\newcommand\bra[1]{\left\langle #1 \right|}

\makeatother

\makeatletter
\def\@varstartsection#1#2#3#4#5#6{%
  \if@noskipsec \leavevmode \fi
  \par
  \@tempskipa #4\relax
  \@afterindenttrue
  \ifdim \@tempskipa <\z@
    \@tempskipa -\@tempskipa \@afterindentfalse
  \fi
  \if@nobreak
    \everypar{}%
  \else
    \addpenalty\@secpenalty\addvspace\@tempskipa
  \fi
  \@ifstar
    {\@ssect{#3}{#4}{#5}{#6}}%
    {\@dblarg{\@varsect{#1}{#2}{#3}{#4}{#5}{#6}}}}
\def\@varsect#1#2#3#4#5#6[#7]#8{%
  \ifnum #2>\c@secnumdepth
    \let\@svsec\@empty
  \else
    \refstepcounter{#1}%
    \protected@edef\@svsec{(\alph{#1}) ~\relax}%
  \fi
  \@tempskipa #5\relax
  \ifdim \@tempskipa>\z@
    \begingroup
      #6{%
        \@hangfrom{\hskip #3\relax\@svsec}%
          \interlinepenalty \@M #8\@@par}%
    \endgroup
    \csname #1mark\endcsname{#7}%
    \addcontentsline{toc}{#1}{%
      \ifnum #2>\c@secnumdepth \else
        \protect\numberline{\csname the#1\endcsname}%
      \fi
      #7}%
  \else
    \def\@svsechd{%
      #6{\hskip #3\relax
      \@svsec #8}%
      \csname #1mark\endcsname{#7}%
      \addcontentsline{toc}{#1}{%
        \ifnum #2>\c@secnumdepth \else
          \protect\numberline{\csname the#1\endcsname}%
        \fi
        #7}}%
  \fi
  \@xsect{#5}}
\renewcommand\subsubsection{\@varstartsection{subsubsection}{3}{\z@}%
                                     {-3ex\@plus -1ex \@minus -.2ex}%
                                     {1.5ex \@plus .2ex}%
                                     {\centering\normalfont\small\itshape\bfseries}}
\def\ref#1{\expandafter\@setref\csname r@#1\endcsname\@firstoftwo{#1}}

\makeatother

\newcommand\psymb{\texttt{\upshape\raisebox{-0.375ex}\large+}}
\newcommand\msymb{\texttt{\upshape\raisebox{0.1875ex}\large-\hspace*{-0.4125em}-}}
\newcommand\zsymb{\texttt{\upshape 0}}
\newcommand\osymb{\texttt{\upshape 1}}

\begin{document}
\maketitle

\begin{abstract}
  \noindent
  The ZX calculus is a mathematical tool to represent and analyse quantum operations by manipulating diagrams which in effect represent tensor networks.
  Two families of nodes of these networks are ones which commute with either $Z$ rotations or $X$ rotations, usually called ``green nodes'' and ``red nodes'' respectively.
  The original formulation of the ZX calculus was motivated in part by properties of the algebras formed by the green and red nodes: notably, that they form a bialgebra --- but only up to scalar factors.
  As a consequence, the diagram transformations and notation for certain unitary operations involve ``scalar gadgets'' which denote contributions to a normalising factor.
  We present renormalised generators for the ZX calculus, which form a bialgebra precisely.
  As a result, no scalar gadgets are required to represent the most common unitary transformations, and the corresponding diagram transformations are generally simpler.
  We also present a similar renormalised version of the ZH calculus.
  We obtain these results by an analysis of conditions under which various ``idealised'' rewrites are sound, leveraging the existing presentations of the ZX and ZH calculi.
\end{abstract}


\section{Introduction}
\vspace*{-.5ex}

The ZX calculus~\cite{CD-2011,DP-2009,Backens-2015,PW-2016,BPW-2017,JPVW-2017,CK-2017,JPV-2017,NW-2017,NW-2018,Vilmart-2019-minimal,Vilmart-2019-ToffoliH,JPV-2019,Wang-2019,Wang-2019-semirings} is a mathematical tool to reason about quantum computation using diagrams.
It uses a graph-based notation broadly similar to quantum circuit notation to represent transformations on one or more qubits, augmented with rules to rewrite diagrams in order to perform computations without recourse to matrices.
While the analysis of complicated procedures may require diagrams of mounting complexity, the ZX calculus makes it easy in many cases to quickly analyse many-qubit procedures.
As well as being potentially less resource-intensive than matrix-based computations, the ZX calculus is more versatile than circuit diagrams.
Specifically, the objects represented by ZX diagrams are in general tensor networks, in which any arrow of time is merely imposed by convention.
As a result of this flexibility, it has found to be productive in application to quantum technologies, for error correction~\cite{BH-2020,GF-2018,GF-2019} and circuit optimisation problems such as reduction of phase gates and CNOT allocation~\cite{RKPW-2019,KW-2019,KM-G-2019,CDDSS-2019,BBW-2019}.

Modern treatments of the ZX calculus~\cite{Backens-2015,PW-2016,BPW-2017,JPVW-2017,CK-2017,JPV-2017,NW-2017,NW-2018,Vilmart-2019-minimal,Vilmart-2019-ToffoliH,JPV-2019,Wang-2019,Wang-2019-semirings} are ``scalar exact'', in that they allow one to infer not only whether $T_1 \propto T_2$ for two transformations $T_1, T_2$, but also whether $T_1 = \lambda T_2$ for a particular scalar $\lambda \in \C$ (as when $T_1$ realises $T_2$ with some probability of success $\lvert \lambda \rvert^2$).
This distinction is important in principle in describing the effect of teleportation~\cite{BBCJPW-1993}, measurement-based computations~\cite{RBB-2003,DKP-2007,Duncan-2013}, surface-code lattice surgery~\cite{CFDvM-2012,BH-2020} in different computational branches; and to distinguish what can be achieved efficiently, as opposed to what can be achieved through postselection~\cite{Aaronson-2005,Beaudrap-2015}.
However, to do so, the user must keep track of ``scalar gadgets'' --- small sub-diagrams which obliquely denote specific scalar factors, which are necessary to represent certain unitary transformations exactly, and which change frequently with rewrites.
This is a consequence of the original way in which the ZX calculus was formulated, in which the principal generators of the ZX calculus form a bialgebra only up to scalar factor corrections.

A related diagrammatic calculus to the ZX calculus is the ZH calculus~\cite{BK-2019,WW-2019,KWK-2019}, which was motivated in part by unitary circuits over the Hadamard+Toffoli gate set~\cite{Shi-2003,Aharonov-2003}.
The ZH~calculus is equivalent in expressive power to the ZX~calculus, and in particular is also scalar exact.
However, while it has a superior facility for expressing these scalar factors, some of its most important rewrites involves accounting for these scalar factors, as does representing a single-qubit Hadamard gates in quantum circuits.

In this article, we present versions of the~ZX and~ZH calculi whose 
rewrite rules involve fewer scalar gadgets, and which can represent basic unitary transformations more simply.
This is achieved through a change in normalisation of the generators.
We expect that these ``well-tempered'' versions will reduce the work required to perform scalar-exact computation with these calculi, by simplifying the rules which are most commonly used in practise.
To summarise our results, we present these calculi immediately below:

\vspace*{-2ex}
\paragraph*{A ``well-tempered'' ZX calculus ---\!\!\!}
%
We present a version of the ZX calculus, with generator nodes
\vspace*{-1.5ex}
    \begin{equation}
    \begin{gathered}
      \input{./Figures-Source/ZX-green-phase-dot-arity.tikz}
    \;,\quad
      \input{./Figures-Source/ZX-red-phase-dot-arity.tikz}
    \;,\quad
      \input{./Figures-Source/ZX-H-box.tikz}
    \;\;,\quad
      \input{./Figures-Source/ZX-nu-box-k.tikz}
    \;\;,
    \end{gathered}
    \end{equation}~\\[-3.5ex]
where $m,n \in \N$.%
  \footnote{%
    Diagrams in this article are read with ``inputs'' on the left and ``outputs'' on the right (similarly to quantum circuit diagrams).
  }
We call these ``green dots'' (or ``Z dots''), ``red dots'' (or ``X dots''), ``Hadamard boxes'', and ``nu~boxes''.
We may represent Hadamard boxes by small unlabelled degree-2 boxes for the sake of brevity.
The indicated generators take parameters $\theta \in \R$ (which may be omitted if $\theta = 0$) or $k \in \R$ (which may be omitted if $k = 1$; the parameter in this case is always to be written at the upper-right corner). 
\label{newZX}%
We define a calculus on these generators with the following axioms:%
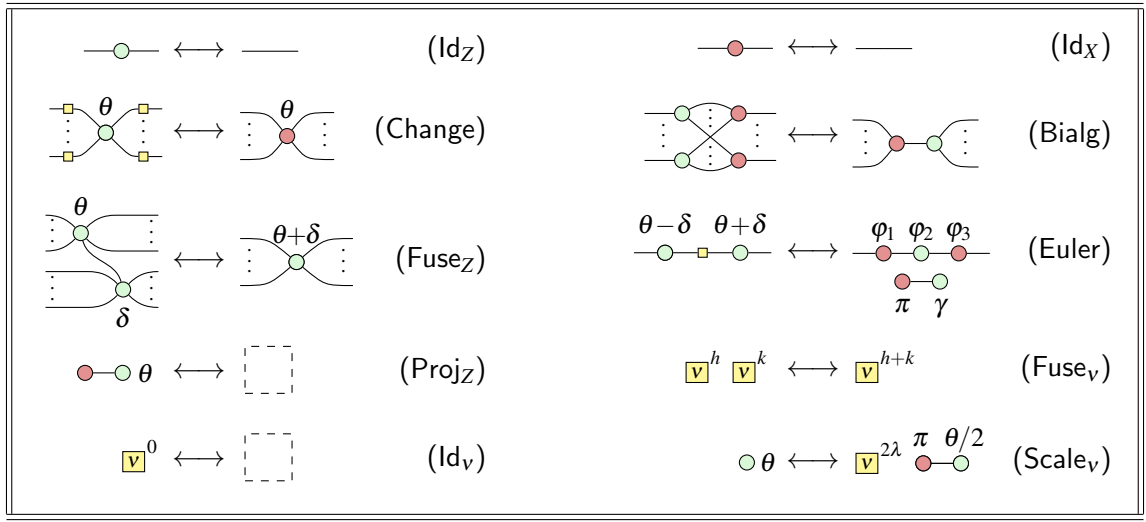
\vspace*{-1ex}
\begin{gather*}
\begin{array}{||c@{\qquad\qquad}c||}
\hline\hline
  \begin{array}{r@{\,{}\longleftrightarrow{}\,}l@{\;\;\;}r}
    \input{./Figures-Source/ZX-green-id.tikz} \,&\, \input{./Figures-Source/id-wire.tikz} &
    \mathclap{\begin{matrix} \\[3ex] \end{matrix}}
    \Rule(Id$_Z$)
  \\
    \input{./Figures-Source/ZX-green-phase-w-H.tikz} & \input{./Figures-Source/ZX-red-phase-dot.tikz} &
    \mathclap{\begin{matrix} \\[4ex] \end{matrix}}
    \Rule(Change)
  \\
    \input{./Figures-Source/ZX-green-phase-fuse.tikz} & \input{./Figures-Source/ZX-green-phase-sum.tikz} &
    \mathclap{\begin{matrix} \\[4ex] \end{matrix}}
    \Rule(Fuse$_Z$)
  \\[-1ex]
    \input{./Figures-Source/ZX-simpler-1.tikz} &\; \input{./Figures-Source/empty.tikz} &
    \mathclap{\begin{matrix} \\[3ex] \end{matrix}}
    \Rule(Proj$_Z$)
  \\
    \input{./Figures-Source/ZX-nu-box-0.tikz} &\; \input{./Figures-Source/empty.tikz} &
    \mathclap{\begin{matrix} \\[3ex] \end{matrix}}
    \Rule(Id$_\nu$)
  \end{array}
  &
  \begin{array}{r@{\,{}\longleftrightarrow{}\,}l@{\;\;\;}r}
    \input{./Figures-Source/ZX-red-id.tikz} \,&\, \input{./Figures-Source/id-wire.tikz} &
    \mathclap{\begin{matrix} \\[3ex] \end{matrix}}
    \Rule(Id$_X$)
  \\
    \input{./Figures-Source/ZX-bialg-many.tikz} & \input{./Figures-Source/ZX-bott-many.tikz} &
    \mathclap{\begin{matrix} \\[5ex] \end{matrix}}
    \Rule(Bialg)
  \\[-.5ex]
    \input{./Figures-Source/ZX-Euler-left.tikz} & \input{./Figures-Source/ZX-Euler-right-new.tikz} &
    \mathclap{\begin{matrix} \\[4ex] \end{matrix}}
    \Rule(Euler)
  \\[-.5ex]
    \input{./Figures-Source/ZX-nu-box-fuse.tikz} &\; \input{./Figures-Source/ZX-nu-box-sum.tikz} &
    \mathclap{\begin{matrix} \\[3ex] \end{matrix}}
    \Rule(Fuse$_\nu$)
  \\
    \input{./Figures-Source/ZX-green-phase-dot-arity-0.tikz}&\;
    \input{./Figures-Source/ZX-green-dot-nu-gadget.tikz}&
    \Rule(Scale$_\nu$)
  \end{array}
\\
\hline\hline
\end{array}
\end{gather*}~\\[-2ex]
These axioms largely follow the ``near minimal axiomatization'' of Vilmart~\cite[Fig.\;2]{Vilmart-2019-minimal}, but with fewer scalar gadgets, and with three additional rules to define the behaviour of the nu boxes.
The principal difference to the usual presentation of the ZX calculus is that the green~dots and the red~dots form a bialgebra, and not a \emph{scaled} bialgebra.
The right-hand sides of the rewrites \Rule(Proj$_Z$) and \Rule(Id$_\nu$) are the empty diagram, the parameters $h,k \in \R$ and angles $\theta, \delta \in \R$ may be arbitrary, the angles $\varphi_1, \varphi_2, \varphi_3 , \gamma \in \R$ in rule \Rule(Euler) are given by
\vspace*{-1ex}
\begin{subequations}%
\allowdisplaybreaks
\label{eqn:VilmartEulerAngleComputation}%
\begin{gather}
    \varphi_1
  \,=\,
    \arg\bigl(
      z_1
    \bigr)
    +
    \arg\bigl(
      z_2
    \bigr)
    +
    \tfrac{\pi}{2}
  ,
\qquad\qquad
    \varphi_2
  \,=\,
    \displaystyle
      2 \arg\bigl(
        z_3
      \bigr)
    , 
\qquad\qquad
      \varphi_3
    \,=\,
      \arg\bigl(
        z_1
      \bigr)
      -
      \arg\bigl(
        z_2
      \bigr)
      +
      \tfrac{\pi}{2}
    ,
\\[0.5ex]
    \gamma
  \,=\,
    \theta
    -
    \arg\bigl(
      z_1
    \bigr)
    -
      \arg(z_3)
    ,
\end{gather}
\end{subequations}~\\[-3.5ex]
where 
$z_1 = \cos(\delta) + i \sin(\theta)$,
$z_2 = \cos(\theta) + i \sin(\delta)$, and
$z_3 = \lvert z_1 \rvert + i \lvert z_2 \rvert$; and the parameter $\lambda$ on the right-hand side of \Rule(Scale$_\nu$) is given by $\lambda = \log_2\bigl(\sec^2(\theta\!\!\:/\!\!\;2)\bigr) - 1$ for $\theta$ not an odd multiple of $\pi$.

\vspace*{-1ex}
\paragraph*{A ``well-tempered'' ZH calculus ---\!\!\!}
%
We also present a version of the ZH calculus, with generator nodes
\vspace*{-2.5ex}
\begin{equation}{}
  \mspace{-64mu}
  \begin{gathered}{}
      \input{./Figures-Source/ZH-white-dot-arity.tikz}
    ,\;
      \input{./Figures-Source/ZH-H-phase-box-arity.tikz}
    ,\;
      \input{./Figures-Source/ZH-gray-dot-arity.tikz}
    ,\;\;
      \input{./Figures-Source/ZH-not-dot.tikz}
    \;,
  \end{gathered}
  \mspace{-18mu}
\end{equation}~\\[-3.5ex]
where $m,n \in \N$.
We call these ``white dots'', ``H-boxes'', ``gray dots'', and ``not dots''.
The latter two correspond to gadgets (``derived generators'') in the original presentation of Ref.~\cite{BK-2019}; we elevate them to the status of generators for the sake of our analysis.
H-boxes take parameters $a \in \C$ (which may be omitted for $a = -1$).%
  \footnote{%
    Not-dots may be labeled with a ``$\neg$'' symbol, as in Ref.~\cite{BK-2019}; we use a dark gray node instead to simplify our diagrams.
  }
\label{newZH}%
We define a calculus on these generators with the following axioms:%
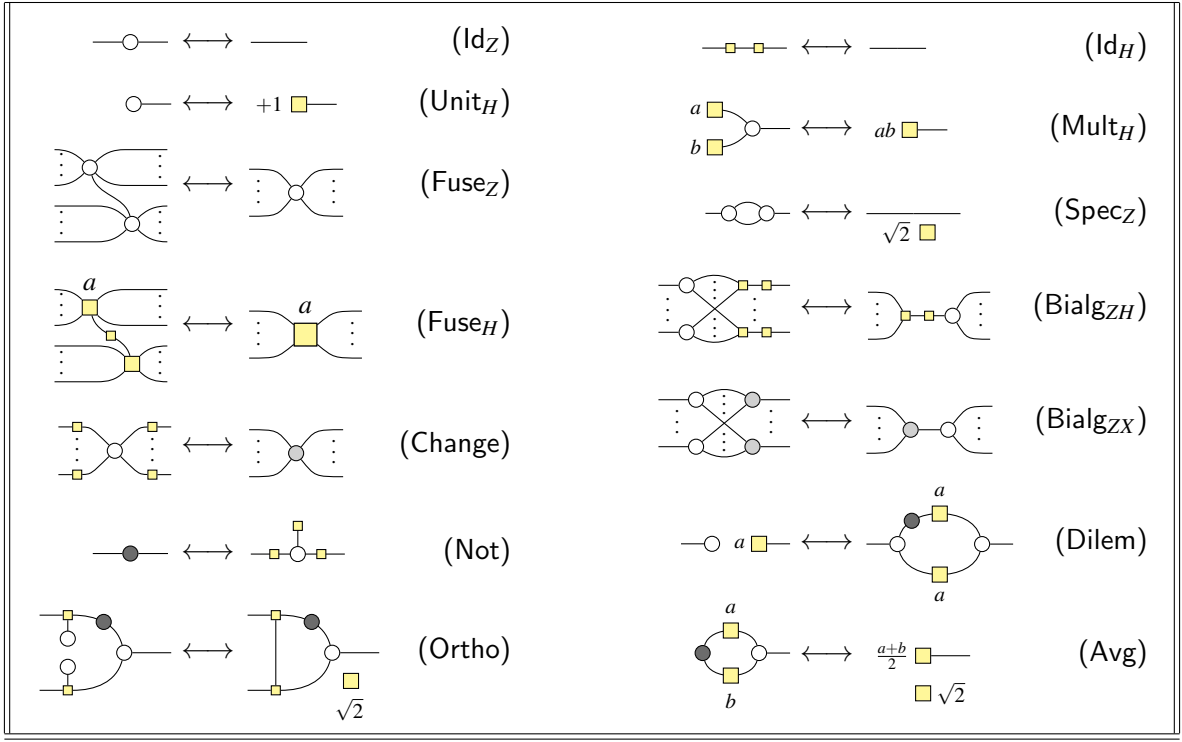
\vspace*{-1.5ex}
\begin{gather*}
\begin{array}{||c@{\qquad\qquad}c||}
\hline\hline
  \begin{array}{r@{\,{}\longleftrightarrow{}\,}l@{\;\;\;}r}
    \input{./Figures-Source/ZH-white-id.tikz} \,&\, \input{./Figures-Source/id-wire.tikz} &
    \mathclap{\begin{matrix} \\[3ex] \end{matrix}}
    \Rule(Id$_Z$)
  \\[-1ex]
    \input{./Figures-Source/ZH-white-prep.tikz} & \input{./Figures-Source/ZH-H-plus1-prep.tikz} &
    \mathclap{\begin{matrix} \\[3ex] \end{matrix}}
    \Rule(Unit$_H$)
  \\[-1ex]
    \input{./Figures-Source/ZH-white-fuse.tikz} & \input{./Figures-Source/ZH-white-dot.tikz} &
    \mathclap{\begin{matrix} \\[6ex] \end{matrix}}
    \Rule(Fuse$_Z$)
  \\[1.75ex]
    \input{./Figures-Source/ZH-H-fuse.tikz} & \input{./Figures-Source/ZH-H-phase-box.tikz} &
    \mathclap{\begin{matrix} \\[8ex] \end{matrix}}
    \Rule(Fuse$_H$)
  \\[-.5ex]
    \input{./Figures-Source/ZH-white-w-H.tikz} & \input{./Figures-Source/ZH-gray-dot.tikz} &
    \mathclap{\begin{matrix} \\[6ex] \end{matrix}}
    \Rule(Change)
  \\[1.5ex]
    \input{./Figures-Source/ZH-not-dot.tikz} \,&\, \input{./Figures-Source/ZH-not-gadget.tikz} &
    \mathclap{\begin{matrix} \\[5ex] \end{matrix}}
    \Rule(Not)
  \\[-2ex]
    \input{./Figures-Source/ZH-ortho-nobridge.tikz} & \input{./Figures-Source/ZH-ortho-bridge-w-sqrt2.tikz} &
    \mathclap{\begin{matrix} \\[5ex] \end{matrix}}
    \Rule(Ortho)
  \\[-.5ex]
  \end{array}
  &
  \begin{array}{r@{\,{}\longleftrightarrow{}\,}l@{\;\;\;}r}
    \input{./Figures-Source/H-id.tikz} \,&\, \input{./Figures-Source/id-wire.tikz} &
    \mathclap{\begin{matrix} \\[3ex] \end{matrix}}
    \Rule(Id$_H$)
  \\[1.5ex]
    \input{./Figures-Source/ZH-mult.tikz} & \input{./Figures-Source/ZH-H-prod-prep.tikz} &
    \mathclap{\begin{matrix} \\[4ex] \end{matrix}}
    \Rule(Mult$_H$)
  \\[-0.5ex]
    \input{./Figures-Source/ZH-white-special.tikz} & \input{./Figures-Source/ZH-id-wire-w-sqrt2.tikz} &
    \mathclap{\begin{matrix} \\[4ex] \end{matrix}}
    \Rule(Spec$_Z$)
  \\[-1ex]
    \input{./Figures-Source/ZH-bialg-white-H.tikz} & \input{./Figures-Source/ZH-bott-white-H.tikz} &
    \mathclap{\begin{matrix} \\[6ex] \end{matrix}}
    \Rule(Bialg$_{ZH}$)
  \\[2ex]
    \input{./Figures-Source/ZH-bialg-white-gray.tikz} & \input{./Figures-Source/ZH-bott-white-gray.tikz} &
    \mathclap{\begin{matrix} \\[6ex] \end{matrix}}
    \Rule(Bialg$_{ZX}$)
  \\[1ex]
    \input{./Figures-Source/ZH-decouple-white-H-a.tikz} & \input{./Figures-Source/ZH-disjunct.tikz} &
    \mathclap{\begin{matrix} \\[6ex] \end{matrix}}
    \Rule(Dilem)
  \\[-2ex]
    \input{./Figures-Source/ZH-average-loop.tikz} &  \input{./Figures-Source/ZH-average-prep-w-sqrt2.tikz} &
    \mathclap{\begin{matrix} \\[4ex] \end{matrix}}
    \Rule(Avg)
  \end{array}
\\
\hline\hline
\end{array}
\end{gather*}~\\[-3ex]
These axioms are closely related to the original axioms presented by Backens and Kissinger~\cite{BK-2019}, with the principal difference that fewer of the rules introduce scalar gadgets.

\vspace*{-3ex}
\paragraph*{A common standard model $\sem{\,\cdot\,}_\nu$ for these calculi --- \!\!\!}
\label{sec:commonModel}
%
The calculi above are sound with respect to a common model $\sem{\,\cdot\,}_\nu$
(described by Eqn.~\eqref{eqn:modelNu}, on page~\pageref{eqn:modelNu} in Appendix~\ref{apx:constructingWellTemperedCalculi}) which identifies the families of white dots and green dots, and the Hadamard with the phase-free H-box.$\phantom{\mathclap{\big\vert}}$ 
This model differs from the standard models of the existing versions of the ZX and ZH calculi in the normalisations of the generators.
\begin{subequations}%
\allowdisplaybreaks
\label{eqn:newNormalisedUnitaryEgs}%
In particular,$\mathclap{\phantom{\big\vert}}$ defining the constant \smash{$\nu = 2^{-1\!\!\:/\!\!\;4}$}, the following equalities hold (\emph{c.f.}~Eqns.~\eqref{eqn:oldModelsNonUnitary} for some of the corresponding diagrams in the pre-existing versions of the ZX and ZH calculi):
\vspace*{-1ex}
\begin{footnotesize}%
\begin{gather}
\label{eqn:newCalculiStates}
\mspace{-24mu}
  \Bigsem{\;\input{./Figures-Source/ZX-nu-red-prep.tikz}\,}_\nu
  \!\!\!\:=
  \ket{\zsymb},
\mspace{42mu}
  \Bigsem{\;\input{./Figures-Source/ZX-nu-green-prep.tikz}\,}_\nu
  \!\!\!\:=
  \ket{\psymb},
\mspace{42mu}
  \bigsem{\input{./Figures-Source/ZH-nu-white-prep.tikz}\;}_\nu
  \,=\,
  \ket{\psymb},
\mspace{42mu}
  \bigsem{\input{./Figures-Source/ZH-nu-H-prep.tikz}\;}_\nu
  \,=\,
  \ket{\msymb},
\mspace{-24mu}
\\[1ex]
\mspace{-24mu}
  \biggsem{\!\input{./Figures-Source/ZX-green-phase-dot-arity-2.tikz}}_\nu
  \!\!\!\!=\!\!\;
  \ket{\zsymb}\!\!\bra{\zsymb}
  +
  \e^{i\theta} \!\!\: \ket{\osymb}\!\!\bra{\osymb}
  ,
\mspace{18mu}
  \biggsem{\!\input{./Figures-Source/ZX-red-phase-dot-arity-2.tikz}}_\nu
  \!\!\!\!=\!\!\;
  \ket{\psymb}\!\!\bra{\psymb}
  +
  \e^{i\theta} \!\!\: \ket{\msymb}\!\!\bra{\msymb}
  ,
\mspace{24mu}
  \Bigsem{\!\input{./Figures-Source/ZH-H-box.tikz}}_\nu
  \!\!\!\!=\!\!\:
  \text{\normalsize$\tfrac{1}{\sqrt 2}$}
  \text{\footnotesize$
    \begin{bmatrix}
      1 \!\!&\!\! \phantom-1 \\ 1 \!\!&\!\! -1
    \end{bmatrix}
  $}
  ,
\mspace{24mu}
  \Bigsem{\!\input{./Figures-Source/ZH-not-dot-short.tikz}\,}_\nu
  \!\!\!\!=
  \text{\footnotesize$
    \begin{bmatrix}
      0 \!\!&\!\! 1 \\ 1 \!\!&\!\! 0
    \end{bmatrix}
  $}
  ,
\\[1ex]
\mspace{-24mu}
  \Biggsem{\!\!\;\input{./Figures-Source/ZX-CNOT-gadget.tikz}\!\:}_{\!\nu}
  \!\!=\,
  \Biggsem{\,\input{./Figures-Source/ZH-CNOT-gadget.tikz}\,}_{\!\nu}
  \!\!=\,
  \mathrm{CNOT}
  \,=
  \text{\footnotesize$\begin{bmatrix}
    1 \!&\! 0 \!&\! 0 \!&\! 0 \\
    0 \!&\! 1 \!&\! 0 \!&\! 0 \\
    0 \!&\! 0 \!&\! 0 \!&\! 1 \\
    0 \!&\! 0 \!&\! 1 \!&\! 0
  \end{bmatrix}$}\!\!\:,
\mspace{72mu}
  \Sem{11ex}{\input{./Figures-Source/ZH-CkZ-gadget.tikz}}_{\!\nu}
  \!\!\!\!\!\:
  =\,
  \mathrm{C}^{k{-}1} \!\!\: Z
  \,=\,
  \def\vdots{\mathclap{\raisebox{1.5ex}.}\mathclap{\raisebox{0.75ex}.}\mathclap{.}}
  \def\ddots{\mathclap{\raisebox{1ex}{.}\raisebox{0.5ex}{.}\raisebox{-0.ex}{.}}}
  \def\cdots{\mathclap{\cdot\!\cdot\!\cdot}}
  \text{\footnotesize$\begin{bmatrix}
    1 \!&\! 0 \!&\! \cdots \!&\! 0 \\
    0 \!&\! \,\ddots \!&\!  \!&\! \vdots \\
    \vdots \!&\!  \!&\! 1 \!&\! 0 \\
    0 \!&\! \cdots \!&\! 0 \!&\! \!\!\; -1
  \end{bmatrix}$}\!\!\:.
\mspace{-24mu}
\end{gather}%
\end{footnotesize}%
\end{subequations}~\\[-5ex]%

\vspace*{-2ex}
\paragraph*{An analysis of normalisation constraints for idealised rewrites --- \!\!\!}
%
Our well-tempered calculi do not banish scalar adjustments to the normalisation altogether for the analysis of unitary circuits.%
  \footnote{%
    See, \eg,~the discussion preceding Lemma~\ref{lemma:greenDotScalars} on page~\pageref{lemma:greenDotScalars}, or the discussion of specialness and supplementarity on page~\pageref{discn:specialAndSupplementarity}.
  }
One may show in fact that it is not possible to do so.
While the construction of the calculi of pages~\pageref{newZX} and~\pageref{newZH} is the motivation for this article, our main technical contribution (in Appendices~\ref{apx:compatibilityRewritesZX} and~\ref{apx:compatibilityRewritesZH}) are the constraints and trade-offs in ``reasonable'' models for which various simplified rewrites are sound, by reduction to the existing versions of these calculi.

\vspace*{-2ex}
\paragraph*{Related work --- \!\!\!}
%
To the best of our knowledge, our presentation is the first work on either the ZX or ZH calculi which attempts to describe simplified rewrites while remaining sound for a model of complex matrices (rather than equivalence classes of such matrices).
Concurrently, Carette and Jeandel~\cite{CJ-2020} developed a classification of all ``\!\;$\mathrm{Z^\ast}$ calculi'', including the ZX calculus and the ZH calculus: the well-tempered calculi correspond to the calculi $\mathrm{Z^{\smash{(\sqrt 2,1)}}X_{\smash{(\sqrt 2,1)}}}$ and $\mathrm{Z^{\smash{(1,1)}} H_{\smash{(\sqrt 2, \!\:-\!\!\;1\!\!\:/\!\!\;2)}}}$ up to a symmetrising isomophism (see the discussion in Section~\ref{sec:classificationZ*} on page~\pageref{sec:classificationZ*}).
\label{discn:introZ*}%
---
We note that our versions of the ZX and ZH calculi are \emph{not} intended to simply replace the pre-existing versions, for all foreseeable applications:
\begin{itemize}[topsep=2pt,itemsep=0pt]
\item  
  The standard model of the pre-existing scalar exact ZX calculus is ideally suited to describe surface code lattice surgery~\cite{BH-2020,BDHP-2019};
\item
  The original version of the ZH calculus denotes integer matrices whenever the H-boxes take integer parameters, and thus is ideally suited for analyses of counting and and gap complexity~\cite{DHM-2002,BKM-2020}.
\end{itemize}
What these new versions of the ZX and ZH calculi \emph{are} intended to do, is to simplify the task of performing scalar-exact computations for procedures which are dominated by unitary transformations, and to facilitate using the two calculi interoperably (through the common standard model $\sem{\,\cdot\,}_\nu$) for that task.

\vspace*{-2ex}
\paragraph*{Structure of the paper ---\!\!\!}
%
Section~\ref{sec:preliminaries} provides background on string diagrams and the existing presentations of the ZX and ZH~calculi (extended remarks on the normalisations of their standard models is left to Appendix~\ref{apx:normalisationZXandZH}).
Section~\ref{sec:constructingNewCalculi} summarises the way in which we may construct the model $\sem{\,\cdot\,}_\nu$ from a parameterised model by imposing successive constraints (supported by analysis in Appendices~\ref{apx:denotationalConstraints}--\ref{apx:constructingWellTemperedCalculi}).
Section~\ref{sec:features} describes some features of interest of these calculi, contrasting them to the pre-existing ZX and ZH calculi --- presenting a few simple examples of derivations in doing so --- and describing their relationships to other work on the ZX and ZH calculi.
We conclude in Section~\ref{sec:conclusions} with a few pragmatic remarks concerning the use of these diagrammatic calculi.

\vspace*{-2ex}
\section{Preliminaries}
\label{sec:preliminaries}
\vspace*{-1ex}

In this section, we present the background for our work, including descriptions of the pre-existing versions of the ZX and ZH calculus for ease of reference and comparison.

\vspace*{-2ex}
\subsection{String diagrams}
\vspace*{-.5ex}

The ZX and ZH calculi are both systems to represent quantum operations by ``string diagrams''.
These diagrams consist of dots or boxes, connected by wires, in effect denoting tensor networks%
  \footnote{%
    That this notation represents a tensor network --- or indeed represents anything whatsoever in a well-defined way --- can be established using category theory; however, no category theory will be needed to understand our results.
  }
whose coefficients range over a set such as $\mathbb C$ or $\mathbb N$.
A wire can have one or two ``loose'' ends which are not connected to a dot or box: these represent inputs and outputs of the operation.
In this article, we consider string diagrams in which wires represent the state-space of a qubit, and loose ends of wires will be oriented towards the left (for inputs) or the right (for outputs) of the diagram.

As with standard quantum circuit diagrams, we may build diagrams from basic generators which represent matrices.
Connecting diagrams left-to-right represents multiplying matrices sequentially, and juxtaposing diagrams vertically represents taking the tensor product.
We may also permute qubit state-spaces by crossing wires over one another (corresponding to exchanging tensor indices) and bend wires back on themselves (corresponding to a basis-dependent isomorphism between the space $\C^2$ and its dual).
The correspondence between diagrams and matrices (defined even for diagrams consisting only of wires with loose ends) is provided by  a ``model'' $\sem{\,\cdot\,}$, which maps each diagram to some matrix over $\mathbb C$.
In this work, we consider only models $\sem{\,\cdot\,}$ for which 
the following equations hold:
\vspace*{-1ex} 
\begin{gather}
\label{eqn:stringGenerators}
\mspace{-18mu}
  \bigsem{\input{./Figures-Source/id-wire.tikz}\,}
    \!\!\:=\!\!\:
    \text{\footnotesize$\begin{bmatrix}
      1 \!\!&\!\! 0 \\[-.25ex]
      0 \!\!&\!\! 1
    \end{bmatrix}$},
    \mspace{36mu}
    \biggsem{\input{./Figures-Source/swap.tikz}\,}
    \!\!\:=\!\!\:
    \text{\footnotesize$\begin{bmatrix}
      1 \!\!&\!\! 0 \!\!&\!\! 0 \!\!&\!\! 0 \\[-.25ex]
      0 \!\!&\!\! 0 \!\!&\!\! 1 \!\!&\!\! 0 \\[-.25ex]
      0 \!\!&\!\! 1 \!\!&\!\! 0 \!\!&\!\! 0 \\[-.25ex]
      0 \!\!&\!\! 0 \!\!&\!\! 0 \!\!&\!\! 1
    \end{bmatrix}$},
    \mspace{36mu}
    \biggsem{\input{./Figures-Source/cup.tikz}\,}
    \!\!\:=\!\!\:
    \ket{\zsymb\zsymb}\!\!\!\; +\! \ket{\osymb\osymb},
    \mspace{27mu}
    \biggsem{\,\input{./Figures-Source/cap.tikz}\,}
    \!\!\:=\!\!\:
    \bra{\zsymb\zsymb}\! + \!\!\!\;\bra{\osymb\osymb}.
\mspace{-9mu}
\end{gather}~\\[-2ex]
If the meaning of each node as tensor in the model $\sem{\,\cdot\,}$ (its \emph{semantics}) is symmetric, \ie~unaffected by permuting its tensor factors --- and if their coefficients are defined with respect to a basis of real vectors --- we may ignore the directions of wires when interpreting or transforming the diagrams.
This is often described as a rule of ``only the topology matters'', which we use freely in our work.

As with any notation, it is possible to do calculations with string diagrams, provided that there are enough ``rules'' describing how to manipulate diagrams in a meaning-preserving way (\emph{i.e.},~preserving the semantics in a model $\sem{\,\cdot\,}$).
A~set of rules (or ``basic rewrites'') to do so defines a \emph{diagrammatic calculus}.
In this article, we consider two specific examples: the ``ZX calculus'' and ``ZH calculus'', described below.

\vspace*{-2ex}
\subsection{The ZX and ZH calculi}
\label{sec:traditionalZXandZH}
\vspace*{-.5ex}

\begin{figure*}[t]
\centering
\input{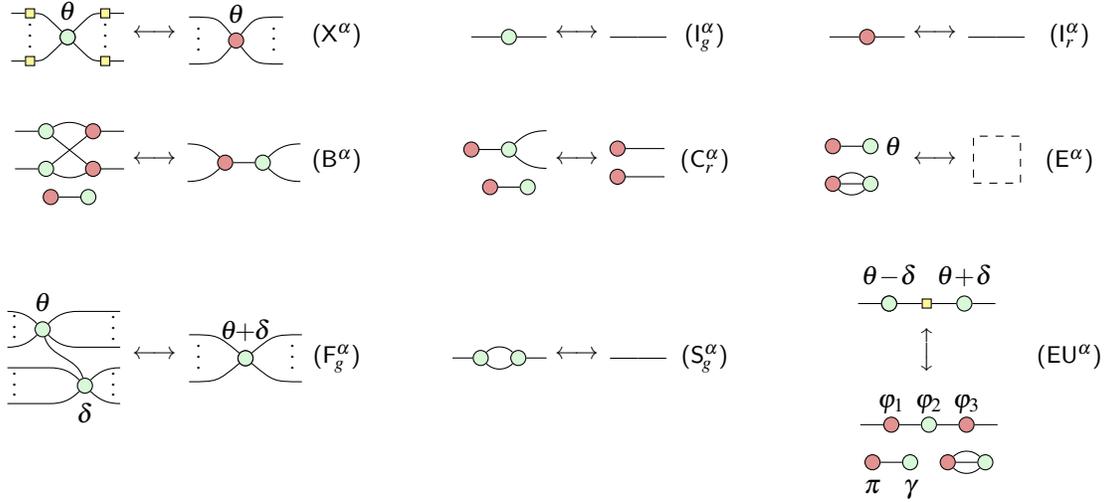}

\vspace*{-0.5ex}
 \caption[]{%
    \label{fig:usualZXrules}
    A slight modification of the rules of Ref.~\cite[Fig.\;2]{Vilmart-2019-minimal}, taken here as a representative of the pre-existing versions of the ZX calculus.
    Throughout, ``\mbox{\,.\llap{\raisebox{0.625ex}.}\llap{\raisebox{1.25ex}.}\,}'' indicates zero or more wires, and $\theta , \delta \in \R$ may be arbitrary (sums may be taken modulo $2\pi$).
    The rules \Rule(F$_{\!g}^\alpha$) and \Rule(S$_{\!g}^\alpha$) together are equivalent to the rule \Rule(S) of Ref.~\cite{Vilmart-2019-minimal}, and the rule \Rule(X$^{\alpha}$) is equivalent to the rule \Rule(H) of Ref.~\cite{Vilmart-2019-minimal}.
    The other rules are identical to \Rule(I$_g$), \Rule(I$_r$), \Rule(B), \Rule(CP), \Rule(IV), and \Rule(EU$'$) of Ref.~\cite[Fig.\;2]{Vilmart-2019-minimal}.
    Note that the right-hand side of the rule \Rule(E$^\alpha$) is the empty diagram, whose interpretation is the scalar $1$; and $\varphi_1$, $\varphi_2$, $\varphi_3$, and $\gamma$ in the rule \Rule(EU$^\alpha$) are as described in Eqns.~\eqref{eqn:VilmartEulerAngleComputation}
  }
\vspace*{-1ex}
\end{figure*}  

The ZX calculus is in fact a collection of closely related diagrammatic calculi of the sort described above --- all equivalent up to scalar factors, and the most recent versions of which are exactly equivalent.
It is effective for representing operations generated by single-qubit rotations and controlled-NOT gates.
The ZH~calculus~\cite{BK-2019,WW-2019} was developed as an alternative to the ZX~calculus for reasoning about quantum computation, in which higher-arity generalisations of the Hadamard box play a central role.
Such nodes are suitable for representing multiply-controlled-Z operations, and thus for representing circuits over the Hadamard-Toffoli gate set~\cite{Shi-2003,Aharonov-2003}.
In this section, we describe examples of the ZX and ZH calculi (minor variations of the pre-existing presentations of these calculi) as the starting point of our analysis.

\vspace*{-1ex}
\paragraph{The ZX calculus.}
We take a slightly modified version of the presentation of Vilmart~\cite[Fig.\;2]{Vilmart-2019-minimal} as representative of existing versions of the ZX calculus.
In addition to the wire diagrams of Eqn.~\eqref{eqn:stringGenerators}, it has the following ``dot'' and ``box'' generators, where $m,n \in \N$ may be arbitrary:
\vspace*{-1ex}
    \begin{equation}
    \label{eqn:ZXnodeFamilies}
    \begin{gathered}
      \input{./Figures-Source/ZX-green-phase-dot-arity.tikz}
    \;,\quad
      \input{./Figures-Source/ZX-red-phase-dot-arity.tikz}
    \;,\quad
      \input{./Figures-Source/ZX-H-box.tikz}
    \;\;.
    \end{gathered}
    \end{equation}~\\[-2.5ex]
The first (lighter coloured) family and the second (darker coloured) family each contain nodes of type $m \to n$ for arbitrary $m,n \in \N$, and admit a ``phase'' parameter $\theta \in \R$.
We call these families ``green'' and ``red'' nodes respectively (or occasionally ``Z'' and ``X'' nodes); the third generator we call the Hadamard box.
(We frequently omit the phase parameter for green nodes with $\theta = 0$, which we call ``phase-free''; we may also represent Hadamard boxes by  boxes without the $H$ symbol.)
We interpret these as tensors using a model $\sem{\,\cdot\,}_\alpha$\,, defined by:
\vspace*{-0.5ex}
  \begin{small}
  \begin{align}{}
  \label{eqn:origStandardModel}
  \begin{split}
  \mspace{-18mu}
  \Biggsem{\!\!\!\input{./Figures-Source/ZX-green-phase-dot-arity.tikz}\!\!\!}_{\!\alpha}
  &=\,
    \ket{\zsymb}^{\!\otimes n}\!\bra{\zsymb}^{\!\!\;\otimes m}
    \,+\;
    \e^{i\theta}
    \ket{\osymb}^{\!\otimes n}\!\bra{\osymb}^{\!\!\;\otimes m},
  \mspace{-18mu}
%
%
\\[1ex]
%
%
  \mspace{-18mu}
  \Biggsem{\!\!\!\input{./Figures-Source/ZX-red-phase-dot-arity.tikz}\!\!\!}_{\!\alpha}
  &=\,
    \ket{\psymb}^{\!\otimes n}\!\bra{\psymb}^{\!\!\;\otimes m}
    \,+\;
    \e^{i\theta}
    \ket{\msymb}^{\!\otimes n}\!\bra{\msymb}^{\!\!\;\otimes m},
  \mspace{-18mu}
\end{split}
%
%
&
%
%
  \Bigsem{\,\input{./Figures-Source/ZX-H-box.tikz}\,}_\alpha
  \,&=\,
  \mbox{\normalsize$\tfrac{1}{\sqrt 2}$} \text{\small$\begin{bmatrix}
    1 \!&\! \phantom-1 \\ 1 \!&\! -1
  \end{bmatrix}$}
  .
  \mspace{-18mu}
  \end{align}
  \end{small}~\\[-4.5ex]

\noindent
The basic rewrites of the calculus are shown in Figure~\ref{fig:usualZXrules} %
(where in the rule \Rule(EU$^\alpha$), the angles $\varphi_1, \varphi_2, \varphi_3, \gamma \in \R$ satisfy Eqn.~\eqref{eqn:VilmartEulerAngleComputation}).

\begin{figure*}[t]
\centering
\input{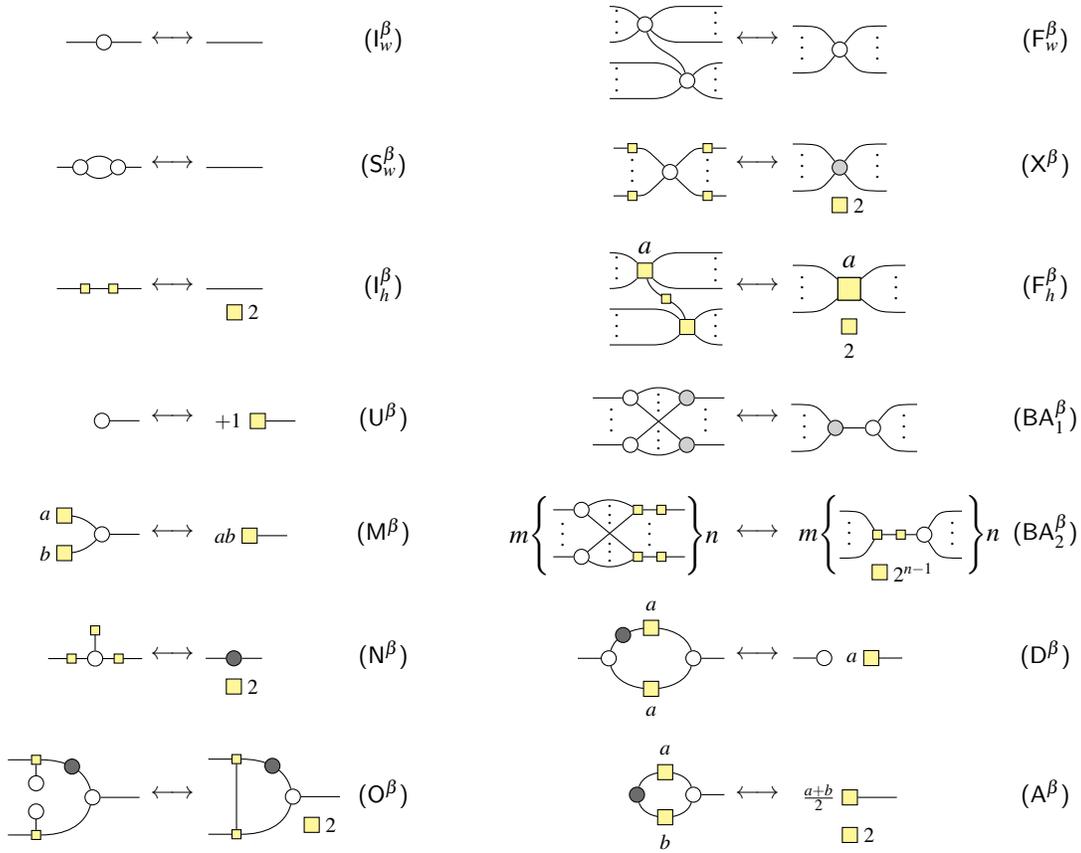}

\vspace*{-3ex}
 \caption[]{
    \label{fig:variantZHrules}
    A set of rules, representing the pre-existing version of the ZH calculus but modifying the rules of Ref.~\cite[Fig.\;1]{BK-2019}.
    Throughout, ``\mbox{\,.\llap{\raisebox{0.625ex}.}\llap{\raisebox{1.25ex}.}\,}'' indicates zero or more wires, and  $m,n \in \N$ and $a,b \in \C$ may be arbitrary.
    We introduce rewrites \Rule(N$^{\smash\beta}$) and \Rule(X$^{\smash\beta}$) in place of the definitions in Ref.~\cite[Eqns.\;1\,\&\,2]{BK-2019}.
    Otherwise, the correspondence between the rewrites above and those of Ref.~\cite{BK-2019} is as follows:\!
    $\smash{\Rule(F$_z^{\smash\beta}$) \equiv \Rule(ZS1)}$;\,\!
    $\smash{\Rule(F$_h^{\smash\beta}$) \equiv \Rule(HS1)}$;\,\!
    $\smash{\Rule(I$_z^{\smash\beta}$) \mathbin\& \Rule(S$_z^\beta$) \equiv \Rule(ZS2)}$;\,\!
    $\smash{\Rule(I$_h^{\smash\beta}$) \equiv \Rule(HS2)}$;\,\!
    $\smash{\Rule(BA$_1^{\smash\beta}$) \equiv \Rule(BA1)}$;\,\!
    $\smash{\Rule(BA$_2^{\smash\beta}$) \equiv \Rule(BA2)}$;\,\!
    $\smash{\Rule(D$_h^{\smash\beta}$) \equiv \Rule(I)}$;\,\!
    $\smash{\Rule(U$_h^{\smash\beta}$) \equiv \Rule(U)}$;\,\!
    $\smash{\Rule(M$_h^{\smash\beta}$) \equiv \Rule(M)}$;\,\!
    $\smash{\Rule(O$_h^{\smash\beta}$) \equiv \Rule(O)}$;\;and\;%
    $\smash{\Rule(A$_h^{\smash\beta}$) \equiv \Rule(A)}$.
    In particular, {\Rule(BA$_2^{\smash\beta}$)} follows by post-composion of both sides of \Rule(BA2) with phase-free H-boxes, and we relabel the ``intro'' rule \Rule(I) as {\Rule(D$^{\smash\beta}$)} to avoid notational clash (we pronounce it as ``dilemma'').
  }
  \vspace*{-1ex}
\end{figure*}

\vspace*{-1ex}
\paragraph{The ZH calculus.}
We consider a presentation of the ``original'' ZH calculus which differs slightly from the presentation of Backens and Kissinger~\cite{BK-2019}.%
  \footnote{%
    In fact, our version of the ZH calculus is a conservative extension of that of Ref.~\cite{BK-2019}: strictly speaking, only the equivalences involving only white dots and H-boxes are common to both.
    To simplify the presentation, we prefer to take gray dots and not-dots as generators which may have different interpretations under different models $\sem{\,\cdot\,}$, rather than gadgets with varying  definitions with respect to different calculi.
    Despite this technical difference, we feel this is still in the vein of ``the ZH calculus''.
  }
In addition to the wire diagrams of Eqn.~\eqref{eqn:stringGenerators}, our presentation of the ZH calculus has the following ``dot'' and ``box'' generators, for $m,n \in \N$ arbitrary:
\vspace*{-0.5ex}
    \begin{equation}
    \label{eqn:ZHnodeFamilies}
    \begin{gathered}
      \input{./Figures-Source/ZH-white-dot-arity.tikz}
    ,\quad
      \input{./Figures-Source/ZH-H-phase-box-arity.tikz}
    ,\quad
      \input{./Figures-Source/ZH-gray-dot-arity.tikz}
    ,\quad
      \input{./Figures-Source/ZH-not-dot.tikz}
    \;\;.
    \end{gathered}
    \end{equation}~\\[-2ex]
The first three families each contain nodes of type $m \to n$ for arbitrary $m, n \in \N$, and the second admits a ``phase'' parameter $a \in \C$.
We call these families ``white dots'', ``H boxes'', and ``gray dots'' respectively; the fourth generator we call the ``not-dot''.%
  \footnote{%
    Backens and Kissinger~\cite{BK-2019} treat the gray nodes and the not-dot as gadgets (``derived generators''), defined in order to simplify the presentation of their rewrites, rather than primitive nodes.
    We find it useful to treat them as primitive generators.
  }
(We may omit the phase of H-boxes with $a = -1$.)
We interpret these as tensors using a model $\sem{\,\cdot\,}_\beta$, defined as follows:
\vspace*{-0.5ex}
\begin{small}%
\begin{equation}{}
\label{eqn:origZHModel}%
\mspace{-72mu}
\begin{aligned}
  \Biggsem{\!\!\!\!\input{./Figures-Source/ZH-white-dot-arity.tikz}\!\!\!}_{\!\beta}
  \!&=\;
    \ket{\zsymb}^{\!\otimes n}\!\bra{\zsymb}^{\!\!\;\otimes m}
    \,+\;
    \ket{\osymb}^{\!\otimes n}\!\bra{\osymb}^{\!\!\;\otimes m},
%
%
&\qquad
%
%
  \Biggsem{\!\!\!\!\input{./Figures-Source/ZH-gray-dot-arity.tikz}\!\!\!}_{\!\beta}
  \;&=\quad
    \mathop{\;\sum \;\sum\;}_{%
      \mathclap{\substack{
        {x \in \{\zsymb,\osymb\}^m \!,}
        \,\;
        {y \in \{\zsymb,\osymb\}^n} \\
      w(x) \!\!\;+\!\!\; w(y) \,\in\, 2\Z
      }}}
      \;\;
      \ket{y}\!\!\bra{x} \!\!\:,
      \mspace{-48mu}
%
%
\\[-.25ex]
%
%
  \Biggsem{\!\!\!\!\input{./Figures-Source/ZH-H-phase-box-arity.tikz}\!\!\!}_{\!\beta}
  \!&=\;
    \mathop{\;\sum \;\sum\;}_{%
      \mathclap{
        \substack{
        {x \in \{\zsymb,\osymb\}^m}
        \\\
        {y \in \{\zsymb,\osymb\}^n}
      }}}
      \,
      a^{x_1 \cdots x_m y_1 \cdots y_n} \!
      \ket{y}\!\!\bra{x}\,, 
%
%
&
%
%
  \Bigsem{\input{./Figures-Source/ZH-not-dot.tikz}}_\beta
  \;&=\;
  \text{\footnotesize$\begin{bmatrix}
    0 \!\!&\!\! 1 \\[-0.25ex] 1 \!\!&\!\! 0
  \end{bmatrix}$}.
  \mspace{-18mu}
\end{aligned}%
\end{equation}%
\end{small}~\\[-2ex]
Note in particular that any degree-0 H box denotes a scalar factor:
\vspace*{-0.75ex}
\begin{equation}
  \label{eqn:ZH-scalar-box}
  \Bigsem{\input{./Figures-Source/ZH-scalar-box.tikz}}_{\!\beta}
  \;=\;\,\;
    \sum_{\mathclap{\text{(singleton)}}}
      \;
      a^{\text{(empty product)}} \cdot 1
  \;=\;
    a^1
  \;=\;
    a.
\end{equation}~\\[-5.25ex]
%

\noindent
The rules of our variant of the ZH calculus are shown in Figure~\ref{fig:variantZHrules}.
We introduce rules \Rule(X$^\beta$) and \Rule(N$^\beta$) to syntactically provide the same meaning for gray dots and the not-dot as in~\cite[Eqns.\,1\,\&\,2]{BK-2019}.
Note that our rule \smash{\Rule(BA$_2^\beta$)} differs significantly from the corresponding rule \Rule(BA2): in particular, \Rule(BA2) doesn't introduce a scalar factor, and involves a gray node of type $1 \to n$ node in place of a Z node of type $1 \to n$.
Our alternative rule \Rule(BA$_2^\beta$) instead highlights how the gadget of two phase-free H boxes itself interacts with the Z dots in a similar way to how the gray dots do in \smash{\Rule(BA$_1^\beta$)} --- \emph{i.e.},~as a scaled bialgebra.%
  \footnote{%
    This observation is implicit in Ref.~\cite{BK-2019}, which describes \Rule(BA2) as a ``bialgebra rule''.
}
As it seems likely to that this interaction will play an important role in how the ZH~calculus may be used in practise, we adopt the rule \smash{\Rule(BA$_2^\beta$)} in place of \Rule(BA2) in our reference presentation of the ZH calculus.

\vspace*{-1ex}
\subsection{The situation with scalars in the existing presentations of the ZX and ZH calculi}
\label{sec:scalarFactors}
\vspace*{-.5ex}

In the existing versions of the ZX and ZH calculus, ``scalar gadgets'' (closed sub-diagrams) are involved both in common basic rewrites, and in the representation of common unitary gates.
We first consider representations of quantum operations which require scalar gadgets.
\begin{subequations}%
\label{eqn:oldModelsNonUnitary}%
For the ZX calculus, we have
\vspace*{-1.5ex}
\begin{footnotesize}%
\begin{gather}
\label{eqn:newZXnormalisedUnitary}
  \ket{\zsymb}
  \,=\,
  \Bigsem{\;\input{./Figures-Source/ZX-ket0-orig.tikz}\,}_\alpha
\mspace{64mu}
  \ket{\psymb}   
  \,=\,
  \Bigsem{\;\input{./Figures-Source/ZX-ketplus-orig.tikz}\,}_\alpha
  ,
\mspace{64mu}
  \mathrm{CNOT}
  \,=\,
  \Sem{6.5ex}{\!\!\;\input{./Figures-Source/ZX-CNOT-orig.tikz}\!\:}_{\!\alpha} \;,
\end{gather}%
\end{footnotesize}~\\[-2ex]
These diagrams involve scalar gadgets with one green dot and one red dot,  representing the constants $1\!\!\;/\!\!\:\sqrt2$ and $\sqrt 2$.
(As an exercise, the reader is invited to prove which is which from first principles.)
Using scalar boxes, the ZH calculus more transparently represents its needed corrections to the normalisation:
\vspace*{-1.5ex}
\begin{footnotesize}
\begin{gather}
  \ket{\psymb}
  \,=\,
  \Bigsem{\input{./Figures-Source/ZH-ketplus-orig.tikz}\;}_\beta
\mspace{64mu}
  \ket{\msymb}
  \,=\,
  \Bigsem{\input{./Figures-Source/ZH-ketminus-orig.tikz}\;}_\beta
\mspace{64mu}
  H
  \,=\,
  \mbox{\normalsize$\tfrac{1}{\sqrt 2}$}
  \text{\footnotesize$
    \begin{bmatrix}
      1 \!\!&\!\! \phantom-1 \\ 1 \!\!&\!\! -1
    \end{bmatrix}
  $}
  \,=\,
  \Sem{7ex}{\!\input{./Figures-Source/ZH-Hadamard-orig.tikz}}_\beta
  .
\end{gather}%
\end{footnotesize}%
\end{subequations}~\\[-3ex]
Compared to conventional quantum circuit notations, the use of scalar gadgets to represent such basic operations is conspicuous.
This might be dismissed as a one-time inconvenience for any given quantum procedure to be represented, if not for the similar scalar gadgets which are introduced or removed by several of the rewrite rules --- most notably, the rules \Rule(B$^\alpha$), \Rule(C$^\alpha_r$), and \Rule(EU$^\alpha$) for the ZX calculus, and rules \Rule(X$^\beta$), \Rule(I$^\beta_h$), \Rule(F$^\beta_h$), \Rule(BA$_2^\beta$) of the ZH calculus.
As a result, frequent bookkeeping of these gadgets is required for scalar exact reasoning.

Partly as a result of the inconvenience of tracking scalar gadgets~\cite{Coecke-personal-2020}, the research programme on the ZX calculus has sometimes ignored scalar factors altogether (see, \eg,~Ref.~\cite{DP-2009}).\footnote{%
 Indeed, some well-regarded participants in this programme~\cite{Coecke-personal-2020,Bakewell-personal-2020} feel that scalar factors are still of minor importance.
}
This suffices to consider the question of how to prove, for two diagrams $D_1$ and $D_2$, when the operators $\sem{D_1}$ and $\sem{D_2}$ are proportional by a non-zero scalar.
For some tasks in quantum information theory, this is adequate --- for instance, two unitary operators are proportional to one another if and only if they differ by at most an unimportant global phase.
Thus, one may produce useful results with the ZX calculus without taking the effort to maintain the normalisation of terms; and the same is true of the ZH calculus.
However, tracking normalisation is important in application to quantum information processing in general, as it may correspond to the \emph{probability} with which a given transformation or physical effect is realised, which is an important issue in quantum technologies.

There are good theoretical motivations for the original standard models (and therefore also for the rewrites of Figures~\ref{fig:usualZXrules} and~\ref{fig:variantZHrules}) for the ZX and ZH calculi as presented in Eqns.~\eqref{eqn:origStandardModel} and~\eqref{eqn:origZHModel}.
For some applications, the models $\sem{\,\cdot\,}_\alpha$ and $\sem{\,\cdot\,}_\beta$ also have good practical motivations.
(These motivations are discussed in some detail in Appendix~\ref{apx:normalisationZXandZH}.)
\label{discn:conventionalNormalisation}%
However, it would be helpful to have scalar-exact variants of the ZX and ZH calculi in which neither the representations of basic unitary gates, nor the most important rewrite rules, involved scalar gadgets.
This serves as the motivation for this work.

\vspace*{-2.5ex}
\section{Constructing differently normalised ZX and ZH calcului}
\label{sec:constructingNewCalculi}
\vspace*{-.75ex}

In this section, we prove the soundness and completeness of the calculi presented on pages~\pageref{newZX} and~\pageref{newZH}, by reduction to the pre-existing versions of the ZX and ZH calculi.
We do this by constructing the model $\sem{\,\cdot\,}_\nu$ as a notation supporting both simple representations of unitary operators and simple rewrites, and considering the scalar differences between  $\sem{\,\cdot\,}_\nu$ and the models  $\sem{\,\cdot\,}_\alpha$  and  $\sem{\,\cdot\,}_\beta$\,.

\vspace*{-1ex}
\subsection{Constraints on denotation}
\label{sec:denotationalConstraints}
\vspace*{-.5ex}

We define $\sem{\,\cdot\,}_\nu$ to satisfy Eqn.~\eqref{eqn:stringGenerators}, and so that 
$\sem{ A }_\nu \propto \sem{ A }_\alpha$ for each ZX generator $A$, and $\sem{ B }_\nu \propto \sem{ B }_\beta$ for each ZH generator $B$.
From this, it follows that
\begin{equation}
  \biggsem{\!\!\!\!\input{./Figures-Source/ZX-green-dot-arity.tikz}\!\!\!}_{\!\nu} \!\!
  \;\propto\;
  \biggsem{\!\!\!\!\input{./Figures-Source/ZH-white-dot-arity.tikz}\!\!\!}_{\!\nu} ,
\qquad\qquad
  \Bigsem{\input{./Figures-Source/ZX-H-box.tikz}}_{\!\nu}
  \;\propto\;
  \Bigsem{\input{./Figures-Source/ZH-H-minus1-box.tikz}}_{\!\nu} .
\end{equation}~\\[-1ex]
To help the calculi to work interoperably, we define $\sem{\,\cdot\,}_\nu$ so that in fact these proportionalities hold with equality.
Apart from this, we define $\sem{\,\cdot\,}_\nu$ as flexibly as possible (subject to an ``only the topology matters'' constraint), to maximise the chances of finding normalisations of the generators which satisfy the constraints we impose.
\begin{subequations}%
\label{eqn:parametersModelNu}%
To this end, we define the semantics of the generators up to some families of non-negative coefficients $(u_k)_{k \in \N}$, $(v_k)_{k \in \N}$, $(g_k)_{k \in \N}$, $(h_k)_{k \in \N}$, and $\xi$\,.
For the ZX calculus, we define
\vspace*{-1.5ex}
\begin{small}%
\begin{align}{}
    \Biggsem{\!\!\input{./Figures-Source/ZX-green-phase-dot-arity.tikz}\!\!}_{\!\nu}
  \,&=\,
    u_{m{+}n} \Bigl(
    \ket{\zsymb}^{\!\otimes n}\!\bra{\zsymb}^{\!\!\;\otimes m}
    \,+\;
    \e^{i\theta}
    \ket{\osymb}^{\!\otimes n}\!\bra{\osymb}^{\!\!\;\otimes m}
    \Bigr),
%
%
\\[.5ex]
%
%
    \Biggsem{\!\!\input{./Figures-Source/ZX-red-phase-dot-arity.tikz}\!\!}_{\!\nu}
  \,&=\,
    v_{m{+}n} \Bigl(
    \ket{\psymb}^{\!\otimes n}\!\bra{\psymb}^{\!\!\;\otimes m}
    \,+\;
    \e^{i\theta}
    \ket{\msymb}^{\!\otimes n}\!\bra{\msymb}^{\!\!\;\otimes m}
    \Bigr),
%
%
\end{align}%
\end{small}~\\[-2ex]
  and for the ZH calculus (and for the Hadamard box of the ZX calculus, taking $a = -1$ and $m=n=1$):
\vspace*{-1.5ex}
\begin{small}
\begin{align}{}
  \mspace{-48mu}
  \Biggsem{\!\!\!\!\input{./Figures-Source/ZH-white-dot-arity.tikz}\!\!\!}_{\!\nu}
  \!\!&=\;
    u_{m\!+\!n} \Bigl(
    \ket{\zsymb}^{\!\otimes n}\!\bra{\zsymb}^{\!\!\;\otimes m}
    \,+\;
    \ket{\osymb}^{\!\otimes n}\!\bra{\osymb}^{\!\!\;\otimes m}
    \Bigr),
%
%
&
%
%
  \Bigsem{\input{./Figures-Source/ZH-not-dot.tikz}}_\nu
  \!&=\;
  \xi \text{\footnotesize$\begin{bmatrix}
    0 \!\!&\!\! 1 \\[-0.25ex] 1 \!\!&\!\! 0
  \end{bmatrix}$},
%
%
\\[.5ex]
%
%
  \mspace{-48mu}
  \Biggsem{\!\!\!\!\input{./Figures-Source/ZH-H-phase-box-arity.tikz}\!\!\!}_{\!\nu}
  \!\!&=\;
    h_{m\!+\!n} 
    \mathop{\;\sum \;\sum\;}_{%
      \mathclap{
        \substack{
        {x \in \{\zsymb,\osymb\}^m}
        \\\
        {y \in \{\zsymb,\osymb\}^n}
      }}}
      \,
      a^{x_1 \cdots x_m y_1 \cdots y_n} \!
      \ket{y}\!\!\bra{x}\!\:, 
%
%
&\mspace{-48mu}
%
%
  \Biggsem{\!\!\!\!\input{./Figures-Source/ZH-gray-dot-arity.tikz}\!\!\!}_{\!\nu}
  \!\!&=\;
    g_{m\!+\!n} \; \mathop{\;\sum \;\sum\;}_{%
      \mathclap{\substack{
        {x \in \{\zsymb,\osymb\}^m \!,}
        \,\;
        {y \in \{\zsymb,\osymb\}^n} \\
      w(x) \!\!\;+\!\!\; w(y) \,\in\, 2\Z
      }}}
      \;
      \ket{y}\!\!\bra{x}.
  \mspace{-48mu}
\\[-5ex]\notag
\end{align}%
\end{small}%
\end{subequations}~\\[-5.5ex]%

\noindent
We next impose constraints to yield calculi with our preferred features.
We first constrain $\sem{\,\cdot\,}_\nu$ to express certain unitary operators simply --- specifically:
\vspace*{-1.5ex}
\begin{subequations}%
\allowdisplaybreaks
\label{eqn:unitaryOps}%
\begin{small}
\begin{alignat}{2}
\label{eqn:unitaryId}
  \bigsem{\,\input{./Figures-Source/ZX-green-id.tikz}\,}_\nu
  \;&=\;
  \bigsem{\,\input{./Figures-Source/ZX-red-id.tikz}\,}_\nu
  \;=\;
  \bigsem{\,\input{./Figures-Source/ZH-white-id.tikz}\,}_\nu
  \;&=\;
  \bigsem{\,\input{./Figures-Source/ZH-gray-id.tikz}\,}_\nu
  \;=\;
  \idop
  \;={}&\;
    \text{\footnotesize$\begin{bmatrix} 1 \!\!&\!\! 0 \\ 0 \!\!&\!\! 1 \end{bmatrix}$},
\\[.5ex]
\label{eqn:unitaryNOTandH}
  \bigsem{\,\input{./Figures-Source/ZH-not-dot.tikz}\,}_\nu
  &=\,
    \mathrm{NOT}
  \,={}
    \text{\footnotesize$\begin{bmatrix} 0 \!&\! 1 \\ 1 \!&\! 0 \end{bmatrix}$},
&
  \Bigsem{\,\input{./Figures-Source/ZX-H-box.tikz}\,}_\nu
  =
  H
  =
  \text{\footnotesize$\frac{1}{\sqrt 2}$}
  &
  \text{\footnotesize$
    \begin{bmatrix}
      1 \!\!&\!\! \phantom-1 \\ 1 \!\!&\!\! -1
    \end{bmatrix}
  $}
  ,
\end{alignat}
\vspace*{-3ex}
\begin{align}
\label{eqn:unitaryCNOT}
  \Biggsem{\,\input{./Figures-Source/ZX-CNOT-gadget.tikz}\,}_\nu
  \!=\,
  \Biggsem{\,\input{./Figures-Source/ZH-CNOT-gadget.tikz}\,}_\nu
  =\,
  \mathrm{CNOT}
  \,={}&
  \text{\footnotesize$\begin{bmatrix}
    1 \!&\! 0 \!&\! 0 \!&\! 0 \\
    0 \!&\! 1 \!&\! 0 \!&\! 0 \\
    0 \!&\! 0 \!&\! 0 \!&\! 1 \\
    0 \!&\! 0 \!&\! 1 \!&\! 0
  \end{bmatrix}$}\!\!\:,
\\
\label{eqn:unitaryCkZ}
  \Sem{13ex}{\input{./Figures-Source/ZH-CkZ-gadget.tikz}}_{\!\nu}
  \!\!
  =
  \mathrm{C}^{k{-}1} \!\!\: Z
    ={}&
  \def\vdots{\mathclap{\raisebox{1.5ex}.}\mathclap{\raisebox{0.75ex}.}\mathclap{.}}
  \def\ddots{\mathclap{\raisebox{1ex}{.}\raisebox{0.5ex}{.}\raisebox{-0.ex}{.}}}
  \def\cdots{\mathclap{\cdot\!\cdot\!\cdot}}
  \text{\footnotesize$\begin{bmatrix}
    1 \!&\! 0 \!&\! \cdots \!&\! 0 \\
    0 \!&\! \,\ddots \!&\!  \!&\! \vdots \\
    \vdots \!&\!  \!&\! 1 \!&\! 0 \\
    0 \!&\! \cdots \!&\! 0 \!&\! \!\!\; -1
  \end{bmatrix}$}\!\!\:,
\end{align}%
\end{small}~\\[-2.5ex]%
where $\mathrm C^{k{-}1}\!\!\:Z \in \mathrm{U}(2^k)$
.
\end{subequations}%
Furthermore, as it is \emph{a priori} unlikely that we can obtain a calculus in which all rules are free from scalar gadgets, we wish to retain the ability (at least in the ZH calculus) to directly express arbitrary scalars, as in Eqn.~\eqref{eqn:ZH-scalar-box}.
We require that $\sem{\,\cdot\,}_\nu$ be able to do the same, so that
\vspace*{-1ex}
\begin{equation}
  \label{eqn:nu-scalar-box}
  \Bigsem{\input{./Figures-Source/ZH-scalar-box.tikz}}_{\!\nu}
  =\;
    a.
\end{equation}~\\[-3ex]
These equations impose the following constraints on the model $\sem{\,\cdot\,}_\nu$ (as we prove in Appendix~\ref{apx:denotationalConstraints}):~\\[-2.5ex]

\begin{lemma}
  \label{lemma:denotationalConstraints}
  Eqns.~\eqref{eqn:parametersModelNu}--\eqref{eqn:nu-scalar-box} hold iff $u_2 \!\!\:=\!\!\: v_2 \!\!\:=\!\!\: \xi \!\!\:=\!\!\: 1$, 
  $u_3 \!\!\:=\!\!\: v_3 \!\!\:=\!\!\: g_3^{-1} \!\!\:=\!\!\: 2^{1\!\!\:/\!\!\;4}$,
  and 
  $h_k \!\!\:=\!\!\: 2^{-k\!\!\;/\!\!\;4}$ for all $k \ge 0$.
\end{lemma}

\vspace*{-1ex}
\subsection{Asserting favourable rewrite rules}
\label{sec:normalisationConstraints}
\vspace*{-.5ex}

Subject to the above, we intend for $\sem{\,\cdot\,}_\nu$ to be a model for calculi whose most important rewrites are free of scalar gadgets. 
As $\sem{\,\cdot\,}_\nu$ is as yet under-determined, we may define it so that certain ``idealised'' rewrites are sound.
In Appendices~\ref{apx:compatibilityRewritesZX} and~\ref{apx:compatibilityRewritesZH}, we characterise normalisation constraints which would be imposed by various rewrites, emphasising those which are necessary for the green and red nodes for the ZX calculus to form either special dagger-Frobenius algebras or bialgebras.
In Appendix~\ref{apx:constructingWellTemperedCalculi}, we draw up a ``wish-list'' of such rewrites and introduce the soundness of each one in turn as constraints (subject to consistency with those that came before) to fix a single model $\sem{\,\cdot\,}_\nu$, presented in Eqns.~\eqref{eqn:modelNu} on page~\pageref{eqn:modelNu}.
%
%
Below, we summarise the result of imposing these successive constraints (with proofs in Appendix~\ref{apx:constructingWellTemperedCalculi}):

\smallskip

\begin{corollary}
  \label{cor:rewritesFromNotation}
  If Eqns.~\eqref{eqn:parametersModelNu}--\eqref{eqn:nu-scalar-box} hold, then the rewrites~\Rule(Id$_Z$) and \Rule(Id$_X$) on page~\pageref{newZX} and the rewrites~\Rule(Id$_Z$), \Rule(Id$_H$), \Rule(Not), \Rule(Bialg$_{ZX}$), \Rule(Mult$_H$), and \Rule(Fuse$_H$) on page~\pageref{newZH} are sound.
\end{corollary}~\\[-5ex]

\smallskip

\begin{lemma}
  \label{lemma:summaryChangeSound}
  Eqns.~\eqref{eqn:parametersModelNu} and the two rewrites \Rule(Change) on pages~\pageref{newZX}--\pageref{newZH} are sound, iff $u_k \!\!\:=\!\!\: v_k$ and $g_k \!\!\:=\!\!\: 2h_2^k u_k$ for all $k \ge 0$.
  In particular, Eqns.~\eqref{eqn:parametersModelNu} and~\eqref{eqn:unitaryNOTandH} hold and the two rewrites \Rule(Change) are sound, iff $h_2 \!\!\:=\!\!\: 2^{-1\!\!\:/\!\!\;2}$, $\xi \!\!\:=\!\!\: 1$, and $u_k \!\!\:=\!\!\: v_k \!\!\:=\!\!\: 2^{(k-2)/\!\!\;2} g_k$ for all $k \ge 0$.
\end{lemma}~\\[-5ex]

\smallskip

\begin{lemma}
  \label{lemma:establishNuModel}
  Let $\nu = 2^{-1\!\!\:/\!\!\;4}$. Then
  Eqns.~\eqref{eqn:parametersModelNu}--\eqref{eqn:nu-scalar-box} hold and the rewrites \Rule(Change), \Rule(Fuse$_Z$), and \Rule(Bialg$_{ZX}$) on pages~\pageref{newZX}--\pageref{newZH} are sound, iff 
  $\xi =  1$, $u_k  =  v_k  =  g_k^{-1} = \nu^{-(k{-}2)}$ for all $k \ge 0$, and $h_k = \nu^k$ for all $k \ge 0$.
  Under these conditions, the alternative ZH rewrites of page~\pageref{newZH} are sound.
\end{lemma}~\\[-5ex]

\smallskip

\noindent
Note that $u_1 = 2^{-1\!\!\:/\!\!\;4}$ under the conditions of the second part of Lemmas~\ref{lemma:summaryChangeSound} and~\ref{lemma:establishNuModel}.
Among other things, this implies that $\smash{\sem{\,\begin{gathered}\input{./Figures-Source/ZX-red-prep.tikz}\end{gathered}\,}_{\nu}} = 2^{1\!\!\:/\!\!\;4} \ket{\zsymb}$, so that we must do \emph{some} bookkeeping of scalars to represent unit vectors.
As we show in Corollary~\ref{cor:dontNormaliseBasisStates} (Appendix~\ref{apx:compatibilityRewritesZX}), the diagram \smash{$\sem{\,\begin{gathered}\input{./Figures-Source/ZX-red-prep.tikz}\end{gathered}\,}$} having non-unit norm is a necessary compromise for any version of the ZX calculus with what one might consider to be ``reasonable'' rewrites.
To mitigate the bookkeeping of scalars that we then require for single-qubit states and projections, we introduce the nu-box generator to represent powers of $\nu = 2^{-1\!\!\:/\!\!\;4}$, so that we may easily represent $\ket{\zsymb} = \bigsem{\;\input{./Figures-Source/ZX-nu-red-prep.tikz}\,}_\nu$ as illustrated in Eqn.~\eqref{eqn:newCalculiStates}.
The rewrites involving the nu-box are then motivated by the following Lemma (proven in Appendix~\ref{apx:constructingWellTemperedCalculi}):

\smallskip

\begin{lemma}
  \label{lemma:greenDotScalars}
  Define $\bigsem{\,\input{./Figures-Source/ZX-nu-box-k.tikz}}_\nu = \nu^k$. 
  Under the conditions of Lemma~\ref{lemma:establishNuModel}, for $\theta$ not an odd multiple of $\pi$,
  \begin{equation}
    \bigsem{\,\input{./Figures-Source/ZX-green-phase-dot-arity-0-small.tikz}\,}_\nu
    \,=\;
    \sqrt{1+\cos(\theta)}\,\e^{i \!\; \theta\!\!\!\:/2}
    \;=\;
    \Bigsem{\,\input{./Figures-Source/ZX-green-dot-nu-gadget-small.tikz}}_{\nu},
  \end{equation}~\\[-2.5ex]%
  where $\lambda = \log_2\bigl(\sec^2(\theta\!\!\:/2)\bigr) - 1$.
  Under these conditions, the alternative ZX rewrites of page~\pageref{newZX} are sound.
\end{lemma}%

\smallskip

\noindent
By introducing the rule \Rule(Scale$_\nu$) in the well-tempered ZX calculus, we provide a means to convert representations of scalar factors arising from isolated green dots or red dots to nu-boxes, representing them (somewhat) more easily using the rules \Rule(Fuse$_\nu$) and \Rule(Id$_\nu$).
The question of the soundness and completeness of the version of the ZX calculus on page~\pageref{newZX} may then be reduced to that of the sub-theory of green, red, and Hadamard nodes.

\vspace*{-1.5ex}
\subsection{Completeness of the new calculi}

The rewrites of the ZH~calculus variant on page~\pageref{newZH} are in clear one-to-one correspondence with those in Figure~\ref{fig:variantZHrules} (which are equivalent to those of Ref.~\cite{BK-2019}).
As we construct $\sem{\,\cdot\,}_\nu$ by reduction to the original standard model $\sem{\,\cdot\,}_\beta$, the completeness of the calculus of page~\pageref{newZH} for $\sem{\,\cdot\,}_\nu$ is underwritten by that of the original presentation of the ZH calculus.
Similarly, for the ZX~calculus variant on page~\pageref{newZX},  \Rule(Bialg) is equivalent to the simplified versions of \Rule(B$^\alpha$) and \Rule(C$_r^\alpha$) of Figure~\ref{fig:usualZXrules}, conditioned on the two rules \Rule(F$_g^\alpha$) and \Rule(X$^\alpha$) also holding.
As $\sem{\,\cdot\,}_\nu$ is constructed by reduction to $\sem{\,\cdot\,}_\alpha$ and supports those rules, the completeness of the well-tempered ZX calculus for $\sem{\,\cdot\,}_\nu$ also follows (given the remarks above regarding $\nu$-boxes).

\section{Features of the well-tempered calculi}
\label{sec:features}
\vspace*{-1.5ex}

We now remark on a number of features of the ZX calculus on page~\pageref{newZX} (and also some features of the ZH calculus on page~\pageref{newZH}) which we consider practically important.

\vspace*{-2ex}
\paragraph{Hopf Law, local complementation, and phase gadgets --- \!\!\!}
One of the objectives of developing the well-tempered ZX calculus is to simplify some of the most important known theorems and gadgets of the ZX calculus.
A particular rule which holds exactly in our well-tempered ZX calculus is  the Hopf Law.%
  \footnote{%
    The ``Hopf Law'' is a property of bialgebras (not specific to the well-tempered ZX calculus), which here amounts to the first and last diagrams of Eqn.~\eqref{eqn:HopfLawDerivation} being equivalent.
    By an abuse of terminology, we refer below to the ``Hopf Law'' when the red and green dot may have zero or multiple free edges (these are easy corollaries, provable with \Rule(Fuse$_X$) and \Rule(Fuse$_Z$) rewrites).
  }
Using an easy theorem \Rule(Fuse$_X$), corresponding to \Rule(Fuse$_Z$) with green nodes exchanged for red nodes and provable using the \Rule(Change) rewrite, we may prove the Hopf law following Ref.~\cite[Example\;2.5]{CD-2011}: 
\vspace*{-1ex}
\begin{gather}
  \label{eqn:HopfLawDerivation}
  \begin{gathered}
  \input{./Figures-Source/ZX-derive-Hopf.tikz}
  \end{gathered}
  \;\;.
  \\[-6ex]\notag
\end{gather}
Another result which simplifies is the transformation of graph states under local Clifford operations.
Following Ref.~\cite{DP-2009}, and using as a Lemma the rewrite \Rule(EU) of  Jeandel~\emph{et al.}~\cite[Fig.\;1]{JPVW-2017} for the model $\sem{\,\cdot\,}_\alpha$ (which one may verify is also sound in $\sem{\,\cdot\,}_\nu$), one may show that a scalar-exact version of Van~den~Nest's theorem holds for $\sem{\,\cdot\,}_\nu$\,:
for instance,
\vspace*{-2ex}
\begin{equation}{}
  \mspace{-18mu}
  \input{./Figures-Source/ZX-local-compl.tikz} \;\; .
  \mspace{-12mu}
\end{equation}~\\[-3.5ex]
This derivation features a \emph{phase gadget}~\cite{KW-2019,BBW-2019}, a ZX term of independent interest which denotes a diagonal operation inducing a relative phase $\e^{i\theta}$ on standard basis states $\ket{x}$ for which $x \cdot z \equiv 1 ~\textup{(mod 2)}$ for some $z \in \{0,1\}^n$.
One may show that the gadgets are precisely unitary in $\sem{\,\cdot\,}_\nu$ by induction on $n \ge 1$:
\vspace*{-3ex}
\begin{equation}
  \input{./Figures-Source/ZX-phase-gadget-conj.tikz} \;\;.
\end{equation}~\\[-6ex]

\vspace*{-2ex}
\paragraph{Scaled specialness and supplementarity --- \!\!\!}%
\label{discn:specialAndSupplementarity}
A compromise which is made in the well-tempered ZX and ZH calculi is that some well-known (but less-often used) results which are free of scalar gadgets in the pre-existing calculi, now do involve scalar gadgets.
The most noteworthy of these is the rule corresponding to \Rule(S$_g^\alpha$) in Figure~\ref{fig:usualZXrules} and \Rule(S$_w^\beta$) in Figure~\ref{fig:variantZHrules} (\emph{c.f.}~Rule~\Rule(Spec$_Z$) on page~\pageref{newZH}):
\vspace*{-2ex}
\begin{equation}
  \input{./Figures-Source/ZX-derive-quasi-special.tikz} \;\;.
\end{equation}~\\[-6.5ex]

\noindent
Another noteworthy rewrite which does \emph{not} hold in a simplified form in the ZX calculus of page~\pageref{newZX}, despite holding in a simple form in the pre-existing ZX calculi, is the ``supplementarity'' rewrite \cite[Lemma\;1]{PW-2016}:
\vspace*{-.5ex}
\begin{equation}{}
\mspace{-12mu}
  \Sem{5ex}{ \input{./Figures-Source/ZX-suppl-left.tikz} }_{\!\nu}
  \!\! =\,  
  \nu \Sem{5ex}{ \input{./Figures-Source/ZX-suppl-left.tikz} }_{\!\alpha}  
  \!\!\! =\,  
  \nu \Sem{5ex}{ \input{./Figures-Source/ZX-suppl-right.tikz} }_{\!\alpha}  
  \!\!\!=\,  
  \nu^{-2} \Sem{5ex}{ \input{./Figures-Source/ZX-suppl-right.tikz} }_{\!\nu}  
  \!\! =\,  
  \Sem{5ex}{ \input{./Figures-Source/ZX-suppl-right-w-dot.tikz}\, }_{\!\nu}
  \!\! =\,  
  \Sem{5ex}{\!\!\input{./Figures-Source/ZX-suppl-right-pun.tikz}\;\;\;}_{\!\nu}\!\!,
\end{equation}~\\[-1.5ex]
where the second equality is the usual supplementarity rewrite in pre-existing scalar-exact versions of ZX, the first and third equalities follow from the normalisation of the generators in $\sem{\,\cdot\,}_\nu$,  the penultimate equality follows from Lemma~\ref{lemma:greenDotScalars}, and the final equality is again the Hopf Law.

\vspace*{-1ex}
\paragraph{Upgrading ZX congruences to ZX equalities --- \!\!\!}
As we note in Section~\ref{sec:scalarFactors}, it is a common practise to perform calculations in the ZX calculus by disregarding scalar factors altogether, only reasoning about congruency of diagrams up to a non-zero proportionality factor.
However, the only ``traditional'' ZX rewrites whose corresponding versions introduce scalar gadgets in the well-tempered calculus are \Rule(Spec$_Z$) and \Rule(Euler).
Thus, any calculation of ZX terms which neglected scalar factors, but could be expressed without either of these rewrites, can without modification now be regarded as a scalar-exact derivation with respect to the model $\sem{\,\cdot\,}_\nu$.
(Those which involve involve \Rule(Euler) but not \Rule(Spec$_Z$) will be correct up to a global phase factor.)

\vspace*{-2ex}
\paragraph{Classification as Z\textsuperscript{*} calculi --- \!\!\!}
\label{sec:classificationZ*}

As we note on page~\pageref{discn:introZ*}, Carrette and Jeandel~\cite{CJ-2020} describe a classification of ``\!\;$\mathrm{Z^\ast}$ calculi'', consisting of bialgebras involving the monoid of the ``Z algebra'' (the green dots of the ZX calculus and white dots of the ZH calculus) with a co-monoid of some other Frobenius algebra (in our case, the red dots and yellow boxes respectively).
Following Ref.~\cite[Appendices\;B.2.1\,\&\,B.3.3]{CJ-2020}, the bialgebras of the well-tempered ZX and ZH calculi may be characterised as follows:
\begin{itemize}
\item
  The bialgebra of the green and red dots of the well-tempered ZX calculus, is equivalent to the calculus $\mathrm{Z}^{\smash{(\sqrt 2,1)}} \mathrm{X}_{\smash{(\sqrt 2,1)}}$, by an isomorphism $\Lambda = 2^{1\!\!\:/\!\!\;4} \cdot \idop$ applied to the outputs of operations (and applying $\Lambda^{-1}$ to the inputs);
\item
  The bialgebra of the white dots and H-boxes in the well-tempered ZH calculus is equivalent to the calculus $\mathrm{Z}^{\smash{(1,1)}} \mathrm{X}_{\smash{(\sqrt 2,-1\!\!\:/\!\!\;2)}}$ --- albeit with a different representation for the phases of H-boxes, as with the original ZH calculus~\cite[p.\,11]{CJ-2020} --- via the isomorphism $\Lambda^{-1}$ applied to the outputs of operations (and applying $\Lambda$ to the inputs).
\end{itemize}

\vspace*{-2ex}
\section{Concluding remarks}
\label{sec:conclusions}
\vspace*{-1.5ex}

Our aim in this article was to present versions of the ZX and ZH calculi --- equivalent in denotation up to scalar factors to the existing versions --- which supported simpler representations of unitary transformations, and simpler versions of the most commonly used rewrites.
This makes the ZX and ZH calculi more practcal as a tool for routine scalar-exact calculation.
While efficient reasoning about quantum processes is part of the  research programme of diagrammatic calculi for quantum computation, for historical reasons the aim of doing so with precision with scalars has often been considered a second-order priority, which can be done \emph{post-hoc} if strictly necessary.
Our results demonstrate how, by a careful choice of notation, one can maintain scalar exactness at every step without too much trouble in practise.

At this point, we confess that the inclusion of nu-boxes in our version of the ZX calculus is only half-seriously intended.
It supports the objective of maintaining scalar exactness as a matter of routine, while supporting the conventional design of the ZX calculus of all parameters being drawn from an \emph{additive} group.
Nevertheless, as a representation even of just positive scalar factors, nu-boxes leave much to be desired.
A more transparent way to represent scalar factors would be to treat the diagrams as an algebra over $\C$, and simply multiply diagrams by complex scalars to modify the normalisation.
The ZH calculus provides a compromise between these two positions, allowing any global scalar factor to be taken as a parameter of an H-box.
We would advocate a further revision of the ZX calculus presented on page~\pageref{newZX}, to include more direct and flexible means of representing scalar factors along these lines.

Taking things one step further: $\sem{\,\cdot\,}_\nu$ is designed to allow the ZX and ZH calculi to ``interoperate'', \emph{e.g.},~by using hybrid diagrams consisting of compositions of ZX and ZH generators. 
We may thus consider an approach to using graphical calculi (\emph{e.g.},~a~practise of developing and using a growing corpus of gadgets and theorems), in which one may adopt the same indifference as to which specific sound-and-complete scalar-exact calculus forms the basis of one's calculation, as one commonly adopts towards set-theoretic foundations or towards constructions of the real numbers from the rationals.

\subsection*{Acknowledgements}

This work was supported by a Fellowship funded by a gift from Tencent Holdings (tencent.com).
I~am grateful to Hector Miller-Bakewell, Bob Coecke, and the anonymous referees for discussions and remarks which informed the presentation of this work.
I would also like to thank Quanlong Wang, Cole Comfort, Sean Tull, Aleks Kissinger, and Titouan Carrette for helpful technical discussions.


\providecommand{\urlalt}[2]{\href{#1}{#2}}
\providecommand{\doi}[1]{doi:\urlalt{http://dx.doi.org/#1}{#1}}


\appendix
\numberwithin{theorem}{section}
\renewcommand\thetheorem{\thesection.\arabic{theorem}}
\renewcommand\thelemma{\thesection.\arabic{lemma}}
\renewcommand\theproposition{\thesection.\arabic{proposition}}
\renewcommand\theproposition{\thesection.\arabic{equation}}


\vspace*{-1ex}
\section{On normalisation in the ZX and ZH calculi}
\label{apx:normalisationZXandZH}
\vspace*{-.5ex}

We now expand on remarks
Section~\ref{sec:scalarFactors} on page \pageref{discn:conventionalNormalisation} about the normalisation conventions in existing presentations of the ZX and ZH calculi.

\vspace*{-1ex}
\subsection{On scalar factors in the ZX calculus}
\vspace*{-.5ex}
 
Work on the ZX calculus is concerned with revising and refine the original presentation of Ref.~\cite{CD-2011} to provide more-or-less minimal, complete, and scalar-exact rewrite systems (so thatit represents an alternative to computing with matrices which is sound for any application involving tensors over $\C$ in which every index has dimension $2$).
As a result, much of the work on the ZX~calculus~\cite{Backens-2015,PW-2016,BPW-2017,JPVW-2017,CK-2017,JPV-2017,NW-2017,NW-2018,Vilmart-2019-minimal,Vilmart-2019-ToffoliH,JPV-2019,Wang-2019,Wang-2019-semirings} use precisely the same standard model --- \ie,~the model $\sem{\,\cdot\,}_\alpha$ presented in Eqns.~\eqref{eqn:origStandardModel}, which is the simplest possible refinement of the standard model from Ref.~\cite{CD-2011}.

The model $\sem{\,\cdot\,}_\alpha$ has the unfortunate feature that a user who is interested in tracking normalisation must frequently be on guard against the introduction or cancellation of scalar factors.
The rules \Rule(B$^\alpha$), \Rule(C$_r^\alpha$), \Rule(E$^\alpha$), and \Rule(EU$^\alpha$) of Figure~\ref{fig:usualZXrules} (on page~\pageref{fig:usualZXrules}) demonstrate the issue of accumulation or cancellation of scalars in the existing scalar-exact presentations of the ZX calculus.
Each of these rules involve small gadgets of phase-free nodes, which to the initiated represent scalars of $2^{1\!\!\:/2}$ or $2^{-1\!\!\:/2}$.
Furthermore, while the two-node gadget of \Rule(EU$^\alpha$) with non-trivial phases is necessary to represent a global phase factor, it necessitates the inclusion of the other phase-free gadget.%
  \footnote{%
  It is also worth remarking that the non-phase-free gadget of \Rule(EU$^\alpha$) is a somewhat indirect representation of the scalar $\e^{i\gamma}$\!\:.
}
The scalars in these rewrite rules reflect the fact that the original presentation of the ZX calculus~\cite[Fig.\;1]{CD-2011} prioritises the role of the green and red dots each as special commutative dagger-Frobenius algebras,%
  \footnote{%
    For the green nodes, this follows from the principle that only topology matters, together with the rules \Rule(I$_g$) and \Rule(S$_g^\alpha$), and using the rule \Rule(F$_{\!g}^\alpha$) to define the higher arity nodes.
    For the red nodes, versions of \Rule(S$_g^\alpha$) and \Rule(F$_{\!g}^\alpha$) can be derived by using \Rule(X$^\alpha$).%
  }
so that together they must form a scaled bialgebra.
\begin{subequations}%
\label{eqn:normalisationFailures}%
Grounding the ZX calculus on such a model has the consequence that the simplest representation of the standard basis states is only up to a supernormalised scalar factor,
\vspace*{-1ex}
\begin{align}{}
\mspace{-48mu}
  \bigsem{\input{./Figures-Source/ZX-red-prep.tikz}}_\alpha
  \!&=\,
    \sqrt{2}\, \ket{\zsymb}\!,
&
  \bigsem{\input{./Figures-Source/ZX-red-pi-prep.tikz}}_\alpha
  \!&=\,
    \sqrt{2}\, \ket{\osymb}\!,
\mspace{-42mu}
\end{align}~\\[-3ex]%
while the controlled-NOT operation is represented only up to a sub-normalised scalar factor:
\vspace*{-1ex}
\begin{equation}
\label{eqn:scaledCNOT-ZX}
  \biggsem{\input{./Figures-Source/ZX-CNOT-gadget.tikz}}_\alpha
  \!=\,
  \tfrac{1}{\sqrt 2}
  \mathrm{CNOT}\,.
\end{equation}%
\end{subequations}~\\[-3ex]
One might suppose (uncharitably) that while lacking normalisation in one of these cases may be regarded as a misfortune, to lack both looks like carelessness.
Such a criticism would be unfair, for a few reasons.

Suppose that, we consider a parameterised model $\sem{\,\cdot\,}$ for the generators of the ZX calculus, similarly to Eqn.~\eqref{eqn:parametersModelNu} on page~\pageref{eqn:parametersModelNu}, satisfies the ``only the topology matters'' meta-rule, and which for good measure fixes the interpretation of the Hadamard box so that it is unitary:
\vspace*{-1ex}
  \begin{small}%
  \begin{align}{}
  \label{eqn:simpleParameterisedZXmodel}%
  \mspace{-48mu}
  \begin{split}
  \Biggsem{\!\!\!\input{./Figures-Source/ZX-green-phase-dot-arity.tikz}\!\!\!}
  \,&=\,
    u_{m{+}n}\, \Bigl( \ket{\zsymb}^{\!\otimes n}\!\bra{\zsymb}^{\!\!\;\otimes m}
    \,+\;
    \e^{i\theta}
    \ket{\osymb}^{\!\otimes n}\!\bra{\osymb}^{\!\!\;\otimes m} \Bigr),
  \mspace{-18mu}
%
%
\\[.5ex]
%
%
  \Biggsem{\!\!\!\input{./Figures-Source/ZX-red-phase-dot-arity.tikz}\!\!\!}
  \,&=\,
    u_{m{+}n}\, \Bigl( \ket{\psymb}^{\!\otimes n}\!\bra{\psymb}^{\!\!\;\otimes m}
    \,+\;
    \e^{i\theta}
    \ket{\msymb}^{\!\otimes n}\!\bra{\msymb}^{\!\!\;\otimes m} \Bigr),
  \mspace{-18mu}
\end{split}
%
%
&
%
%
  \Bigsem{\,\input{./Figures-Source/ZX-H-box.tikz}\,}
  \,&=\,
  \tfrac{1}{\sqrt 2} \text{\small$\begin{bmatrix}
    1 \!&\! \phantom-1 \\ 1 \!&\! -1
  \end{bmatrix}$}
  .
  \mspace{-18mu}
  \end{align}%
  \end{small}~\\[-2ex]%
\label{discn:dontNormaliseStates}%
If we then impose the constraint  
$\bigsem{\,\input{./Figures-Source/ZX-red-prep.tikz}\,} = \ket{\zsymb}$, then as we show in
Corollary~\ref{cor:dontNormaliseBasisStates} on page~\pageref{cor:dontNormaliseBasisStates}),
it is impossible for the green and red nodes to form either special commutative dagger-Frobenius algebras (without scaling) or bialgebras (without scaling).
Worse, at least one of the rewrites in each column of the following rewrites would require an additional scalar gadget to be sound:
\vspace*{-.5ex}
\begin{align}
  \begin{aligned}
 \begin{aligned}
    \begin{tikzpicture}
      \node (Z) at (0,0) [Z dot] {};
      \node (Z') at (0.5,0) [Z dot] {};
      \draw (Z) -- ++(-.3125,0);
      \draw (Z') -- ++(+.3125,0);
      \draw [out=45, in=135] (Z) to (Z');
      \draw [out=-45, in=-135] (Z) to (Z');
    \end{tikzpicture}
  \end{aligned}
  \,&\longleftrightarrow\,
  \begin{aligned}
    \begin{tikzpicture}
      \node (Z) at (0,0) {\color{white}$\theta$};
      \coordinate (Z) at (0,0);
      \draw (Z) -- ++(-.375,0);
      \draw (Z) -- ++(+.375,0);
    \end{tikzpicture}
  \end{aligned}
\qquad&\qquad
  \begin{aligned}
    \begin{tikzpicture}
      \node (X) at (.3125,.25) [X dot] {};
      \node (X') at (.3125,-.25) [X dot] {};
      \draw (X) -- ++(0.4125,0);
      \draw (X') -- ++(0.4125,0);
      \node (Z) at (-.3125,.25) [Z dot] {};
      \node (Z') at (-.3125,-.25) [Z dot] {};
      \draw (Z) -- ++(-0.4125,0);
      \draw (Z') -- ++(-0.4125,0);
      \draw [out=30,in=150] (Z) to (X);
      \draw [out=-30,in=-150] (Z') to (X');
      \draw (Z) to (X');
      \draw (X) to (Z');
    \end{tikzpicture}
  \end{aligned}
  \,&\longleftrightarrow\,
  \begin{aligned}
    \begin{tikzpicture}
      \node (X) at (-.5,0) [X dot] {};
      \node (Z) at (0,0) [Z dot] {};
      \draw (X) -- (Z); 
      \draw [out=-45,in=180] (Z) to ++(0.5,-0.25);
      \draw [out=45,in=180] (Z) to ++(0.5,0.25);
      \draw [out=-135,in=0] (X) to ++(-0.5,-0.25);
      \draw [out=135,in=0] (X) to ++(-0.5,0.25);
    \end{tikzpicture}
  \end{aligned}
\\[1.5ex]
    \begin{aligned}
    \begin{tikzpicture}
      \node (Z) at (0,0) [Z dot] {};
      \draw (Z) -- ++(-0.5,0); 
      \draw [out=-45,in=180] (Z) to ++(0.5,-0.25)
        node (epsilon) [Z dot] {};
      \draw [out=45,in=180] (Z) to ++(0.5625,0.25);
    \end{tikzpicture}
  \end{aligned}
  \;&\longleftrightarrow\;
  \begin{aligned}
    \begin{tikzpicture}
      \draw (0,0) -- ++(0.875,0) node [draw=none] {};
    \end{tikzpicture}
  \end{aligned}
&
  \begin{aligned}
    \begin{tikzpicture}
      \node (X) at (-.5,0) [X dot] {};
      \node (Z) at (0,0) [Z dot] {};
      \draw (X) -- (Z); 
      \draw [out=-45,in=180] (Z) to ++(0.5,-0.25);
      \draw [out=45,in=180] (Z) to ++(0.5,0.25);
    \end{tikzpicture}
  \end{aligned}
  \,&\longleftrightarrow\,
  \begin{aligned}
    \begin{tikzpicture}
      \node (X) at (-.1875,.1875) [X dot] {};
      \node (X') at (-.1875,-.1875) [X dot] {};
      \draw (X) -- ++(0.625,0);
      \draw (X') -- ++(0.625,0);
    \end{tikzpicture}
  \end{aligned}
\end{aligned}
\end{align}%
This would involve adjustments to the normalisation with many of the rewrites, and could be considered a steep price to pay for a single-node representation of standard basis states, compared to correcting for the scalar factor involved with introducing a fresh qubit (or projecting a qubit onto some single-qubit state).
On these grounds, it seems reasonable to abandon the goal of having single-node diagrams to represent the states $\ket{\zsymb}$ or $\ket{\osymb}$.

It remains to consider the normalisation of the left-hand side of Eqn.~\eqref{eqn:scaledCNOT-ZX}.
Subject to the parameterised model of Eqns.~\eqref{eqn:simpleParameterisedZXmodel}, one may show that Eqn.~\eqref{eqn:scaledCNOT-ZX} necessarily holds if the green and red nodes form special commutative dagger-Frobenius algebras (this is an easy corollary of Lemma~\ref{lemma:coeffConstraintsSpecial}, on page~\pageref{lemma:coeffConstraintsSpecial}).
Thus, for a reasonably defined ZX calculus, one is confronted with the choice either to fix the left-hand side of Eqn.~\eqref{eqn:scaledCNOT-ZX} to be unitary, or to define the normalisation so that the classical structures in each basis are as nicely behaved as possible.
As the latter choice is compatible with setting $u_k = 1$ for all $k \ge 0$, it is reasonable and unsurprising that this should be the first normalisation convention chosen.

Remarkably, there is also a strong \emph{post-hoc} justification for the normalisation of Eqn.~\eqref{eqn:origStandardModel}, in that it corresponds precisely to the description of surface code lattice surgery in terms of CPTP maps~\cite{BH-2020}.
The sub-normalisation described in Eqn.~\eqref{eqn:scaledCNOT-ZX} precisely reflects the probability of $\tfrac{1}{2}$ of performing a CNOT by lattice surgery (using the simplest construction) without requiring a Pauli frame shift.

It is clear that the usual normalisation of ZX diagrams have both theoretical and practical justifications.
However, for the application of reasoning about general quantum procedures in terms of unitary circuits, it remains the case that this normalisation requires the user to be alert to changes in the normalisation arising from rewrites, as well as from the lack of normalisation seen in Eqns.~\eqref{eqn:normalisationFailures}.
This motivates our investigation into a renormalised version of the ZX calculus.

\vspace*{-1ex}
\subsection{On scalar factors in the ZH calculus}
\vspace*{-.5ex}

While the ZH~calculus was designed as a part of the same research programme as the ZX~calculus, it was designed with different priorities.
Rather than an emphasis on strong complementarity which  features in the study of the ZX~calculus, a main motivation in the development of the ZH~calculus is to represent quantum circuits specifically over the gate-set Hadamard+Toffoli~\cite{Shi-2003,Aharonov-2003}, and other closely related gate-sets.
This motivates an interest in whether the following two gadgets represent unitary operators:
\vspace*{-.5ex}
\begin{equation}
\label{eqn:mainZHgadgets}
    \begin{gathered}
      \input{./Figures-Source/ZH-H-box.tikz}  
    \end{gathered}
    \;\; ,
    \qquad\qquad\qquad
    \begin{gathered}
      \input{./Figures-Source/ZH-CkZ-gadget.tikz}  
    \end{gathered}
  \;.
\end{equation}~\\[-3ex]
The first of these is proportional to the single-qubit Hadamard gate, and the second to a $(k{-}1)$-controlled $Z$ operation (\ie,~a~$Z$ operation which is coherently controlled on $k{-}1$ other qubits).

The white dots play a similar role in the ZH~calculus as in the ZX~calculus, and form a special commutative dagger-Frobenius algebra.
Suppose that we take the usual interpretation of these nodes, and also that our interpretation of the H boxes depends only on the topology.
Then it is not difficult to see that at most one of the two gadgets of Eqn.~\eqref{eqn:mainZHgadgets} will denote a unitary for $k=1$.
The simple normalisation convention of Ref.~\cite{BK-2019} suffices for the right-hand diagram to be unitary for all $k \ge 0$, which is particularly reasonable for representing unitary circuits which may have highly-controlled phase operations or highly-controlled NOT gates.

The normalisation on the standard model of the ZH calculus is also notable in that it consists of integer matrices, which means in principle that it is potentially directly useful in the analysis of counting complexity and gap-complexity~\cite{DHM-2002,BKM-2020}.
Thus, the normalisation of the standard model of the ZH calculus provides it with a potential for versatility which is worth bearing in mind.

A less desirable consequence of the choice of normalisation in Ref.~\cite{BK-2019} is that rewrite rules (such as those in Figure~\ref{fig:variantZHrules} on page~\pageref{fig:variantZHrules}) must frequently account for contributions of factors of $2$ (or $\tfrac{1}{2}$) to the normalisation.
In particular, such factors of $2$ are required to cancel pairs of Hadamard gates in rule~\smash{\Rule(I$_h^\beta$)}, in the relationships between white and grey dots in rules~\smash{\Rule(X$^\beta$)} and~\smash{\Rule(N$^\beta$)}, and to fuse two H boxes in rule~\smash{\Rule(F$_h^\beta$)}.
Note that the contribution of $2^{n-1}$ involved in \smash{\Rule(BA$_2^\beta$)} is a consequence of our alternative representation, rather than being original to Ref.~\cite{BK-2019} ---
however, this is only avoided in the original presentation through a hybrid-bialgebra relation between the white dots, H-boxes, and the gray dots (which there are ``derived generators'' also involving a factor of $2$).

Again, as with the ZX calculus, it is clear that the original normalisation of ZH diagrams may be justified on theoretical and practical grounds.
However, the application of reasoning about general quantum processes motivates an investigation into a renormalised version of the ZH calculus as well.


\vspace*{-1ex}
\section{Proof of the denotational constraints on the model $\sem{\,\cdot\,}_\nu$}
\label{apx:denotationalConstraints}
\vspace*{-1.5ex}

\begin{lemma}[\emph{c.f.} Lemma~\ref{lemma:denotationalConstraints}]%
  \label{lemma:denotationalConstraintsRedux}%
  Eqns.~\eqref{eqn:parametersModelNu}--\eqref{eqn:nu-scalar-box} hold iff $u_2 \!\!\:=\!\!\: v_2 \!\!\:=\!\!\: \xi \!\!\:=\!\!\: 1$, 
  $u_3 \!\!\:=\!\!\: v_3 \!\!\:=\!\!\: g_3^{-1} \!\!\:=\!\!\: 2^{1\!\!\:/\!\!\;4}$,
  and 
  $h_k \!\!\:=\!\!\: 2^{-k\!\!\;/\!\!\;4}$ for all $k \ge 0$.
\end{lemma}
\vspace*{-1ex}

\begin{proof}
  From Eqns.~\eqref{eqn:unitaryId} and~\eqref{eqn:unitaryNOTandH} we have get $u_2 \!\!\;=\!\!\; v_2 \!\!\;=\!\!\; \xi \!\!\;=\!\!\; 1$ and $h_2 = 2^{-1\!\!\:/\!\!\;2}$, and from Eqn.~\eqref{eqn:nu-scalar-box} we obtain $h_0 = 1$.
  From Eqn.~\eqref{eqn:unitaryCkZ} for the special case $k=1$, we may show by reduction to  $\sem{\,\cdot\,}_\alpha$ that
  \vspace*{-1.25ex}
  \begin{equation}
  \begin{aligned}[b]
    \Sem{6.25ex}{\,\input{./Figures-Source/ZH-CZ-gadget.tikz}\,}_{\!\nu}\!
  =\,
    \Sem{6.25ex}{\,\input{./Figures-Source/ZX-CZ-gadget.tikz}\,}_{\!\nu}\!
  &=\,
    \sqrt{2} h_2 u_3^2
    \Sem{6.25ex}{\,\input{./Figures-Source/ZX-CZ-gadget.tikz}\,}_{\!\alpha}\!
  =\,
   h_2 u_3^2 \, \mathrm{C}\!\!\;Z
  \,,
  \end{aligned}
  \end{equation}~\\[-1.5ex]
  from which it follows that $u_3 = 2^{1\!\!\:/\!\!\;4}$.
  Similarly, we may show from Eqn.~\eqref{eqn:unitaryCNOT} that
  \vspace*{-1ex}
  \begin{equation}
  \begin{aligned}[b]
    \biggsem{\,\input{./Figures-Source/ZX-CNOT-gadget.tikz}\,}_{\!\nu}\!
  &=\,
    u_3 v_3
    \biggsem{\,\input{./Figures-Source/ZX-CNOT-gadget.tikz}\,}_{\!\alpha}\!
  =\,
  \frac{\,u_3 v_3\,}{\sqrt 2}\,\mathrm{CNOT}
  \,,
  \end{aligned}
  \end{equation}~\\[-2ex]
  from which it follows that $v_3 = u_3 = 2^{1\!\!\:/\!\!\;4}$; a similar argument from $\sem{\,\cdot\,}_\beta$ yields 
  \vspace*{-1ex}
  \begin{equation}
  \begin{aligned}[b]
    \biggsem{\,\input{./Figures-Source/ZH-CNOT-gadget.tikz}\,}_{\!\nu}\!
  &=\,
    u_3 g_3
    \biggsem{\,\input{./Figures-Source/ZH-CNOT-gadget.tikz}\,}_{\!\alpha}\!
  =\,
    \frac{\,u_3 g_3\,}{\sqrt 2}\,\mathrm{CNOT}
    \,,
  \end{aligned}
  \end{equation}~\\[-2ex]
  so that $g_3 = 2^{-1\!\!\:/\!\!\;4}$.
  Finally, returning to Eqn.~\eqref{eqn:unitaryCkZ} for arbitrary $k \ge 1$, we have
  \vspace*{-1ex}
  \begin{equation}
  \begin{aligned}[b]
  \Sem{10ex}{\!\input{./Figures-Source/ZH-CkZ-gadget.tikz}\;}_{\!\nu}\!\!
  &=\;
  h_k u_3^k \;
    \sum_{\mathclap{x \in \{\zsymb,\osymb\}^k}} \; (-1)^{x_1 x_2 \cdots x_k} \ket{x}\!\!\bra{x};
\end{aligned}  
\end{equation}~\\[-2ex]
here, the right-hand side is a diagonal map with eigenvalues whose norms are all equal to $h_k 2^{k\!\!\;/\!\!\;4}$.
This map is unitary iff $h_k = 2^{-k\!\!\;/\!\!\;4}$, in which case it equals $\mathrm{diag}(+1,\cdots,+1,-1)$.
The converse follows easily from the equations above.
\end{proof}


\vspace*{-1ex}
\section{Compatibility of idealised rewrites in ZX calculi}
\label{apx:compatibilityRewritesZX}
\vspace*{-.5ex}

In this Section, we prove the relationships between various rewrites and normalisation factors for generators, in versions of the ZX calculus with ``reasonable'' normalisations of its generators.
In particular, we are interested in those rewrites which, if sound, would cause the green (and the red) nodes to form a special commutative dagger-Frobenius algebra, and for the green nodes and red nodes to form a bialgebra (specifically, for the phase-free green nodes to form a coalgebra which is compatible with an algebra formed by the phase-free red nodes).
Figure~\ref{fig:candidateZXrewrites} presents candidate ZX rewrites to this effect.
Note that the properties \Rule(Unit R), \Rule(Counit Z), \Rule(Bialg ZR), and \Rule(Copy ZR) are necessary and sufficient for the green and red nodes to form a bialgebra of the sort described above; the rules \Rule(Id Z), \Rule(Fuse Z), and \Rule(Special Z) are necessary and sufficient for the green nodes to form a special commutative dagger-Frobenius algebra.

The main results of this Appendix are summarised in Figure~\ref{fig:compatibilityZXsummary}.
These describe constraints on a model $\sem{\,\cdot\,}$, defined below in terms of a family of parameters $(u_k)_{k \ge 0}$, which are imposed if one requires any given idealised rewrite to be sound.

\begin{figure}[!t]
\allowdisplaybreaks
\input{ZX-idealised-apx-figs}
\vspace*{-4ex}
\caption{%
  \label{fig:candidateZXrewrites}
  Idealised rewrite rules for ZX calculi, which hold variously if the green nodes form a special commutative dagger-Frobenius algebra or a bialgebra over the red nodes.
}

\bigskip

\centering
\begin{tabular}{|r@{$\iff$}l|r@{$\iff$}l|}
\hline$\Bigg.$
  \Rule(Id Z) & $u_2 = 1$
&
  $\left.
  \begin{aligned}
  \Rule(Unit R)
\\[-.25ex]
  \Rule(Counit R)  
  \end{aligned}\right\}$
   &
  $u_3 = u_1^{-1}$
\\
\hline$\Bigg.$
  \Rule(Copy ZR)  & $u_3 = \sqrt 2\,u_1$
&
  \Rule(Bialg ZR) & $u_3 = 2^{1\!\!\:/\!\!\;4}$
\\
\hline$\Bigg.$
  \Rule(Fuse Z) & $u_k = u_1^{2-k} \;(\forall k \ge 0)$
&
  \Rule(Special Z) &  $u_3 = 1$
\\
\hline
\end{tabular}
\vspace*{-1ex}
\caption{%
  Summary of the results of Section~\ref{sec:soundnessConditionsZX}, associating conditions for the soundness of rewrites for an Ockhamic model of a ZX calculus. 
}
\label{fig:compatibilityZXsummary}
\end{figure}

\vspace*{-1ex}
\subsection{Ockhamic models of ZX calculi}
\label{sec:ockhamicZX}
\vspace*{-.5ex}

For the purposes of this Appendix, we are only interested in a particular kind of model $\sem{\,\cdot\,}$, which we call ``Ockhamic''.
We present this definition to describe normalisations of the ZX calculus which avoid scalar factors arising from topology or from non-unitarity representations of Hadamard gates, in which the bases $\{\ket{\zsymb},\ket{\osymb}\}$ and $\{\ket{\psymb},\ket{\msymb}\}$ are on essentially an equal footing.
(These models are slightly more restrictive than the description of the model $\sem{\,\cdot\,}_\nu$ set out in Eqns.~\eqref{eqn:parametersModelNu}, though such a model which also satisfies Eqn.~\eqref{eqn:unitaryNOTandH} and for which the rewrite \Rule(Switch) on page~\pageref{newZX} is sound will be Ockhamic.)

\begin{definition}
\label{def:ockhamicZX}
\allowdisplaybreaks
A model $\sem{\,\cdot\,}$ for a scalar-exact version 
of the ZX calculus is \emph{Ockhamic} if it maps nodes to complex matrices, satisfies Eqns.~\eqref{eqn:stringGenerators}, and if there exist a sequence $(u_k)_{k \in \N}$ of positive scalars such that:
\vspace*{-.5ex}
  \begin{small}
  \begin{align}{}
  \label{eqn:ockhamicZX}
  \begin{split}
  \mspace{-18mu}
  \Biggsem{\!\!\!\input{./Figures-Source/ZX-green-phase-dot-arity.tikz}\!\!\!}
  &=\,
    u_{m{+}n}\Bigl(\ket{\zsymb}^{\!\otimes n}\!\bra{\zsymb}^{\!\!\;\otimes m}
    \,+\;
    \e^{i\theta}
    \ket{\osymb}^{\!\otimes n}\!\bra{\osymb}^{\!\!\;\otimes m} \Bigr),
  \mspace{-18mu}
%
%
\\[.5ex]
%
%
  \mspace{-18mu}
  \Biggsem{\!\!\!\input{./Figures-Source/ZX-red-phase-dot-arity.tikz}\!\!\!}
  &=\,
    u_{m{+}n} \Bigl(\ket{\psymb}^{\!\otimes n}\!\bra{\psymb}^{\!\!\;\otimes m}
    \,+\;
    \e^{i\theta}
    \ket{\msymb}^{\!\otimes n}\!\bra{\msymb}^{\!\!\;\otimes m}\Bigr),
  \mspace{-18mu}
\end{split}
%
%
&
%
%
  \Bigsem{\,\input{./Figures-Source/ZX-H-box.tikz}\,}_\alpha
  \,&=\,
  \tfrac{1}{\sqrt 2} \text{\small$\begin{bmatrix}
    1 \!&\! \phantom-1 \\ 1 \!&\! -1
  \end{bmatrix}$}
  .
  \mspace{-18mu}
  \end{align}
  \end{small}~\\[-1ex]%
  We say that a (version of the) ZX calculus is itself \emph{Ockhamic} if it admits an Ockhamic model.
\end{definition}

\medskip
\noindent
We admit that there is at least one plausible reason to consider a ``non-Ockhamic'' model $\sem{\,\cdot\,}$ of the ZX generators.
If one defines $\sem{\,\cdot\,}$ in such a way that \emph{both} components of green or red nodes may depend on the phase parameter $\theta$, one may define a calculus in which degree-2 green and red nodes are mapped uniformly to elements of $\mathrm{SU}(2)$.
This would allow for calculi which have simpler versions of the rule \Rule(Euler) --- at the cost, however, of introducing complex scalar gadgets for the bialgebra rule (\emph{c.f.}~Ref.~\cite{Bakewell-2020}).
While the theoretically-motivated may find it worth-while to consider the possible benefits of such added flexibility, we consider those models which do not introduce complex phases in the bialgebra rule to be of central interest.

\vspace*{-1ex}
\subsection{Compatibility of rules in Ockahmic ZX calculi}
\vspace*{-.5ex}

We now consider conditions under which the rewrite rules of Figure~\ref{fig:candidateZXrewrites} are sound for an Ockhamic (model of some version of the) ZX calculus.
We proceed mainly by reduction to the rewrites and model $\sem{\,\cdot\,}_\alpha$ of the existing versions of the ZX calculus, as exemplified by that of Section~\ref{sec:traditionalZXandZH}.

\vspace*{-1ex}
\subsubsection{Conditions for soundness of individual rules}
\label{sec:soundnessConditionsZX}%
\vspace*{-.5ex}


\begin{lemma}
\label{lemma:unitRcounitZ-ZX}
In an Ockhamic ZX calculus, \Rule(Unit R) and \Rule(Counit Z) are each sound iff  $u_3 = u_1^{-1}$.
\end{lemma}
\begin{proof}
  This follows from the fact that in an Ockhamic ZX calculus,
  \begin{equation}
  \begin{aligned}[b]
    \\[-4ex]
    \biggsem{\input{./Figures-Source/ZX-green-counit.tikz}}
  \;=\;
    u_1 u_3
    \biggsem{\input{./Figures-Source/ZX-green-counit.tikz}}_\alpha
  \;=\;
    u_1 u_3
    \bigsem{\input{./Figures-Source/id-wire.tikz}}_\alpha
  \;=\;
    u_1 u_3
    \bigsem{\input{./Figures-Source/id-wire.tikz}};
  \end{aligned}
  \end{equation}
  a similar observation holds regarding the red nodes and the rule \Rule(Unit R).
\end{proof}


\begin{lemma}
  \label{lemma:copyZR-ZX}
    In an Ockhamic ZX calculus, the \textup{(Copy ZR)} is sound iff $u_3 = \sqrt 2\,u_1$.
\end{lemma}
\begin{proof}
  This follows from the fact that in an Ockhamic ZX calculus,
  \begin{equation}
    \biggsem{\input{./Figures-Source/ZX-g-copy-r.tikz}}
  \,=\,
    u_1 u_3
    \biggsem{\input{./Figures-Source/ZX-g-copy-r.tikz}}_{\!\alpha} \!\!
  =\,
    \frac{u_1 u_3}{\sqrt 2}
    \biggsem{\input{./Figures-Source/ZX-r-and-r.tikz}}_{\!\alpha} \!\!
  =\,
    \frac{u_3}{\sqrt 2\, u_1}
    \biggsem{\input{./Figures-Source/ZX-r-and-r.tikz}} .
  \tag*{\qedhere}
  \end{equation}
\end{proof}


\begin{lemma}
  \label{lemma:bialgZR-ZX}
    In an Ockhamic ZX calculus, \Rule(Bialg ZR) is sound only iff $u_3 = 2^{1\!\!\;/\!\!\;4}$.
\end{lemma}
\begin{proof}
  This follows from the fact that $u_3 > 0$, and that in an Ockhamic ZX calculus,
  \begin{equation}{}
    \biggsem{\input{./Figures-Source/ZX-bialg.tikz}}
  =\,
    u_3^4
    \biggsem{\input{./Figures-Source/ZX-bialg.tikz}}_{\!\alpha}\!\!
=\,
    \frac{u_3^4}{\sqrt 2}
    \biggsem{\input{./Figures-Source/ZX-bott.tikz}}_{\!\alpha}\!\! 
=\,
    \frac{u_3^2}{\sqrt 2}
    \biggsem{\input{./Figures-Source/ZX-bott.tikz}} .
  \mspace{-36mu}
  \notag\qedhere
  \end{equation}
\end{proof}


\begin{lemma}[trivial]
\label{lemma:IdZ-ZX}
In an Ockhamic ZX calculus, \Rule(Id Z) is sound iff $u_2 = 1$.
\end{lemma}


\begin{lemma}
\label{lemma:fuseZ-ZX}
In an Ockhamic ZX calculus, \Rule(Fuse Z) is sound iff $u_k = u_1^{2-k}$ for all $k \ge 0$.
\end{lemma}

\begin{proof}
  This follows from the fact that in an Ockhamic ZX calculus,
  \begin{small}%
  \begin{equation}{}
  \mspace{-24mu}
  \begin{aligned}[b]
    \Sem{8ex}{\!\!\input{./Figures-Source/ZX-green-phase-fuse-arity.tikz}\!\!}
  &=\,
    u_{k{+}m{+}1} u_{\ell{+}n{+}1}
    \Sem{8ex}{\!\!\input{./Figures-Source/ZX-green-phase-fuse-arity.tikz}\!\!}_{\!\alpha}
    \!\!
  \\[1ex]&=\,
    u_{k{+}m{+}1} u_{\ell{+}n{+}1}
    \Sem{4ex}{\!\!\input{./Figures-Source/ZX-green-phase-sum-arity-sum.tikz}\!\!}_{\!\alpha}
  \\[2ex]&=\,
    \frac{u_{k{+}m{+}1} u_{\ell{+}n{+}1}}{u_{k{+}\ell{+}m{+}n}}
    \Sem{4ex}{\!\!\input{./Figures-Source/ZX-green-phase-sum-arity-sum.tikz}\!\!}
  .
  \end{aligned}
  \end{equation}%
  \end{small}%
  That is, \Rule(Fuse Z) is sound iff $u_{M{+}N} = {u_{M{+}1} \!\: u_{N{+}1}}$ for all $M,N \ge 0$.
  Suppose that this relation holds among the coefficients $u_k$: then in particular,
  \begin{itemize}
  \item
    setting $M \!\!\!\;=\!\!\; 0$ and $N \!\!\!\;=\!\!\; 0$, we obtain $u_0 \!=\! u_1^2$;
  \item
    setting $M \!\!\!\;=\!\!\; n$ and $N \!\!\!\;=\!\!\; 0$, we obtain $u_{n} \!=\! {u_{n{+}1} u_1}$, so that $u_{n{+}1} \!\!\;/\!\!\; u_n \!=\! u_1^{-1}$.
  \end{itemize}
  Thus $u_n = u_0 \bigl( u_n / u_0 \bigr) = \bigl( u_1^2 \bigr) \bigl( u_1^{-1} \bigr)^n = u_1^{2-n}$. 
  Conversely, supposing that we have $u_n = u_1^{2-n}$, then
  \begin{equation}
  \begin{aligned}[b]
      \frac{u_{M{+}N}}{u_{M{+}1} u_{N{+}1}}
    \;&=\;
      \frac{%
        u_1^{2-M-N}
      }{%
        u_1^{2-(M+1)} u_1^{2-(N+1)}
      }
  \;=\;
    1,
  \end{aligned}
  \end{equation}
  so that \Rule(Fuse Z) is sound.
\end{proof}


\begin{lemma}
  \label{lemma:specialZ-ZX}
  In an Ockhamic ZX calculus, \textup{(Special Z)} is sound iff $u_3 = 1$.
\end{lemma}
\begin{proof}
  This follows from $u_3 > 0$ and that, in an Ockhamic ZX calculus, we have
  \begin{equation}
    \begin{aligned}[b]
    \\[-4ex]
      \Bigsem{\input{./Figures-Source/ZX-green-special.tikz}}
    \;=\;
      u_3^2 \Bigsem{\input{./Figures-Source/ZX-green-special.tikz}}_\alpha
    \;=\;
      u_3^2 \bigsem{\input{./Figures-Source/id-wire.tikz}}_\alpha
    \;=\;
      u_3^2 \bigsem{\input{./Figures-Source/id-wire.tikz}}
  .
  \end{aligned}
  \tag*{\qedhere}
  \end{equation}
\end{proof}

\vspace*{-1ex}
\subsubsection{Incompatibility results}

Here, we consider characterisations of Ockhamic models in which \textbf{(a)}~the green and red nodes form a bialgebra, \textbf{(b)}~the green nodes form a special commutative dagger-Frobenius algebra, and \textbf{(c)}~the arity-1 red node $\,\input{./Figures-Source/ZX-red-prep.tikz}\,$ precisely represents the vector $\ket{0}$. 

\begin{lemma}
  \label{lemma:coeffConstraintsBialgebra}
  In an Ockhamic ZX calculus, the green and red nodes form a bialgebra --- that is, \Rule(Unit~R), \Rule(Counit~Z), \Rule(Copy~ZR), and \Rule(Bialg~ZR) are all sound --- iff $u_1 = 2^{-1\!\!\;/\!\!\;4}$ and $u_3 = 2^{1\!\!\;/\!\!\;4}$.
  Furthermore, any two of \Rule(Counit Z), \Rule(Copy ZR), and \Rule(Bialg ZR) imply the other.
\end{lemma}
\begin{proof}
  First note that \Rule(Unit R) holds if and only if \Rule(Counit Z) holds.
  Then:
  \begin{itemize}
  \item  
    The soundness of \Rule(Bialg ZR) is equivalent to $u_3 = 2^{1\!\!\:/\!\!\;4}$ by Lemma~\ref{lemma:copyZR-ZX}.
    Given $u_3 = 2^{1\!\!\:/\!\!\;4}$, either of the other two rules \Rule(Counit Z) or \Rule(Copy ZR) are equivalent to $u_1 = 2^{-1\!\!\:/\!\!\;4}$.
  \item
    The soundness of \Rule(Counit Z) and \Rule(Copy ZR) together is equivalent to $u_3 = \sqrt 2 u_1 = \sqrt 2 u_3^{-1}$, which for $u_3 > 0$ is equivalent to $u_3 = 2^{1\!\!\:/\!\!\;4}$, which in turn is equivalent to \Rule(Bialg ZR).
    \qedhere 
  \end{itemize}
\end{proof}

\begin{lemma}
  \label{lemma:coeffConstraintsSpecial}
  In an Ockhamic ZX calculus, the green nodes form a special commutative dagger-Frobenius algebra --- that is, \Rule(Id~Z), \Rule(Fuse~Z), and \Rule(Special~Z) are all sound --- iff $u_k = 1$ for all $k \ge 0$.
\end{lemma}
\begin{proof}
  If \Rule(Id Z) and \Rule(Special Z) are both sound, we have $u_2 = u_3 = 1$; and conversely.
  Given that $u_2 = u_3 = 1$, if \Rule(Fuse Z) is sound, then $u_1 = u_1^{3-2} = u_3^{-1} = 1$, and more generally $u_k = u_1^{2-k} = 1$; the converse here is simple as well.
\end{proof}

\begin{corollary}
  \label{cor:eitherBialgebraOrSpecial}
  In an Ockhamic ZX calculus, either the green and red nodes do not form a bialgebra, or the green nodes do not form a special commutative dagger-Frobenius algebra (nor do the red nodes).
\end{corollary}

\begin{lemma}[trivial]
  In an Ockhamic ZX calculus, $\bigsem{\,\input{./Figures-Source/ZX-red-prep.tikz}\,} = \ket{\zsymb}$ iff $u_1 = 2^{-1\!\!\;/\!\!\;2}$.
\end{lemma}

\begin{corollary}
  \label{cor:dontNormaliseBasisStates}
  In an Ockhamic ZX calculus in which $\bigsem{\input{./Figures-Source/ZX-red-prep.tikz}} = \ket{\zsymb}$, at most one of \Rule(Counit Z), \Rule(Bialg ZR), and \Rule(Copy ZR), are sound, and at most one of \Rule(Fuse Z) and \Rule(Special Z) are sound.
  In particular, the red nodes and green nodes form neither a bialgebra nor special commutative dagger-Frobenius algebras.
\end{corollary}
\begin{proof}
  In such a ZX calculus, we have $u_1 = 2^{-1\!\!\;/\!\!\;2}$.
  Given this, we have the following equivalences:
  \begin{itemize}
  \item
    \Rule(Counit Z) is sound if and only if $u_3 = 2^{1\!\!\:/\!\!\;2}$ by Lemma~\ref{lemma:unitRcounitZ-ZX};
  \item
    \Rule(Copy ZR) is sound if and only if $u_3 = 1$ by Lemma~\ref{lemma:copyZR-ZX};
  \item
    \Rule(Bialg ZR) is sound if and only if $u_3 = 2^{1\!\!\:/\!\!\;4}$ by Lemma~\ref{lemma:bialgZR-ZX};
  \item
    \Rule(Fuse Z) is sound if and only if $u_3 = 2^{1\!\!\:/\!\!\;2}$ by Lemma~\ref{lemma:fuseZ-ZX};
  \item
    \Rule(Special Z) is sound if and only if $u_3 = 1$ by Lemma~\ref{lemma:specialZ-ZX}.
  \end{itemize}
  The Corollary then follows.
\end{proof}


\vspace*{-4ex}
\section{Compatibility of rewrite rules in ZH calculi}
\label{apx:compatibilityRewritesZH}
\vspace*{-.5ex}

In this Section, we prove the relationships between various rewrites and normalisation factors for generators, in versions of the ZH calculus with ``reasonable'' normalisations of its generators.
Figure~\ref{fig:candidateZHrewrites} presents the idealised ZH rewrites: each is in effect a simplified version of some rule from Figure~\ref{fig:variantZHrules}.

The main results of this Appendix are summarised in Figure~\ref{fig:compatibilityZHsummary}.
These describe constraints on a model $\sem{\,\cdot\,}$, defined below in terms of a family of parameters $(u_k)_{k \ge 0}$, which are imposed if one requires any given idealised rewrite to be sound.

\begin{figure}[tp]
\allowdisplaybreaks
\input{ZH-idealised-apx-figs}

\vspace*{-4ex}
\caption{%
  \label{fig:candidateZHrewrites}
  Idealised rewrite rules for ZH calculi, consisting of the simplest diagram transformations which are equivalent (up to a scalar) to the rules of Figure~\ref{fig:usualZXrules} and Figure~\ref{fig:variantZHrules}.
  The phases $a,b \in \C$ may be arbitrary.
}

\bigskip

\centering
\begin{tabular}{|r@{$\iff$}l|r@{$\iff$}l|}
\hline$\Bigg.$
  \Rule(Id Z) & $u_2 = 1$
&
  \Rule(Id H) & $h_2 = 2^{-1\!\!\:/\!\!\;2}$
\\
\hline$\Bigg.$
  \Rule(Not)  & $u_3 = (2h_1 h_2^2)^{-1}$
&
  \Rule(Switch ZG)  & $g_k = 2h_2^k u_k \; (\forall k \ge 0)$
\\
\hline$\Bigg.$
  \Rule(Mult ZH)  & $h_1 = u_3^{-1}$
&
  \Rule(Unit ZH) & $h_1 = u_1$
\\
\hline$\Bigg.$
  \Rule(Bialg ZG) &
    $g_k = u_k^{-1} = u_1^{k{-}2} \;(\forall k \ge 1)$
&
  \Rule(Bialg ZH) &
    $u_k = \bigl(\!\!\:\sqrt 2 \,h_k\bigr)^{-1} = u_1^{-(k{-}2)} \;(\forall k \ge 1)$
\\
\hline$\Bigg.$
  \Rule(Fuse Z)  & 
    $u_k = u_1^{-(k{-}2)} \;(\forall k \ge 0)$
&
  \Rule(Special Z)  & $u_3 = 1$
\\
\hline$\Bigg.$
  \Rule(Fuse H) & $h_k = 2^{-k/4} \;(\forall k \ge 0)$
&
  \Rule(Orth ZH) & $u_1 = 2^{-1\!\!\:/\!\!\;2}$
\\
\hline$\Bigg.$
  \Rule(Dilem ZH) & $u_1 = h_1 (2h_2^2)^{-1}$
&
  \Rule(Avg ZH) & $u_3 = h_1 (2h_2^2)^{-1}$
\\
\hline
\end{tabular}
\vspace*{-1ex}
\caption{%
  Summary of the results of Section~\ref{sec:soundnessConditionsZX}, associating conditions for the soundness of rewrites for an Ockhamic model of a ZX calculus. 
}
\label{fig:compatibilityZHsummary}
\end{figure}

\vspace*{-1ex}
\subsection{Ockhamic models of ZH calculi}

Following the approach of Section~\ref{sec:ockhamicZX}, we consider what constraints one might argue that ``any reasonable'' normalisation of the ZH generators must satisfy.
This may differ significantly from what a ``reasonable'' normalisation is for a ZX~calculus, as the ZH calculus is less closely tied to the priorities of foundations of physics.
In particular, the model described by Backens and Kissinger~\cite{BK-2019} for the calculus effectively fixes a special role for the standard basis, so that there is no reason \emph{a priori} that the white nodes, gray nodes, and H boxes should have related normalisation factors. (This is particularly true of gray nodes, whose definition as a shorthand in Ref.~\cite{BK-2019} itself involves an additional scalar factor).

The only apparent candidates for constraints on a ``reasonable'' normalisation is then that, as the $H$ box is unitary in the ZX calculus, so too should the not-dot be in the ZH calculus; and that furthermore a degree-$0$ H-box with phase $a$ should represent the scalar $a$.
This motivates us to define:
\begin{definition}
\label{def:ockhamicZH}
\allowdisplaybreaks
A model for a scalar-exact version 
of the ZH calculus is \emph{Ockhamic} if it maps nodes to complex matrices, satisfies Eqns.~\eqref{eqn:stringGenerators}, and if there exist sequences $(u_k)_{k \in \N}$, $(h_k)_{k \in \N}$, and $(g_k)_{k \in \N}$ of positive scalars (where in particular $h_0 = 1$) such that:
\vspace*{-0.5ex}
\begin{small}%
\begin{equation}{}
\label{eqn:ockhamicZH}
\mspace{-72mu}
\begin{aligned}
  \Biggsem{\!\!\!\!\input{./Figures-Source/ZH-white-dot-arity.tikz}\!\!\!}_{\!\beta}
  \!&=\,
    u_{m{+}n} \Bigl(\ket{\zsymb}^{\!\otimes n}\!\bra{\zsymb}^{\!\!\;\otimes m}
    +\!\;
    \ket{\osymb}^{\!\otimes n}\!\bra{\osymb}^{\!\!\;\otimes m}\Bigr),
%
%
&\qquad
%
%
  \mspace{-24mu}
  \Biggsem{\!\!\!\!\input{./Figures-Source/ZH-gray-dot-arity.tikz}\!\!\!}_{\!\beta}
  \!&=\,
    g_{m{+}n} \; \mathop{\;\sum \;\sum\;}_{%
      \mathclap{\substack{
        {x \in \{\zsymb,\osymb\}^m \!,}
        \,\;
        {y \in \{\zsymb,\osymb\}^n} \\
      w(x) \!\!\;+\!\!\; w(y) \,\in\, 2\Z
      }}}
      \;\;
      \ket{y}\!\!\bra{x} \!\!\:,
      \mspace{-48mu}
%
%
\\[1ex]
\\[-3ex]
%
%
  \Biggsem{\!\!\!\!\input{./Figures-Source/ZH-H-phase-box-arity.tikz}\!\!\!}_{\!\beta}
  \!&=\;
    h_{m{+}n} \;\mathop{\;\sum \;\sum\;}_{%
      \mathclap{
        \substack{
        {x \in \{\zsymb,\osymb\}^m}
        \\\
        {y \in \{\zsymb,\osymb\}^n}
      }}}
      \,
      a^{x_1 \cdots x_m y_1 \cdots y_n} \!
      \ket{y}\!\!\bra{x}\,, 
%
%
&
%
%
  \Bigsem{\input{./Figures-Source/ZH-not-dot.tikz}}_\beta
  \;&=\;
  \text{\footnotesize$\begin{bmatrix}
    0 \!\!&\!\! 1 \\[-0.25ex] 1 \!\!&\!\! 0
  \end{bmatrix}$}.
  \mspace{-18mu}
\end{aligned}%
\mspace{-72mu}
\end{equation}%
\end{small}~\\[-2ex]
  We say that a (version of the) ZH calculus is itself \emph{Ockhamic} if it admits an Ockhamic model.
\end{definition}

\vspace*{-3ex}
\paragraph{Remark.}
As with the definition of Ockhamic models of ZX calculi (Definition~\ref{def:ockhamicZX} in Appendix~\ref{apx:compatibilityRewritesZX} on page~\pageref{def:ockhamicZX}), the models described above are slightly more restrictive than those described in Eqns.~\eqref{eqn:parametersModelNu}.
However, any model which satisfies Eqns.~\eqref{eqn:parametersModelNu} and~\eqref{eqn:nu-scalar-box}, and also specifically the constraint on the not dot in Eqn.~\eqref{eqn:unitaryNOTandH}, will be Ockhamic.

\vspace*{-1ex}
\subsection{Compatibility of rules in Ockahmic ZH calculi}
\vspace*{-.5ex}

We now consider conditions under which the rewrite rules of Figure~\ref{fig:candidateZHrewrites} are sound for an Ockhamic (model of some version of the) ZH calculus.
We proceed mainly by reduction to the rewrites and model $\sem{\,\cdot\,}_\beta$ of the existing version of the ZH calculus, as exemplified by that of Section~\ref{sec:traditionalZXandZH}.

\vspace*{-1ex}
\subsubsection{Conditions for soundness of individual rules}
\label{sec:soundnessConditionsZH}
\vspace*{-.5ex}


\begin{lemma}[trivial]
\label{lemma:idZ-ZH}
In an Ockhamic ZH calculus, \Rule(Id Z) is sound ifff $u_2 = 1$.
\end{lemma}


\begin{lemma}[trivial]
\label{lemma:idH-ZH}
In an Ockhamic ZH calculus, \Rule(Id H) is sound iff $h_2 = 2^{-1\!\!\:/\!\!\;2}$.
\end{lemma}


\begin{lemma}
\label{lemma:not-ZH}
In an Ockhamic ZH calculus, \Rule(Not) is sound iff $u_3 =  (2 h_1 h_2^2)^{-1}$.
\end{lemma}
\begin{proof}
  This follows from the fact that in an Ockhamic ZH calculus,
    \begin{equation}
    \begin{aligned}[b]
    \Bigsem{\input{./Figures-Source/ZH-not-gadget.tikz}}
  =\,
    h_1 h_2^2 u_3 
    \Bigsem{\input{./Figures-Source/ZH-not-gadget.tikz}}_{\!\beta}
  =\,
    2 h_1 h_2^2 u_3 
    \Bigsem{\input{./Figures-Source/ZH-not-dot.tikz}}_{\!\beta}
  =\,
    2 h_1 h_2^2 u_3 
    \Bigsem{\input{./Figures-Source/ZH-not-dot.tikz}}
  .
  \end{aligned}
  \tag*{\qedhere}
  \end{equation}  
\end{proof}


\begin{lemma}
\label{lemma:switchZG-ZH}
In an Ockhamic ZH calculus, \Rule(Switch ZG) is sound iff $g_k  = 2 h_2^{k} u_k$ for all $k \ge 0$.
\end{lemma}
\begin{proof}
  This follows from the fact that in an Ockhamic ZH calculus,
    \begin{equation}
    \begin{aligned}[b]
    \Biggsem{\!\!\input{./Figures-Source/ZH-white-w-H-arity.tikz}\!\!}
  =\;
    h_2^{m{+}n} u_{m{+}n} 
    \Biggsem{\!\!\input{./Figures-Source/ZH-white-w-H-arity.tikz}\!\!}_{\!\beta}
  &=\;
    2 h_2^{m{+}n} u_{m{+}n} 
    \Biggsem{\!\!\input{./Figures-Source/ZH-gray-dot-arity.tikz}\!\!}_{\!\beta}
  \\&=\;
    \frac{2 h_2^{m{+}n} u_{m{+}n}}{g_{m{+}n}}
    \Biggsem{\!\!\input{./Figures-Source/ZH-gray-dot-arity.tikz}\!\!}
  .
  \end{aligned}
  \tag*{\qedhere}
  \end{equation}  
\end{proof}


\begin{lemma}
\label{lemma:multZH-ZH}
In an Ockhamic ZH calculus, \Rule(Mult ZH) is sound iff $u_3 = h_1^{-1}$.
\end{lemma}
\begin{proof}
  This follows from the fact that in an Ockhamic ZH calculus,
    \begin{equation}
    \begin{aligned}[b]
    \biggsem{\!\!\input{./Figures-Source/ZH-mult.tikz}\,}
  =\;
    h_1^2 u_3
    \biggsem{\!\!\input{./Figures-Source/ZH-mult.tikz}\,}_{\!\beta}
  =\;
    h_1^2 u_3
    \biggsem{\!\!\input{./Figures-Source/ZH-H-prod-prep.tikz}\,}_{\!\beta}
  =\;
    h_1 u_3
    \biggsem{\!\!\input{./Figures-Source/ZH-H-prod-prep.tikz}\,}
  .
  \end{aligned}
  \tag*{\qedhere}
  \end{equation}  
\end{proof}


\begin{lemma}
\label{lemma:unitZH-ZH}
In an Ockhamic ZH calculus, \Rule(Unit ZH) is sound iff $h_1 = u_1$.
\end{lemma}
\begin{proof}
  This follows from the fact that in an Ockhamic ZH calculus,
    \begin{equation}
    \begin{aligned}[b]
    \bigsem{\;\input{./Figures-Source/ZH-white-prep.tikz}\,}
  =\;
    u_1
    \bigsem{\;\input{./Figures-Source/ZH-white-prep.tikz}\,}_{\!\beta}
  =\;
    u_1
    \bigsem{\!\!\input{./Figures-Source/ZH-H-plus1-prep.tikz}\,}_{\!\beta}
  =\;
    \frac{u_1}{h_1}
    \bigsem{\!\!\input{./Figures-Source/ZH-H-plus1-prep.tikz}\,}
  .
  \end{aligned}
  \tag*{\qedhere}
  \end{equation}  
\end{proof}


\begin{lemma}
\label{lemma:bialgZG-ZH}
In an Ockhamic ZH calculus, \Rule(Bialg ZG) is sound iff $g_k = u_k^{-1} = u_1^{k{-}2}$ for all $k \ge 1$.
\end{lemma}
\begin{proof}
  This follows from the fact that in an Ockhamic ZH calculus,
    \begin{equation}{}
  \label{eqn:contemplatingBialgZG}
    \mspace{-24mu}
    \begin{aligned}[b]
    \Biggsem{\!\!\!\!\input{./Figures-Source/ZH-bialg-white-gray-arity.tikz}\!\!\!\!}
  =\,
    u_{n{+}1}^m g_{m{+}1}^n
    \Biggsem{\!\!\!\input{./Figures-Source/ZH-bialg-white-gray-arity.tikz}\!\!\!}_{\!\beta}
  \!&=\,
    u_{n{+}1}^m g_{m{+}1}^n
    \Biggsem{\!\!\!\!\input{./Figures-Source/ZH-bott-white-gray-arity.tikz}\!\!\!\!}_{\!\beta}
  \\&=\,
    \frac{u_{n{+}1}^m g_{m{+}1}^n}{u_{n{+}1} g_{m{+}1}} 
    \Biggsem{\!\!\!\!\input{./Figures-Source/ZH-bott-white-gray-arity.tikz}\!\!\!\!}
  .
  \end{aligned}
  \mspace{-18mu}
  \end{equation}  
  That is, \Rule(Bialg ZG) is sound if and only if $u_{n{+}1}^{m{-}1} = g_{m{+}1}^{-(n{-}1)}$ for all $m,n \ge 0$.
  Suppose that this relation holds among the coefficients $u_k$ and $g_k$%
  : then, in particular,
  \begin{itemize}
  \item
    Setting $m = 0$, we obtain $u_{n{+}1}^{-1} = g_1^{-(n{-}1)}$, so that $u_{n{+}1} = g_1^{n{-}1}$. 
  \item
    Setting $m = 2$, we obtain $u_{n{+}1}^1 = g_3^{-(n{-}1)}$.
    Together with the preceding case, we then have $g_3 = g_1^{-1}$.
  \item
    Setting $n = 0$, we obtain $u_1^{m{-}1} = g_{m{+}1}$; in particular, we then have $g_1 = u_1^{-1}$.
\end{itemize}
From these last equalities, we have $u_k = g_1^{k{-}2} = u_1^{-(k-2)}$ and $g_k = u_1^{k{-}2}$.
Conversely, if $u_k = g_k = u_1^{-(k{-}2)}$ for all $k \ge 1$, we have
\begin{equation}
\begin{aligned}[b]
    \frac{
      u_{n{+}1}^m g_{m{+}1}^n
    }{
      u_{n{+}1} g_{m{+}1}
    }
  \;=\;
    \frac{
      u_1^{-m(n{-}1)} u_1^{n(m{-}1)}
    }{
      u_1^{-(n{-}1)} u_1^{(m{-}1)}
    }
  \;=\;
    \frac{
      u_1^{(mn{-}n) -(mn{-}m)}
    }{
      u_1^{-(n{-}1)} u_1^{(m{-}1)}
    }
    \;=\;
    1,
\end{aligned}
\end{equation}
so that \Rule(Bialg ZG) is sound, by Eqn.~\eqref{eqn:contemplatingBialgZG}.
\end{proof}


\begin{lemma}
\label{lemma:bialgZH-ZH}%
In an Ockhamic ZH calculus, \Rule(Bialg ZH) is sound iff $u_k = (\sqrt 2 h_k)^{-1} = u_1^{-(k{-}2)}$ for all $k \ge 1$.
In particular, if \Rule(Bialg ZH) is sound, then $u_3 = u_1^{-1}$, $u_2 = 1$, $h_3 = (2 h_1)^{-1}$, and $h_2 = 2^{-1\!\!\:/\!\!\;2}$.
\end{lemma}
\begin{proof}
  This follows from the fact that in an Ockhamic ZH calculus,
    \begin{equation}
    \label{eqn:contemplatingBialgZH}
    \begin{aligned}[b]
    \Biggsem{\!\!\input{./Figures-Source/ZH-bialg-white-H-arity.tikz}\!\!}
  &=\;
    h_2^{n} \!\; h_{m{+}1}^{n} u_{n{+}1}^m 
    \Biggsem{\!\!\input{./Figures-Source/ZH-bialg-white-H-arity.tikz}\!\!}_{\!\beta}
  \\&=\;
    2^{n{-}1} h_2^{n} \!\; h_{m{+}1}^{n} u_{n{+}1}^m 
    \Biggsem{\!\!\input{./Figures-Source/ZH-bott-white-H-arity.tikz}\!\!}_{\!\beta}
  \\&=\;
    \frac{2^{n{-}1} h_2^{n} \!\; h_{m{+}1}^{n} u_{n{+}1}^m}{h_{m{+}1} h_2 u_{n{+}1}} 
    \Biggsem{\!\!\input{./Figures-Source/ZH-bott-white-H-arity.tikz}\!\!}
  .
  \end{aligned}
  \end{equation}  
  That is, \Rule(Bialg ZH) is sound if and only if $ u_{n{+}1}^{m{-}1} = \bigl(2 h_2 \!\; h_{m{+}1}\bigr)^{-(n{-}1)}$ for all $m,n \ge 0$.
  Suppose that this relation holds among the coefficients $u_k$ and $h_k$.
  Considering the cases $m = 1$ and $n = 1$ yield the equalities $(2 h_2^2)^{-(n{-}1)} = u_{n{+}1}^0 = 1$ and $u_2^{m{-}1} = (2^{1\!\!\:/\!\!\;2} h_{m{+}1})^0 = 1$, which implies that $h_2 = 2^{-1\!\!\:/\!\!\;2}$ and $u_2 = 1$.
  In particular, it follows that
  \begin{equation}
    u_{n{+}1}^{m{-}1} = \bigl(2^{1\!\!\:/\!\!\;2} h_{m{+}1}\bigr)^{-(n{-}1)}
  \end{equation}
  for all $m, n \ge 0$.
  We next consider the constraints implied by other specific values of $m$ and $n$: 
  \begin{itemize}
  \item
    Setting $m = 0$, we obtain $u_{n{+}1}^{-1} = (2^{1\!\!\:/\!\!\;2} h_1)^{-(n-1)}$, so that
    \begin{equation}
    \label{eqn:bialgZH-u-charn}
      u_{n{+}1} = (2^{1\!\!\:/\!\!\;2} h_1)^{n{-}1} . 
    \end{equation}
  \item
    Setting $m = 2$, we obtain $u_{n{+}1}^1 = (2^{1\!\!\:/\!\!\;2} h_3)^{-(n{-}1)}$.
    Together with Eqn.~\eqref{eqn:bialgZH-u-charn}, we then have
    \begin{equation}
    \begin{alignedat}{2}
        \bigl(2^{1/\!\!\;2} h_1\bigl)^{n{-}1}
      &=\;
        \bigr(2^{1/\!\!\;2} h_3\bigr)^{-(n{-}1)}
      \quad&\implies\quad
        h_3
      &=
        (2 h_1)^{-1}.     
    \end{alignedat}
    \end{equation}
  \item
    Setting $n = 0$, we obtain $u_1^{m{-}1} = (2^{1\!\!\:/\!\!\;2} h_{m{+}1})^1$, so that
    \begin{equation}
        h_{m{+}1}
      \;=\;
        2^{-1\!\!\:/\!\!\;2} u_1^{m{-}1}.
    \end{equation}
\end{itemize}
From this last equality, we have $h_k = 2^{-1\!\!\:/\!\!\;2} u_1^{k{-}2}$.
In particular, $h_1 = 2^{-1\!\!\:/\!\!\;2} u_1^{-1}$ and $h_3 = 2^{-1\!\!\:/\!\!\;2} u_1$, so that $h_3 = 2^{-1} h_1^{-1}$.
We may then use Eqn.~\eqref{eqn:bialgZH-u-charn} to show $u_k = (2^{1\!\!\:/\!\!\;2} h_1)^{k-2} = u_1^{-(k-2)}$; in particular, $u_3 = u_1^{-1}$.

Conversely, suppose that $h_k = 2^{-1\!\!\:/\!\!\;2} u_1^{k{-}2}$ and $u_k = u_1^{-(k-2)}$.
It follows that, in particular, $h_2 = 2^{-1\!\!\:/\!\!\;2}$.
We may then show that
\begin{equation}
\begin{aligned}
      \frac{
        2^{n{-}1} h_2^{n} \!\; h_{m{+}1}^{n} u_{n{+}1}^m
      }{
        h_{m{+}1} h_2 u_{n{+}1}
      }
    \;&=\; 
      \frac{
        (2^{n{-}1})
        (2^{-n\!\!\:/\!\!\;2})
        (2^{-n\!\!\:/\!\!\;2} u_1^{n(m{-}1)})
        (u_1^{-m(n{-}1)})
      }{
        (2^{-1\!\!\:/\!\!\;2} u_1^{m{-}1})
        (2^{-1\!\!\:/\!\!\;2})
        (u_1^{-(n{-}1)})
      }
    \;=\;
      \frac{
        2^{-1}
        u_1^{(nm - n) - (nm - m)}
      }{
        2^{-1} u_1^{(m{-}1) - (n{-}1)}
      }
    \;=\;
      1,
\end{aligned}  
\end{equation}
so that \Rule(Bialg ZH) is sound by Eqn.~\eqref{eqn:contemplatingBialgZH}.
\end{proof}


\begin{lemma}[\emph{c.f.} Lemma~\ref{lemma:fuseZ-ZX}]
\label{lemma:fuseZ-ZH}
In an Ockhamic ZH calculus, \Rule(Fuse Z) is sound iff $u_k \!=\! u_1^{-(k{-}2)}$ for all $k \ge 0$.
\end{lemma}


\begin{lemma}[\emph{c.f.} Lemma~\ref{lemma:specialZ-ZX}]
\label{lemma:specialZ-ZH}
In an Ockhamic ZH calculus, \Rule(Special Z) is sound iff $u_3 = 1$.
\end{lemma}


\begin{lemma}
\label{lemma:fuseH-ZH}
In an Ockhamic ZH calculus, \Rule(Fuse H) is sound iff $h_k =  2^{-k\!\!:/\!\!;4}$ for all $k \ge 0$.
\end{lemma}
\begin{proof}
  This follows from the fact that in an Ockhamic ZH calculus,
    \begin{equation}
    \label{eqn:contemplatingFuseH}
    \begin{aligned}[b]
      \Sem{7ex}{\!\input{./Figures-Source/ZH-H-fuse-arity.tikz}\!}
    &=\;
      h_2 \:\! h_{k{+}m{+}1} \:\! h_{\ell{+}n{+}1}
      \Sem{7ex}{\!\input{./Figures-Source/ZH-H-fuse-arity.tikz}\!}_{\!\beta}
    \\[1ex]&=\;
      2 h_2 \:\! h_{k{+}m{+}1} \:\! h_{\ell{+}n{+}1}
      \Sem{4ex}{\!\!\input{./Figures-Source/ZH-H-phase-box-arity-sum.tikz}\!\!}_{\!\beta}
    \\[1ex]&=\;
      \frac{2 h_2 \:\! h_{k{+}m{+}1} \:\! h_{\ell{+}n{+}1}}{h_{k{+}\ell{+}m{+}n}}
      \Sem{4ex}{\!\!\input{./Figures-Source/ZH-H-phase-box-arity-sum.tikz}\!\!}
  .
  \end{aligned}
  \end{equation}  
  That is, \Rule(Fuse H) is sound iff $h_{M{+}N} = {2 \!\: h_2 \!\: h_{M{+}1} \!\: h_{N{+}1}}$ for all $M,N \ge 0$.
  Suppose that this relation holds among the coefficients $h_k$: then in particular,
  \begin{itemize}
  \item
    setting $M = 0$ and $N = 1$, we obtain $h_1 = {2 \:\! h_2^2 \:\! h_1}$, so that $h_2 = 2^{-1\!\!\:/\!\!\;2}$;
  \item
    setting $M = 0$ and $N = 0$, we obtain $1 = h_0 = {2 \:\! h_2 \:\! h_1 \:\! h_1} = {\sqrt 2 \:\! h_1^2}$, so that $h_1 = 2^{-1\!\!\:/\!\!\;4}$;
  \item
    then, for $M = 0$ in general, we have $h_N = 2 \:\! h_2 \!\: h_1 \!\: h_{N+1} = 2^{1\!\!\:/\!\!\;4} h_{N{+}1}$.
  \end{itemize}
  Thus $h_k = {h_0 \bigl( h_k / h_0 \bigr)} = 1 \cdot \bigl( 2^{1\!\!\:/\!\!\;4} \bigr)^{-1} = 2^{-k\!\!\:/\!\!\;4}$.
  Conversely: if we have $h_k = 2^{-k\!\!\:/\!\!\;4}$, then
  \begin{equation}
      \frac{h_{M{+}N}}{2 h_2 h_{M{+}1} h_{N{+}1}}
    \;=\;
      \frac{%
        2^{-(M{+}N)/\!\!\;4}
      }{%
        (2^{4\!\!\:/\!\!\;4}) (2^{-2\!\!\:/\!\!\;4})(2^{-(M{+}1)/\!\!\;4})(2^{-(N{+}1)/\!\!\;4})
      }
    \;=\;
      1,
  \end{equation}
  so that \Rule(Fuse H) is sound by Eqn.~\eqref{eqn:contemplatingFuseH}.
\end{proof}

\begin{lemma}
  \label{lemma:orthZH-ZH}
  In an Ockhamic ZH calculus, \Rule(Orth ZH) is sound iff $u_1 = 2^{-1\!\!\:/\!\!\;2}$.
\end{lemma}
\begin{proof}
  This follows from the fact that in an Ockhamic ZH calculus,
  \begin{equation}
  \begin{aligned}[b]
    \Biggsem{\,\input{./Figures-Source/ZH-ortho-bridge.tikz}\,}
  \;=\;
    h_3^2 u_3 \Biggsem{\,\input{./Figures-Source/ZH-ortho-bridge.tikz}\,}_{\!\beta}
  \;=\;
    \frac{h_3^2 u_3}{2} \Biggsem{\,\input{./Figures-Source/ZH-ortho-nobridge.tikz}\,}_{\!\beta}
  \;=\;
    \frac{1}{2u_1^2} \Biggsem{\,\input{./Figures-Source/ZH-ortho-nobridge.tikz}\,} .
  \end{aligned}
  \tag*{\qedhere}
  \end{equation}
\end{proof}

\begin{lemma}
  \label{lemma:disjZH-ZH}
  In an Ockhamic ZH calculus, \Rule(Dilem ZH) is sound iff  $u_1 = h_1 (2h_2^2)^{-1}$\,.
\end{lemma}
\begin{proof}
  This follows from the fact that in an Ockhamic ZH calculus,
  \begin{equation}
  \begin{aligned}[b]
    \Bigsem{\,\input{./Figures-Source/ZH-decouple-white-H-a.tikz}\,}
  \;=\;
    h_1 u_1 \Bigsem{\,\input{./Figures-Source/ZH-decouple-white-H-a.tikz}\,}_{\!\beta}
  \;=\;
    \frac{h_1 u_1}{2} \Biggsem{\,\input{./Figures-Source/ZH-disjunct.tikz}\,}_{\!\beta}
  \;=\;
    \frac{h_1}{2 u_1 h_2^2 } \Biggsem{\,\input{./Figures-Source/ZH-disjunct.tikz}\,} .
  \end{aligned}
  \tag*{\qedhere}
  \end{equation}
\end{proof}

\begin{lemma}
  \label{lemma:avgZH-ZH}
  In an Ockhamic ZH calculus, \Rule(Avg ZH) is sound iff $u_3 = h_1 (2 h_2^2)^{-1}$.
\end{lemma}
\begin{proof}
  This follows from the fact that in an Ockhamic ZH calculus,
  \begin{equation}
  \begin{aligned}[b]
    \Bigsem{\,\input{./Figures-Source/ZH-average-prep.tikz}\,}
  \;=\;
    h_1 \Bigsem{\,\input{./Figures-Source/ZH-average-prep.tikz}\,}_{\!\beta}
  \;=\;
    \frac{h_1}{2} \Biggsem{\,\input{./Figures-Source/ZH-average-loop.tikz}\,}_{\!\beta}
  \;=\;
    \frac{h_1}{2 h_2^2 u_3} \Biggsem{\,\input{./Figures-Source/ZH-disjunct.tikz}\,} .
  \end{aligned}
  \tag*{\qedhere}
  \end{equation}
\end{proof}

\vspace*{-1ex}
\subsubsection{Compatibility and incompatibility results}

Because a few rules of Figure~\ref{fig:candidateZHrewrites} are infinitary (expressing equivalence of infinitely many pairs of diagrams), and the others concern a common set of coefficients (namely $u_1$, $u_3$, $h_1$, and $h_2$), we may easily describe a network of rules and pairs of rules which imply or contradict others.
In the following, we pre-suppose a fixed Ockhamic model in which the coefficients $u_k$, $g_k$, and $h_k$ are defined.

\begin{theorem}
  \label{thm:IdHNotMult-ZH}
  In an Ockhamic ZH calculus, if any two of \Rule(Id H), \Rule(Not), and \Rule(Mult ZH) are sound, then all are sound.
\end{theorem}
\begin{proof}
  This follows from Lemmas~\ref{lemma:idH-ZH}, \ref{lemma:not-ZH}, and~\ref{lemma:multZH-ZH}, in that any two of $2 h_2^2 = 1$, $u_3 = h_1^{-1}$, and $u_3 = h_1^{-1} (2h_2^2)^{-1}$ implies the other.
\end{proof}

\begin{theorem}
  \label{thm:IdHUnitDisj-ZH}
  In an Ockhamic ZH calculus, if any two of \Rule(Id H), \Rule(Unit ZH), and \Rule(Dilem ZH) are sound, then all are sound.
\end{theorem}
\begin{proof}
  This follows from Lemmas~\ref{lemma:idH-ZH}, \ref{lemma:unitZH-ZH}, and~\ref{lemma:disjZH-ZH}, in that any two of $2 h_2^2 = 1$, $u_1 = h_1$, and $u_1 = h_1 (2h_2^2)^{-1}$ implies the other.
\end{proof}

\begin{theorem}
  \label{thm:bialgZGbialgZH}
  In a ZH calculus with an Ockhamic model in which $g_k = \sqrt 2\,h_k$ for all $k \ge 1$, \Rule(Bialg ZG) is sound if and only if \Rule(Bialg ZH) is.
  Furthermore, an Ockhamic ZH calculus in which \Rule(Bialg ZG) and  \Rule(Bialg ZH) are both sound has an Ockhamic model in which $g_k = \sqrt 2\, h_k$ for all $k \ge 1$.
\end{theorem}
\begin{proof}
  If either \Rule(Bialg ZG) or \Rule(Bialg ZH) are sound, then the relation $g_k = \sqrt 2\,h_k$ establishes $g_k = \sqrt 2\,h_k = u_1^{k-2}$, by Lemmas~\ref{lemma:switchZG-ZH} and~\ref{lemma:bialgZH-ZH}.
  This is then equivalent to the soundness of both \Rule(Bialg~ZG) and \Rule(Bialg~ZH).
\end{proof}

\begin{theorem}
  If \Rule(Switch ZG) and \Rule(Bialg ZG) are both sound, then \Rule(Id Z) and \Rule(Id H) are also sound --- but \Rule(Special Z) and \Rule(Orth ZH) are not.
\end{theorem}
\begin{proof}
  By Lemmas~\ref{lemma:switchZG-ZH} and \ref{lemma:bialgZG-ZH}, if \Rule(Switch ZG) and \Rule(Bialg ZG) are sound, then $g_k = u_k^{-1} = 2 h_2^k u_k$ and $u_k = u_1^{-(k{-}2)}$ for all $k \ge 1$.
  In particular, we have $g_2 = u_2^{-1} = 1$, so that \Rule(Id Z) is sound by Lemma~\ref{lemma:idZ-ZH}.
  It also follows that $1 = g_2 = 2 h_2^2 u_2 = 2 h_2^2$, which implies that $h_2 = 2^{-1\!\!\:/\!\!\;2}$, so that \Rule(Id H) is sound by Lemma~\ref{lemma:idH-ZH}.
  We may then infer
  \begin{equation}
    u_1^2 = g_1^{-1} u_1 = (2 h_2 u_1)^{-1} u_1 = 2^{-1\!\!\:/\!\!\;2}, 
  \end{equation}
  so that $u_1 = 2^{-1\!\!\:/\!\!\;4}$; together with Lemmas~\ref{lemma:specialZ-ZH} and~\ref{lemma:orthZH-ZH}, it follows that \Rule(Special Z) and \Rule(Orth ZH) are unsound.
\end{proof}

\begin{theorem}
  \label{thm:BialgZH}
  In an Ockhamic ZH calculus, if \Rule(Bialg ZH) is sound, then \Rule(Id Z) and \Rule(Id H) are sound.
\end{theorem}
\begin{proof}
  Suppose that \Rule(Bialg ZH) is sound.
  By Lemma~\ref{lemma:bialgZG-ZH}, we immediately have $u_2 = 1$ and $h_2 = 2^{-1\!\!\:/\!\!\;2}$, from which it follows that \Rule(Id Z) and \Rule(Id H) are sound by Lemmas~\ref{lemma:idZ-ZH} and~\ref{lemma:idH-ZH}.
\end{proof}

\begin{theorem}
  \label{thm:bialgZH-equivNotMultUnitDisjZH}
  In an Ockhamic ZH calculus, if \Rule(Bialg ZH) is sound, then each of the rules \Rule(Not), \Rule(Mult~ZH), \Rule(Unit ZH), and \Rule(Dilem ZH) are sound if and only if the others are.
\end{theorem}
\begin{proof}
  Suppose that \Rule(Bialg ZH) is sound.
  By Lemma~\ref{lemma:bialgZH-ZH}, we then have $u_k = u_1^{-(k{-}2)}$ and $h_k = 2^{-1\!\!\:/\!\!\;2} u_1^{k{-}2}$, and \Rule(Id H) is sound by Theorem~\ref{thm:BialgZH}.
  It then follows that \Rule(Not) is sound iff \Rule(Mult ZH) is, by Theorem~\ref{thm:IdHNotMult-ZH}; and \Rule(Unit ZH) is sound iff \Rule(Dilem ZH) is, by Theorem~\ref{thm:IdHUnitDisj-ZH}.
  We proceed by cases:
  \begin{itemize}
  \item
    If \Rule(Mult ZH) is sound, we then have $u_1^{-1} = u_3 = h_1^{-1} = 2^{1\!\!\:/\!\!\;2} u_1$ by Lemma~\ref{lemma:multZH-ZH}.
  \item
    If \Rule(Unit ZH) is sound, then $u_1 = h_1 = 2^{-1\!\!\:/\!\!\;2} u_1^{-1}$ by Lemma~\ref{lemma:unitZH-ZH}.
  \end{itemize}
  In both cases above, we have $u_1 = 2^{-1\!\!\:/\!\!\;4}$.
  Conversely, if $u_1 = 2^{-1\!\!\:/\!\!\;4}$, then $h_1 = 2^{-1\!\!\:/\!\!\;2} u_1^{-1} = u_1$ and $u_3 = u_1^{-1} = 2^{1\!\!\:/\!\!\;4}$, from which we may infer that \Rule(Mult ZG) and \Rule(Unit ZH) are both sound.
\end{proof}

\begin{corollary}
  \label{cor:bialgZHNot-noSpecialOrth}
  In an Ockhamic ZH calculus, if \Rule(Bialg ZH) is sound, and any of \Rule(Not), \Rule(Mult~ZH), \Rule(Unit ZH), and \Rule(Dilem ZH) are sound, then \Rule(Special Z) and \Rule(Orth ZH) are not.
\end{corollary}
\begin{proof}
  From the proof of Theorem~\ref{thm:bialgZH-equivNotMultUnitDisjZH}, in these conditions we have $u_1 = 2^{-1\!\!\:/\!\!\;4}$ and $u_3 = 2^{1\!\!\:/\!\!\;4}$, which are inconsistent with both \Rule(Special Z) and \Rule(Orth ZH) by Lemmas~\ref{lemma:specialZ-ZH} and~\ref{lemma:orthZH-ZH}.
\end{proof}

\begin{theorem}
  In an Ockhamic ZH calculus for which \Rule(Bialg ZH) and \Rule(Bialg ZG) are sound, \Rule(Switch~ZG) is also sound if and only if any of the rules \Rule(Not), \Rule(Mult ZH), \Rule(Unit ZH), and \Rule(Dilem~ZH) are.
\end{theorem}
\begin{proof}
  Suppose that \Rule(Bialg ZG) and \Rule(Bialg ZH) are both sound.
  By Theorem~\ref{thm:bialgZGbialgZH}, we then have $g_k = \sqrt 2\,h_k = u_k^{-1} = u_1^{k-2}$ for all $k \ge 1$.
  If \Rule(Switch ZG) is also sound, then $u_1^2 = g_1^{-1} u_1 = (2h_2)^{-1} = 2^{1\!\!\:/\!\!\;2}$ by Lemma~\ref{lemma:switchZG-ZH}.
  By the proof of Theorem~\ref{thm:bialgZH-equivNotMultUnitDisjZH}, this is equivalent to the soundness of each of \Rule(Not), \Rule(Mult~ZH), \Rule(Unit ZH), and \Rule(Dilem~ZH).
  Conversely, if $u_1 = 2^{-1\!\!\:/\!\!\;4}$, then $g_k = u_k^{-1} = u_1^{k{-}2} = 2 u_1^{2k} u_1^{-(k{-}2)} = 2 h_2^k u_k$, so that \Rule(Switch~ZG) is sound.
\end{proof}

\begin{theorem}
  In a ZH calculus with an Ockhamic model in which $u_0 = u_1^2$, if \Rule(Bialg ZH) is sound, then \Rule(Fuse Z) is sound.
\end{theorem}
\begin{proof}
  This follows immediately from Lemmas~\ref{lemma:bialgZH-ZH} and~\ref{lemma:fuseZ-ZH}.
\end{proof}

\begin{theorem}
  In an Ockhamic ZH calculus, \Rule(Fuse Z) is sound only if \Rule(Id Z) is.
\end{theorem}
\begin{proof}
  This follows immediately from Lemmas~\ref{lemma:idZ-ZH} and~\ref{lemma:fuseZ-ZH}.
\end{proof}

\begin{theorem}
  In an Ockhamic ZH calculus, at most two of \Rule(Fuse Z), \Rule(Special Z), and \Rule(Orth Z) are sound.   
\end{theorem}
\begin{proof}
  This follows immediately from Lemmas~\ref{lemma:fuseZ-ZH}, \ref{lemma:specialZ-ZH}, and~\ref{lemma:orthZH-ZH}.
\end{proof}

\begin{theorem}
  In an Ockhamic ZH calculus, if \Rule(Special Z) is sound, then each of \Rule(Id H) or \Rule(Unit ZH) are sound if and only if both \Rule(Not) and \Rule(Avg ZH) are sound.
\end{theorem}
\begin{proof}
  If \Rule(Special Z) is sound, then $u_3 = 1$ by Lemma~\ref{lemma:specialZ-ZH}.
  If either \Rule(Id H) or \Rule(Unit ZH) are sound, then by Lemmas~\ref{lemma:idH-ZH} and~\ref{lemma:unitZH-ZH} we have $h_1 = 1$ which is equivalent to the soundness of both \Rule(Not) and \Rule(Avg ZH) by Lemmas~\ref{lemma:not-ZH} and~\ref{lemma:avgZH-ZH}.
  Conversely, if both \Rule(Not) and \Rule(Avg ZH) are sound, we obtain $h_1 = 1$ and $2h_2^2 = 1$, which implies the soundness of each of
  \Rule(Id H) and \Rule(Unit ZH).  
\end{proof}


\vspace*{-3ex}
\section{Constructing the well-tempered calculi}
\label{apx:constructingWellTemperedCalculi}
\vspace*{-.5ex}

\begin{figure}[!t]
\input{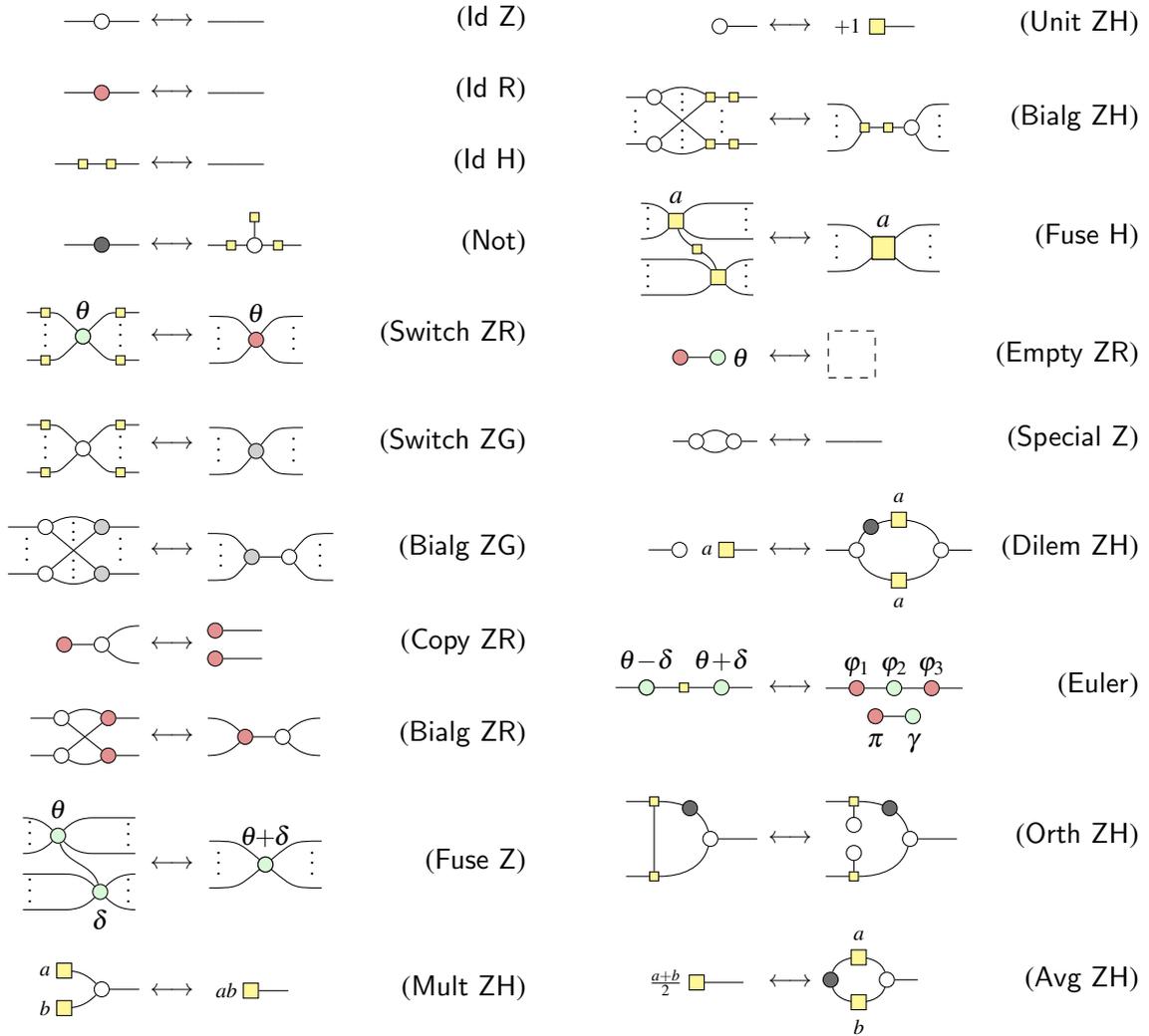}
\vspace*{-3.0ex}
\caption{%
  \label{fig:candidateRewrites}
  Idealised rewrites for ZX and ZH calculi, consisting of the simplest rewrites equivalent (up to a scalar gadget) to the rules of Figures~\ref{fig:usualZXrules} and~\ref{fig:variantZHrules}.
  The parameters $\theta\!, \delta \in \R$ and $a,b \in \C$ may be arbitrary, while $\varphi_1, \varphi_2, \varphi_3, \gamma \in \R$ in~\protect\Rule(Euler) are as described in Eqn.~\eqref{eqn:VilmartEulerAngleComputation} on page~\pageref{eqn:VilmartEulerAngleComputation}.
  As we identify the Z dots of the ZX calculus with the white dots of the ZH calculus, we represent \smash{\protect\Rule(S$_g^\alpha$)} and \smash{\protect\Rule(S$_z^{\smash\beta}$)} by the rule \smash{\protect\Rule(Special Z)}; and \smash{\protect\Rule(F$_{\!z}^{\smash\beta}$)} corresponds to a special case of \smash{\protect\Rule(F$_{\!g}^\alpha$)}, which is represented by \smash{\protect\Rule(Fuse Z)}.
  For ease of reading without colour, we represent every rule involving phase-free Z dots using the corresponding generator from the ZH calculus.
}
\vspace*{-.5ex}
\end{figure}

In this Section, we provide the proofs of the Lemmata and Corollaries in Sections~\ref{sec:constructingNewCalculi} and~\ref{sec:features} which allow us to construct the model $\sem{\,\cdot\,}_\nu$, and to establish the completeness of the calculi on pages~\pageref{newZX} and~\pageref{newZH}.
Our analysis will rely on the characterisations of Appendices~\ref{apx:compatibilityRewritesZX} and~\ref{apx:compatibilityRewritesZH} of the conditions under which idealised ZX and ZH rewrites are sound.

Figure~\ref{fig:candidateRewrites} presents a ``wish-list'' of idealised rewrite rules (which are not all compatible).
These consist of scalar-simplified versions of the ZX and ZH calculus rules from Figures~\ref{fig:usualZXrules} and~\ref{fig:variantZHrules}, listed very roughly in a proposed order of importance in routine calculations.
Our proofs of the soundness of such idealised rewrite rules will often make use of results proven in Appendices~\ref{apx:compatibilityRewritesZX} and~\ref{apx:compatibilityRewritesZH}, but in all cases ultimately rest on reductions to the semantics of diagrams in the pre-existing models $\sem{\,\cdot\,}_{\alpha}$ and $\sem{\,\cdot\,}_{\beta}$ for the generators in the ZX and ZH calculi (as described in Section~\ref{sec:traditionalZXandZH}).

\vspace*{-1ex}
\paragraph{Rules which are sound by denotational constraints --- {}\!\!\!\!}
From the constraints on the denotation of specific unitary operators, which we impose in Section~\ref{sec:denotationalConstraints}, we are already assured of the soundness of several of the idealised rewrites of Figure~\ref{fig:candidateRewrites}.
Following on from Lemma~\ref{lemma:denotationalConstraintsRedux}, we have:

\smallskip
\begin{lemma}[\emph{c.f.}~Corollary~\ref{cor:rewritesFromNotation}]
  \label{lemma:soundFromDenotation}
  If Eqns.~\eqref{eqn:parametersModelNu}--\eqref{eqn:nu-scalar-box} hold, each of~\Rule(Id\:Z), \Rule(Id\:R), \Rule(Id\:H), \Rule(Not), \Rule(Bialg\:ZR), \Rule(Mult\:ZH), and~\Rule(Fuse\:H) are sound.
\end{lemma}

\vspace*{-1.5ex}
\begin{proof}
  This is immediately clear for~\Rule(Id Z), \Rule(Id R), and \Rule(Id H).
  Following from Lemma~\ref{lemma:denotationalConstraintsRedux}, from $u_3 = 2^{1\!\!\:/\!\!\;4}$, $h_1 = 2^{-1\!\!\:/\!\!\;4}$, $h_2 = 2^{-1\!\!\:/\!\!\;2}$, and $h_3 = 2^{-3\!\!\:/\!\!\;4}$ the soundness of \Rule(Bialg ZR) follows from Lemma~\ref{lemma:bialgZR-ZX}  (Appendix~\ref{apx:compatibilityRewritesZX}, page~\pageref{lemma:bialgZR-ZX}) and the supplementary fact that $u_3 = v_3$; the soundness of \Rule(Not), \Rule(Mult~ZH), and~\Rule(Fuse H) follow from Lemmas~\ref{lemma:not-ZH}, \ref{lemma:multZH-ZH}, and \ref{lemma:fuseH-ZH} respectively (Appendix~\ref{apx:compatibilityRewritesZH}, pages~\pageref{lemma:not-ZH} 
  and \pageref{lemma:fuseH-ZH}). 
\end{proof}

\vspace*{-2.5ex}
\paragraph{Colour switch rewrites --- {}\!\!\!\!}
We next impose the soundness of \Rule(Switch~ZR) and \Rule(Switch~ZG) as constraints on $\sem{\,\cdot\,}_\nu$\,, as these rewrites are helpful in reasoning about the interaction between Z dots and H-boxes / Hadamard nodes:
\begin{gather}
  \label{eqn:switchSound}
  \biggsem{\input{./Figures-Source/ZX-red-dot.tikz}}_{\!\nu}
  =
  \biggsem{\input{./Figures-Source/ZH-white-w-H.tikz}}_{\!\nu}
  =
  \biggsem{\input{./Figures-Source/ZH-gray-dot.tikz}}_{\!\nu}
\\[-3.5ex]
\notag
\end{gather}
This effectively identifies the phase-free red nodes with the gray nodes (though the normalisation factors $v_k$ and $g_k$ will differ for $k \ne 2$, as a result of how they are defined).

\smallskip

\begin{lemma}[\emph{c.f.}~Lemma~\ref{lemma:summaryChangeSound}]
  \label{prop:switchSound}
  Eqns.~\eqref{eqn:parametersModelNu} and~\eqref{eqn:switchSound} hold iff the rewrites \Rule(Switch ZR) and \Rule(Switch ZG) are sound, which hold iff $u_k \!\!\:=\!\!\: v_k$ and $g_k \!\!\:=\!\!\: 2h_2^k u_k$ for all $k \ge 0$.
  In particular, Eqns.~\eqref{eqn:parametersModelNu}, \eqref{eqn:unitaryNOTandH}, and~\eqref{eqn:switchSound} hold iff $h_2 \!\!\:=\!\!\: 2^{-1\!\!\:/\!\!\;2}$, $\xi \!\!\:=\!\!\: 1$, and $u_k \!\!\:=\!\!\: v_k \!\!\:=\!\!\: 2^{(k-2)/\!\!\;2} g_k$ for all $k \ge 0$.
\end{lemma}

\vspace*{-1ex}
\begin{proof}
  Eqns.~\eqref{eqn:parametersModelNu} and~\eqref{eqn:switchSound} are  equivalent to the soundness of \Rule(Switch ZR) and \Rule(Switch ZG).
  The relation $u_k = v_k$ holds immediately, while the relation $g_k = 2h_2^k u_k$ then holds for all $k \ge 0$ by Lemma~\ref{lemma:switchZG-ZH} (Appendix~\ref{apx:compatibilityRewritesZH}, page~\pageref{lemma:switchZG-ZH}).
  If furthermore Eqn.~\eqref{eqn:unitaryNOTandH} holds, then $\xi = 1$ and $h_2 = 2^{-1\!\!\:/\!\!\;2}$ (and conversely); the latter equality is then equivalent to $g_k = 2^{-(k{-}2)/\!\!\;2} u_k$.
\end{proof}

\begin{corollary}[trivial]
  If Eqns.~\eqref{eqn:parametersModelNu}--\eqref{eqn:nu-scalar-box} and~\eqref{eqn:switchSound} hold, \Rule(Switch\;ZR) and \Rule(Switch\;ZG) are sound.
\end{corollary}

\paragraph{Bialgebra and fusion rules --- {}\!\!\!\!} We next require soundness of~\Rule(Bialg~ZG), which describes the interaction of gray nodes and Z dots in the ZH~calculus:
\begin{equation}
\label{eqn:bialgZGsound}
    \Biggsem{\input{./Figures-Source/ZH-bialg-white-gray.tikz}}_{\!\nu}
  =
    \biggsem{\input{./Figures-Source/ZH-bott-white-gray.tikz}}_{\!\nu}\,.
\end{equation}~\\[-1ex]
Because of the soundness of both \Rule(Switch~ZR) and~\Rule(Switch~ZG), this would imply the soundness of the rewrite~\Rule(Copy ZR) as a special case.
One may show that Eqn.~\eqref{eqn:bialgZGsound} ``nearly'' implies the soundness of~\Rule(Fuse~Z) as well, as it imposes sufficiently strong constraints on the coefficients $u_k$ for $k > 0$, while the coefficient $u_0$ is independent of the rewrite \Rule(Bialg~ZG).
As~\Rule(Fuse~Z) is extremely useful in its own right to reason about commuting diagonal operations, we also impose the soundness of~\Rule(Fuse~Z) as a constraint to fix the value of $u_0$\,:
\vspace*{-1ex}
\begin{equation}
\label{eqn:fuseZsound}
  \Sem{7.5ex}{\,\smash{\input{./Figures-Source/ZX-green-phase-fuse.tikz}}\,}_{\!\nu}
  =
  \Biggsem{\input{./Figures-Source/ZX-green-phase-sum.tikz}}_{\!\nu}\,.
\end{equation}~\\[-5ex]
\begin{lemma}[\emph{c.f.}~Lemma~\ref{lemma:establishNuModel}]
  \label{lemma:bialgFuseSound}
  Eqns.~\eqref{eqn:parametersModelNu}--\eqref{eqn:nu-scalar-box} and \eqref{eqn:switchSound}--\eqref{eqn:fuseZsound} all hold iff 
  $\xi \!\!\:=\!\!\: 1$, $u_k \!\!\:=\!\!\: v_k \!\!\:=\!\!\: g_k^{-1} \!\!\:=\!\!\:  2^{(k{-}2)/\!\!\;4}$ for all $k \ge 0$, and $h_k = 2^{-k\!\!\:/\!\!\;4}$ for all $k \ge 0$.
\end{lemma}
\begin{proof}
    Eqn.~\eqref{eqn:bialgZGsound} is equivalent to the soundness of \Rule(Bialg ZG) for $\sem{\,\cdot\,}_\nu$, which by Lemma~\ref{lemma:bialgZG-ZH} (Appendix~\ref{apx:compatibilityRewritesZH}, page~\pageref{lemma:bialgZG-ZH}) is equivalent to $g_k = u_k^{-1} = u_1^{k{-}2}$ for all $k \ge 1$.
    Furthermore, in the case $\theta = \delta = 0$, Eqn.~\eqref{eqn:fuseZsound} is equivalent to the soundness of \Rule(Fuse Z) for $\sem{\,\cdot\,}_\nu$, which by Lemma~\ref{lemma:fuseZ-ZX} (Appendix~\ref{apx:compatibilityRewritesZX}, page~\pageref{lemma:fuseZ-ZX}) is equivalent to $u_k^{-1} = u_1^{k{-}2}$ for all $k \ge 0$.    
    Following from Lemmata~\ref{lemma:soundFromDenotation} and~\ref{prop:switchSound}, if Eqns.~\eqref{eqn:parametersModelNu}, \eqref{eqn:unitaryOps}, \eqref{eqn:nu-scalar-box}, and \eqref{eqn:switchSound} also hold, we also have $\xi = 1$, $u_2 = v_2 = 1$, $u_3 = v_3 = g_3^{-1} = 2^{1\!\!\:/\!\!\;4}$, $u_k = v_k = 2^{(k-2)/\!\!\;2} g_k$ for all $k \ge 0$, and $h_k = 2^{-k\!\!\:/\!\!\;4}$ for all $k \ge 0$.
    In particular, we have $h_2 = 2^{-1\!\!\:/\!\!\;2}$ and $u_3 = 2^{1\!\!\:/\!\!\;4}$; it then follows that $u_1 = u_3^{-1} = 2^{-1\!\!\:/\!\!\;4}$ and $u_0 = u_1^2 = 2^{-1\!\!\:/\!\!\;2}$, and that $g_0 = 2 h_2^0 u_0 = 2^{1\!\!\:/\!\!\;2} = u_0^{-1}$.
    Conversely, if $u_k = v_k = g_k^{-1} = 2^{(k{-}2)/\!\!\;4}$ for all $k \ge 0$ and $h_k = 2^{-k\!\!\:/\!\!\;4}$ for all $k \ge 0$, it follows in particular that $u_k^{-1} = u_1^{-(k{-}2)}$ for all $k \ge 0$, that $u_0 = v_0 = 1$, and that $u_3 = v_3 = g_3^{-1} = 2^{1\!\!\:/\!\!\;4}$, and that $g_k = 2 h_2^k u_k$; thus we may infer that Eqns.~\eqref{eqn:parametersModelNu}, \eqref{eqn:unitaryOps}, \eqref{eqn:switchSound}, \eqref{eqn:bialgZGsound}, and Eqn.~\eqref{eqn:fuseZsound} all hold.
\end{proof}

\begin{corollary}
  If Eqns.~\eqref{eqn:parametersModelNu}--\eqref{eqn:nu-scalar-box} and \eqref{eqn:switchSound}--\eqref{eqn:fuseZsound} hold, each of \Rule(Copy\:ZR), \Rule(Unit\:ZH), \Rule(Bialg\:ZH), \Rule(Empty\:ZR), \Rule(Dilem\:ZH), and \Rule(Euler) are sound.
\end{corollary}
\begin{proof}
  The soundness of \Rule(Copy ZR), \Rule(Unit ZH), \Rule(Bialg~ZH), and \Rule(Dilem ZH) follow from Lemmas~\ref{lemma:copyZR-ZX}, \ref{lemma:unitZH-ZH}, \ref{lemma:bialgZH-ZH}, and~\ref{lemma:disjZH-ZH} respectively (on pages~\pageref{lemma:copyZR-ZX}, \pageref{lemma:unitZH-ZH}, and~\pageref{lemma:disjZH-ZH}).
  For \Rule(Empty ZR) we directly compute
  \begin{equation}
    \Bigsem{\input{./Figures-Source/ZX-simpler-1.tikz}\!}_\nu 
  =\,
    2 a_1^2 \Bigl(\!\!\;\bra{\zsymb} + \e^{i\theta}\! \bra{\osymb}\!\!\;\Bigr)\!\!\;
    \ket{\zsymb}
  \,=\,
    1
  \,=\,
    \bigsem{\quad}_{\nu}\,.
  \end{equation}
  For \Rule(Euler), we show that the two closed two-node gadgets in  Rule~\Rule(EU$^\alpha$) in Figure~\ref{fig:usualZXrules} (page~\pageref{fig:usualZXrules}) express the same global phase factor as the single closed two-node gadget of \Rule(Euler):
  \begin{equation}
  \begin{aligned}[b]
    \Bigsem{\input{./Figures-Source/ZX-scalar-phase-old.tikz}\,}_{\alpha} \!\!
  \;&=\;
    \Bigl( \bra{\zsymb\zsymb\zsymb} + \bra{\osymb\osymb\osymb} \Bigr) \Bigl( \ket{\psymb\psymb\psymb} + \ket{\msymb\msymb\msymb} \Bigr) 
    \cdot
    \Bigl(\!\!\;\bra{\zsymb} + \e^{i\gamma}\! \bra{\osymb}\!\!\;\Bigr)
    \Bigl( \sqrt 2 \ket{\osymb} \Bigr)
    \mspace{-60mu}
  \\&=\;
    \bigl( \tfrac{1}{2\sqrt 2} + \tfrac{1}{2\sqrt 2}  + \tfrac{1}{2\sqrt 2}  - \tfrac{1}{2\sqrt 2} ) \cdot \bigl( \sqrt 2 \!\: \e^{i\gamma} \bigr)
  \\&=\;
    \e^{i\gamma}
  \;=\;
    \Bigl[ u_1 \Bigl(\!\!\;\bra{\zsymb} + \e^{i\gamma}\! \bra{\osymb}\!\!\;\Bigr) \Bigr] \Bigl[ \sqrt 2 u_1 \ket{\osymb} \Bigr]
  =
    \Bigsem{\input{./Figures-Source/ZX-scalar-phase-new.tikz}}_\nu \;,
  \end{aligned}
  \end{equation}
  so that the soundness of \Rule(Euler) for the model $\sem{\,\cdot\,}_\nu$ follows from the soundness of \Rule(EU$^\alpha$) for $\sem{\,\cdot\,}_\alpha$\,.
\end{proof}

\vspace*{-2ex}
\paragraph{Scalars and unsound idealised rules --- {}\!\!\!\!}
Imposing the constraints that we have on $\sem{\,\cdot\,}_\nu$ not only selects rewrites from Figure~\ref{fig:candidateRewrites} which are sound, but also implicitly selects rewrites which are not sound:

\smallskip

\begin{proposition}
  \label{prop:idealisedUnsound}
  If Eqns.~\eqref{eqn:parametersModelNu}--\eqref{eqn:nu-scalar-box} and \eqref{eqn:switchSound}--\eqref{eqn:fuseZsound} hold, \Rule(Special\:Z), \Rule(Orth\:ZH), and~\Rule(Avg\:ZH) are \underline{\upshape unsound}.
\end{proposition}

\vspace*{-1ex}
\begin{proof}
  Given $u_1 = 2^{-1\!\!\:/\!\!\;4}$, $u_3 = 2^{1\!\!\:/\!\!\;4}$, $h_1 = 2^{-1\!\!\:/\!\!\;4}$, and $h_2 = 2^{-1\!\!\:/\!\!\;2}$, the unsoundness of \Rule(Special Z) and \Rule(Avg ZH) follows from Lemmata~\ref{lemma:specialZ-ZX}, \ref{lemma:orthZH-ZH}, and~\ref{lemma:avgZH-ZH} (on pages~\pageref{lemma:specialZ-ZX} 
  and~\pageref{lemma:avgZH-ZH}).
  Specifically, from the proof of Lemma~\ref{lemma:specialZ-ZX} (page~\pageref{lemma:specialZ-ZX}), we instead have
  \begin{equation}
    \label{eqn:newSpecialScaling}
      \bigsem{\input{./Figures-Source/ZH-white-special.tikz}}_{\nu}
    =\,
      u_3^2 \, \bigsem{\;\input{./Figures-Source/id-wire.tikz}\;}_{\nu}
    =\,
      \sqrt{2} \, \bigsem{\;\input{./Figures-Source/id-wire.tikz}\;}_{\nu}\,;
  \end{equation}
  from the proof of Lemma~\ref{lemma:orthZH-ZH} (page~\pageref{lemma:orthZH-ZH}), we have
  \begin{equation}
  \label{eqn:newOrthoScaling}
  \begin{aligned}[b]
  \Sem{6ex}{\,\input{./Figures-Source/ZH-ortho-bridge.tikz}\,}_{\!\nu} \!\!
  &=\,
  \frac{1}{2 u_1^2}
  \Sem{6ex}{\,\input{./Figures-Source/ZH-ortho-nobridge.tikz}\,}_{\!\nu} \!\!
  =\,
  \frac{1}{\sqrt 2}
  \Sem{6ex}{\,\input{./Figures-Source/ZH-ortho-nobridge.tikz}\,}_{\!\nu} \!\!
  \,,
  \end{aligned}
  \end{equation}%
  and from the proof of Lemma~\ref{lemma:avgZH-ZH} (page~\pageref{lemma:avgZH-ZH}), we have
  \begin{equation}
  \label{eqn:newAvgScaling}
      \Bigsem{\input{./Figures-Source/ZH-average-prep.tikz}}_{\nu}
    =\;
      \frac{h_1}{2 h_2^2 u_3} \, \Biggsem{\;\input{./Figures-Source/ZH-average-loop.tikz}\;}_{\nu}
    \!=\;
      \frac{1}{\sqrt 2} \, \Biggsem{\;\input{./Figures-Source/ZH-average-loop.tikz}\;}_{\nu}\,;
  \end{equation}
  so that none of these rewrites are sound without additional scalar factors.
\end{proof}
\noindent
We may easily identify variants of the rewrites \Rule(Special\;Z), \Rule(Orth\;ZH), and \Rule(Avg\;ZH) which \emph{are} sound, by describing how to represent the appropriate scalars.
In ZH diagrams, the scalars may (by construction) be represented very directly, using Eqn.~\eqref{eqn:nu-scalar-box}.
We digress momentarily to describe how scalars may be represented in the model $\sem{\,\cdot\,}_\nu$ using ZX generators:

\medskip

\begin{lemma}[\emph{c.f.}~Lemma~\ref{lemma:greenDotScalars}]
  \label{lemma:nuZXscalardot}
  If Eqns.~\eqref{eqn:parametersModelNu}--\eqref{eqn:nu-scalar-box} and \eqref{eqn:switchSound}--\eqref{eqn:fuseZsound} hold, then for $\theta \in \R$ we have
  \vspace*{-0.5ex}
  \begin{equation}
  \label{eqn:nuGreenScalars}
    \bigsem{\,\input{./Figures-Source/ZX-green-phase-dot-arity-0-small.tikz}\,}_\nu
    =\;
    \Bigsem{\,\input{./Figures-Source/ZH-H-green-phase.tikz}}_\nu
    \;=\;
    \sqrt{1+\cos(\theta)}\,\e^{i \!\; \theta\!\!\!\:/2}.
  \end{equation}~\\[-5ex]%
  In particular,
  \begin{align}
  \label{eqn:nuScalarSpecialCases}
    \bigsem{\,\input{./Figures-Source/ZX-green-dot-arity-0.tikz}\,}_\nu
    =\,
    \bigsem{\,\input{./Figures-Source/ZH-H-sqrt2-scalar.tikz}\,}_\nu
    &=\,
    \sqrt 2\,,
&\quad
    \bigsem{\,\input{./Figures-Source/ZX-green-pi2-phase-dot-arity-0.tikz}\,}_\nu
    =\,
    \bigsem{\input{./Figures-Source/ZH-H-sqrti-scalar.tikz}}_\nu
    &=\,
    \sqrt{i\:},
&\quad
    \bigsem{\,\input{./Figures-Source/ZX-green-pi-phase-dot-arity-0.tikz}\,}_\nu
    =\,
    \bigsem{\input{./Figures-Source/ZH-H-zero-scalar.tikz}}_\nu
    &=\,
    0.
  \end{align}~\\[-6ex]
\end{lemma}
\begin{proof}
  Following from Lemma~\ref{lemma:bialgFuseSound}, we have $u_0 = 2^{-1\!\!\:/\!\!\;2}$.
  From Eqn.~\eqref{eqn:parametersModelNu}, we then have
  \begin{equation}
    \bigsem{\,\input{./Figures-Source/ZX-green-phase-dot-arity-0-small.tikz}\,}_\nu
    =\;
    u_0\bigl[1 + \cos(\theta) + i \sin(\theta)\bigr]
    \;=\;
      \tfrac{1+\cos(\theta)}{\sqrt 2}
      \,+\, i\;\! \tfrac{\sin(\theta)}{\sqrt 2}
    \;=:\;
      z_\theta\,.
  \end{equation}
  It is then easy to show that $
      \lvert z_\theta \rvert
    =
      \sqrt{1 + \cos(\theta)}    
    =
      \sqrt{2\cos^2(\theta\!\!\:/2)}    
$\,,
  and that 
  \begin{equation}{}
  \mspace{-18mu}
  \begin{aligned}[b]
    \lvert z_\theta \rvert \, \e^{i\!\;\theta\!\!\:/2}
  \;=\;
    \Bigl(\sqrt 2 \cos(\theta\!\!\:/2)\Bigr)
    \Bigl(\cos(\theta\!\!\:/2) + i \sin(\theta\!\!\:/2)\Bigr)
  \;&=\;
    \tfrac{1}{\sqrt 2}
    \Bigl(2\cos^2(\theta\!\!\:/2) + 2i \cos(\theta\!\!\:/2)\sin(\theta\!\!\:/2)\Bigr)
  \\&=\;
    \tfrac{1}{\sqrt 2}
    \Bigl([1 + \cos(\theta)] + i \sin(\theta) \Bigr)
  \;=\;
    z_\theta.
  \end{aligned}
  \mspace{-18mu}
  \end{equation}
  Then Eqn.~\eqref{eqn:nuGreenScalars} holds; Eqns.~\eqref{eqn:nuScalarSpecialCases} follow as easy corollaries.
\end{proof}

\begin{corollary}[%
  from Eqns.~\eqref{eqn:newSpecialScaling}--\eqref{eqn:newAvgScaling}
  and Lemma~\ref{lemma:nuZXscalardot}%
  ]
  If Eqns.~\eqref{eqn:parametersModelNu}--\eqref{eqn:nu-scalar-box} and \eqref{eqn:switchSound}--\eqref{eqn:fuseZsound} hold, then
  \begin{align}
    \Bigsem{\input{./Figures-Source/ZX-green-special.tikz}}_\nu
    \,=\;
    \Bigsem{\input{./Figures-Source/ZH-white-special.tikz}}_\nu
    \,&=\;
    \biggsem{\input{./Figures-Source/ZH-id-wire-w-sqrt2.tikz}}_\nu
    \,=\;
    \biggsem{\input{./Figures-Source/ZX-id-wire-w-sqrt2.tikz}}_\nu  \;;
  \\[2ex]
    \Sem{7ex}{\,\input{./Figures-Source/ZH-ortho-bridge-w-sqrt2.tikz}\,}_{\!\nu} \!\!
  &=\,
    \Sem{5ex}{\,\input{./Figures-Source/ZH-ortho-nobridge.tikz}\,}_{\!\nu} \;;
  \\[2ex]
    \Sem{6ex}{\,\input{./Figures-Source/ZH-average-prep-w-sqrt2.tikz}\,}_{\!\nu} \!\!
  &=\,
    \Sem{5ex}{\,\input{./Figures-Source/ZH-average-loop.tikz}\,}_{\!\nu} \;.
  \end{align}
\end{corollary}
\begin{proof}
This follows directly from Eqns.~\eqref{eqn:newSpecialScaling}--\eqref{eqn:newAvgScaling} and Lemma~\ref{lemma:nuZXscalardot}.
\end{proof}

\noindent
It remains to define the meaning of the nu-boxes for the ZX calculus, which we do by setting
\begin{equation}
  \label{eqn:defineNuBoxes}
  \sem{\;\input{./Figures-Source/ZX-nu-box-k.tikz}\!}_\nu = \nu^k  \;.
\end{equation}

\begin{corollary}
  If Eqns.~\eqref{eqn:parametersModelNu}--\eqref{eqn:nu-scalar-box}, \eqref{eqn:switchSound}--\eqref{eqn:fuseZsound}, and \eqref{eqn:defineNuBoxes} hold, then all of the rewrites of pages~\pageref{newZX} and~\pageref{newZH} are sound.
\end{corollary}
\begin{proof}
  This follows from Lemma~\ref{lemma:nuZXscalardot}, and from the correspondence of the rewrites of pages~\pageref{newZX} and~\pageref{newZH} with those of Figure~\ref{fig:candidateRewrites}, noting in particular that \Rule(Bialg) on page~\pageref{newZX} is equivalent to \Rule(Bialg~ZR).
\end{proof}

\vspace*{-2ex}
\paragraph{The model $\sem{\,\cdot\,}_\nu$ described explicitly --- {}\!\!\!\!}
For the sake of completeness, we now describe the model $\sem{\,\cdot\,}_\nu$\,, whose parameters we have fixed in this Appendix.
Following from the conditions of Lemma~\ref{lemma:bialgFuseSound} and Eqn.~\eqref{eqn:defineNuBoxes}, we define the model $\sem{\,\cdot\,}_\nu$ as follows, defining the scalar $\nu = 2^{-1\!\!\:/\!\!\;4}$\,:
\vspace*{-0.5ex}
\begin{small}%
\begin{subequations}%
\label{eqn:modelNu}
\begin{align}{}
  \Biggsem{\!\!\!\!\!\input{./Figures-Source/ZH-white-dot-arity.tikz}\!\!\!\!}_{\!\nu}
  \!=\;
  \Biggsem{\!\!\!\!\!\input{./Figures-Source/ZX-green-dot-arity.tikz}\!\!\!\!\!}_{\!\nu}
  \!\!&=\,
    \nu^{-(m{+}n-2)} \Bigl(
    \ket{\zsymb}^{\!\otimes n}\!\bra{\zsymb}^{\!\!\;\otimes m}
    \,+\;
    \ket{\osymb}^{\!\otimes n}\!\bra{\osymb}^{\!\!\;\otimes m}
    \Bigr);
%
%
\\[1ex]
%
%
  \text{--- more generally,}
  \quad
  \Biggsem{\!\!\!\!\!\input{./Figures-Source/ZX-green-phase-dot-arity.tikz}\!\!\!\!\!}_{\!\nu}
  \!&=\,
    \nu^{-(m{+}n-2)} \Bigl(
    \ket{\zsymb}^{\!\otimes n}\!\bra{\zsymb}^{\!\!\;\otimes m}
    \,+\;
    \e^{i\theta}
    \ket{\osymb}^{\!\otimes n}\!\bra{\osymb}^{\!\!\;\otimes m}
    \Bigr);
%
%
\\[3ex]
%
%
  \Biggsem{\!\!\!\!\!\input{./Figures-Source/ZH-gray-dot-arity.tikz}\!\!\!\!}_{\!\nu}
  \!=\;
  \Biggsem{\!\!\!\!\!\input{./Figures-Source/ZX-red-dot-arity.tikz}\!\!\!\!}_{\!\nu}
  \!\!&=\,
    \nu^{(m{+}n-2)} \; \mathop{\;\sum \;\sum\;}_{%
      \mathclap{\substack{
        {x \in \{\zsymb,\osymb\}^m \!,}
        \,\;
        {y \in \{\zsymb,\osymb\}^n} \\
      w(x) \!\!\;+\!\!\; w(y) \,\in\, 2\Z
      }}}
      \;\;
      \ket{y}\!\!\bra{x} \!\!\:,
%
%
\\[-.5ex]
%
%
  \text{--- more generally,}
  \quad
  \Biggsem{\!\!\!\!\!\input{./Figures-Source/ZX-red-phase-dot-arity.tikz}\!\!\!\!\!}_{\!\nu}
  \!&=\,
    \nu^{-(m{+}n-2)} \Bigl(
    \ket{\psymb}^{\!\otimes n}\!\bra{\psymb}^{\!\!\;\otimes m}
    \,+\;
    \e^{i\theta}
    \ket{\msymb}^{\!\otimes n}\!\bra{\msymb}^{\!\!\;\otimes m}
    \Bigr);
\end{align}~\\[-4.5ex]
\begin{gather}
  \mspace{-30mu}
  \Biggsem{\!\!\!\!\input{./Figures-Source/ZH-H-phase-box-arity.tikz}\!\!\!}_{\!\nu}
  \!\!\!\!\:=\,
    \nu^{(m{+}n)} \;\!
    \mathop{\;\sum \;\sum\;}_{%
      \mathclap{
        {x \in \{\zsymb,\osymb\}^m \!,}
        \,\;
        {y \in \{\zsymb,\osymb\}^n}
      }}
      \,
      a^{x_1 \cdots x_m y_1 \cdots y_n} \!
      \ket{y}\!\!\bra{x}, 
%
%
\mspace{30mu}
%
%
  \Bigsem{\input{./Figures-Source/ZH-not-dot.tikz}}_\nu
  \!\!\!\:=\,
  \text{\footnotesize$\begin{bmatrix}
    0 \!\!&\!\! 1 \\[-0.25ex] 1 \!\!&\!\! 0
  \end{bmatrix}$},
%
%
\mspace{36mu}
%
%
  \Bigsem{\input{./Figures-Source/ZX-nu-box-k.tikz}\!}_\nu
  \!\!\!\:=\,
    \nu^k
  \!\!\;.\;
%
%
\mspace{-30mu}
\end{gather}%
\end{subequations}~\\[-1.5ex]
\end{small}%


\end{document}